\documentclass[english,11pt,letter]{article}

\usepackage[pdftex]{hyperref}
\usepackage[utf8]{inputenc}
\usepackage[OT4]{fontenc}
\usepackage{array}
\usepackage{amsthm}
\usepackage{amsmath}
\usepackage{amssymb}
\usepackage{enumerate}
\usepackage{graphicx}
\usepackage{comment}
\usepackage[english]{babel}
\usepackage{fullpage}
\usepackage{algorithm}
\usepackage[noend]{algpseudocode}
\usepackage{todonotes}
\usepackage{caption}
\usepackage{subcaption}
\usepackage{booktabs}
\usepackage{tocloft}

\newenvironment{thm}{\begin{theorem}}{\end{theorem}}
\newcommand{\cqed}{\renewcommand{\qedsymbol}{$\lrcorner$}}
\newcommand{\maybeqed}{}

\newtheorem{definition}{Definition}[section]
\newtheorem{lemma}[definition]{Lemma}
\newtheorem{corollary}[definition]{Corollary}
\newtheorem{claim}[definition]{Claim}
\newtheorem{theorem}[definition]{Theorem}

\newcommand{\cluster}{cluster}
\newcommand{\portals}{portals}
\newcommand{\pieces}{pieces}
\newcommand{\piece}{\ensuremath{\mathfrak{r}}}
\newcommand{\pornum}{\ensuremath{p_\text{num}}} 
\newcommand{\portot}{\ensuremath{p_\text{tot}}} 
\newcommand{\pormax}{\ensuremath{p_\text{max}}} 
\newcommand{\pienum}{\ensuremath{r_\text{num}}} 
\newcommand{\pietot}{\ensuremath{r_\text{tot}}} 
\newcommand{\piemax}{\ensuremath{r^{\#}_\text{max}}} 
\newcommand{\piemaxvass}{\ensuremath{r_\text{max}}}  
\newcommand{\acttot}{\ensuremath{a_\text{tot}}} 
\newcommand{\porset}{\ensuremath{\mathcal{P}}}
\newcommand{\pieset}{\ensuremath{\mathcal{R}}}

\newcommand{\dist}{\delta}
\newcommand{\faladist}{\widetilde{\delta}}
\newcommand{\emuclo}[1]{\widetilde{#1}}

\newcommand{\mst}{\ensuremath{\mathit{MST}}}
\newcommand{\mstx}[1]{\ensuremath{\mathit{MST}\left(#1\right)}}
\newcommand{\st}{\ensuremath{\mathit{ST}}}
\newcommand{\tnew}{t}

\newcommand{\degree}[2]{\mathit{deg}_#1(#2)}
\newcommand{\mclo}[1]{\overline{#1}}

\newcommand{\nei}[1]{\Gamma(#1)}

\newcommand{\drzewo}{T}
\newcommand{\dlug}{d}
\newcommand{\GD}{{GD}}
\newcommand{\indu}[2]{#1[#2]}
\newcommand{\DO}{\mathbb{D}}

\newcommand{\bigo}{O}
\newcommand{\eps}{\varepsilon}
\newcommand{\GDapx}{\mu}
\newcommand{\effeps}{\tau}
\newcommand{\apxfactor}{\alpha}
\newcommand{\degth}{\eta}
\newcommand{\stepeps}{\varsigma}

\newcommand{\strecz}{D}

\newcommand{\findrep}{FindReplacements}
\newcommand{\findreptxt}{\textsc{\findrep}}
\newcommand{\remove}{\textsc{remove}_\degth}
\newcommand{\poziom}{\texttt{lvl}}
\newcommand{\degeps}{\eps}
\newcommand{\ccomp}{\mathcal{\drzewo}}

\newcommand{\operdistance}{\mathtt{distance}}
\newcommand{\opernearest}{\mathtt{nearest}}
\newcommand{\operactivate}{\mathtt{activate}}
\newcommand{\opermerge}{\mathtt{merge}}
\newcommand{\operdeactivate}{\mathtt{deactivate}}
\newcommand{\opersplit}{\mathtt{split}}
\newcommand{\operquery}{\mathtt{query}}
\newcommand{\appref}[1]{}
\newcommand{\extabstract}[1]{}
\newcommand{\fullversion}[1]{#1}

\begin{document}
\thispagestyle{empty}

\title{The Power of Dynamic Distance Oracles:\\Efficient Dynamic Algorithms for the Steiner Tree}

\author{Jakub Łącki\footnote{University of Warsaw, \texttt{j.lacki@mimuw.edu.pl}. Jakub Łącki is a recipient of the Google Europe Fellowship in Graph Algorithms, and this research is supported in part by this Google Fellowship.}
\\
\and
Jakub O\'{c}wieja\footnote{University of Warsaw,
\texttt{j.ocwieja@mimuw.edu.pl}. Partially supported by ERC grant PAAl no.~259515 and polish
funds for years 2011-2014 for co-financed international projects.}
\and
Marcin Pilipczuk\footnote{University of Bergen, \texttt{malcin@mimuw.edu.pl}. The research leading to these results has received funding from the European Research Council under the European Union's Seventh Framework Programme (FP/2007-2013) / ERC Grant Agreement n. 267959. Partially supported by ERC grant PAAl no.~259515.} \\
\and
Piotr Sankowski\footnote{University of Warsaw,
\texttt{sank@mimuw.edu.pl}. Partially supported by ERC grant PAAl no.~259515, NCN grant N206 567940 and the Foundation for Polish Science and polish
funds for years 2011-2014 for co-financed international projects..}  \\
\and
Anna Zych\footnote{University of Warsaw,
\texttt{anka@mimuw.edu.pl}. Partially supported by ERC grant PAAl no.~259515 and polish
funds for years 2011-2014 for co-financed international projects.} }
\date{}

\maketitle

\begin{abstract}
    In this paper we study the Steiner tree problem over a dynamic set of terminals.
    We consider the model where we are given an $n$-vertex graph $G=(V,E,w)$ with positive
    real edge weights, and our goal is to maintain a tree which is a good approximation of the minimum Steiner tree spanning a terminal set $S \subseteq V$, which changes over time. The changes applied to the terminal
    set are either terminal additions (\emph{incremental} scenario), terminal removals
    (\emph{decremental} scenario), or both ({\em fully dynamic} scenario). Our task here is
    twofold. We want to support updates in sublinear $o(n)$ time, and keep the approximation factor of the algorithm as small as possible.

    We show that we can maintain a $(6+\eps)$-approximate Steiner tree of a general graph in $\tilde{O}(\sqrt{n} \log D)$ time per terminal addition or removal.
    Here, $\strecz$ denotes the stretch of the metric induced by $G$.
    For planar graphs we achieve the same running time and the approximation ratio of $(2+\eps)$.
    Moreover, we show faster algorithms for incremental and decremental scenarios.
    Finally, we show that if we allow higher approximation ratio, even more efficient algorithms are possible.
    In particular we show a polylogarithmic time $(4+\eps)$-approximate algorithm for planar graphs.

    One of the main building blocks of our algorithms are dynamic distance oracles for vertex-labeled graphs, which are of independent interest.
    We also improve and use the online algorithms for the Steiner tree problem.

\end{abstract}

\newpage
\tableofcontents
\newpage

\section{Introduction}
Imagine a network and a set of users that want to maintain a cheap multicast tree in this network during a conference call~\cite{multicast}. The users can join and leave, but in the considered time scale the network remains static.
In other words we are considering the following problem.
We are given a graph $G=(V,E,w)$ with positive edge weights $w:E \to \mathbb{R}^+$.
The goal is to maintain information about approximate Steiner tree in $G$ for a dynamically changing set $S \subseteq V$ of terminals.

This problem was first introduced in the pioneering paper by Imase and Waxman~\cite{ImaseW91} and its study was later continued in~\cite{Wiese,GuGK13,GuK14}. However, all these papers focus on minimizing the number of changes to the tree that are necessary to maintain a good approximation, and ignore the problem of efficiently finding these changes.
The problem of maintaining a Steiner tree is also one of the important open problems in the network community~\cite{industry}, and while it has been studied for many years, the research resulted only in several heuristic approaches~\cite{Bauer95aries,644569,662961,Raghavan99arearrangeable} none of which has been formally proved to be efficient. In this paper we show the first sublinear time algorithm which maintains an approximate Steiner tree under terminal additions and deletions.


Our paper deals with two variants of the problem of maintaining the Steiner tree. 
Throughout this paper, we assume that in the \emph{online Steiner tree} problem (we also say \emph{online setting}) the goal is to maintain a Steiner tree making few changes to the tree after each terminal addition or removal.
On the other hand, in the \emph{dynamic Steiner tree} problem (\emph{dynamic setting}), the requirement is that each update is processed faster than by recomputing the tree from scratch.
This aligns with the usual definition of a dynamic algorithm used in the algorithmic community~(see e.g. \cite{demetrescu2010}). In this paper we study both settings and use the techniques for the online setting to show new results in the dynamic setting.

As our point of reference we observe that it is possible to construct $\tilde{O}(n)$ time dynamic
algorithm for Steiner tree using the dynamic polylogarithmic time minimum spanning forest (dynamic MSF) algorithm~\cite{DBLP:journals/jacm/HolmLT01}.
This solution is obtained by first computing the metric closure $\mclo{G}$ of the graph $G$, and then
maintaining the minimum spanning tree (MST) over the set of terminals $S$ in $\mclo{G}[S]$.
It is a well-known fact that this yields a $2$-approximate Steiner tree.
In order to update $\mclo{G}[S]$ we need to insert and remove terminals together with their incident edges, which requires $\Theta(n)$ calls to the dynamic MSF structure.
However, such a linear bound is far from being satisfactory as does not lead to any improvement in the running time for sparse networks where $m=O(n)$.\footnote{It is widely observed that most real-world networks are sparse~\cite{sparsegraphs}.}
In such networks after each update we can actually compute the $2$-approximate Steiner tree in $O(n \log n)$ time from scratch~\cite{kurt:fast-mst}.
Hence, the main challenge is to break the linear time barrier for maintaining constant approximate
tree. Only algorithms with such sublinear complexity could be of some practical importance and could potentially
be used to reduce the communication cost of dynamic multicast trees. Our paper aims to be a theoretical proof of
concept that from algorithmic complexity perspective this is indeed possible.

As observed by~\cite{Chan:2002:DSC:509907.509911} and~\cite{patrascu11subconn}, the dynamic
problems with vertex updates are much more challenging, but are actually closer to real-world network models than
problems with edge updates. In computer networks, vertex updates happen more often, as they correspond to software
events (server reboot, misconfiguration, or simply activation of a user), whereas edge updates are much less likely, as they correspond to
physical events (cable cut/repair).

Finally, we note that the Steiner tree problem is one of the most fundamental problems in combinatorial optimization.
It has been studied in many different settings, starting from classical approximation algorithms~\cite{BP89,Mitchell96guillotinesubdivisions,DBLP:journals/jacm/Arora98,Robins05,DBLP:journals/talg/BorradaileKM09,Byrka13}, through online~\cite{ImaseW91,Wiese,GuGK13,GuK14} and stochastic models~\cite{GuptaPRS04,GargGLS08}, ending
with game theoretic approaches~\cite{Anshelevich04theprice,pos}. Taking into account the significance and the wide interest in this problem, it is somewhat surprising that no efficient dynamic algorithms for this problem have been developed so far.
It might be related to the fact that this would require combining ideas
from the area of approximation algorithms with the tools specific to dynamic algorithms. This is the first paper that manages to do so. 

In this extended abstract we only present the main ideas, skipping the proofs and the formal analysis, which will be presented in the full version of this paper.

\subsection{Our results}
\label{sec:ourresults}
\begin{table*}
\begin{center}
\bgroup
\def\arraystretch{1.3}
\begin{tabular}{|c|c|c|c|c|}
\hline
{\bf Setting} & {\bf apx.} & {\bf update time} & {\bf preprocessing time} \\
\hline
general, fully dynamic & $6+\eps$ & $\tilde{O}(\sqrt{n}\log D)$ & $\tilde{O}(\sqrt{n}(m + n \log D))$ \\
\hline
general, incremental & $6+\eps$ & $\tilde{O}(\sqrt{n})$ & $\tilde{O}(m\sqrt{n})$ \\
\hline
general, decremental & $6+\eps$ & $\tilde{O}(\sqrt{n})$ & $\tilde{O}(m\sqrt{n})$ \\
\hline
planar, fully dynamic & $2+\eps$ & $\tilde{O}(\sqrt{n}\log D)$ & $\tilde{O}(n \log D)$  \\
\hline
planar, incremental & $2+\eps$ & $\tilde{O}(\log^3 n \log D)$ & $\tilde{O}(n \log D)$ \\
\hline
\end{tabular}
\egroup
\vspace{-0.5cm}
\end{center}
\caption{\label{tabelka}Summary of our algorithms for maintaining Steiner tree in general and planar graphs. For an input graph $G$, $n=|V(G)|$, $m=|E(G)|$, and $D$ is the stretch of the metric induced by $G$. The dependence of the running times on $\eps^{-1}$ is polynomial. The update times are amortized, some of them are also in expectation\extabstract{.}
\fullversion{(see the corresponding theorem statements).}
}
\end{table*}
The main result of this paper are sublinear time algorithms for maintaining an approximate Steiner tree.
We provide different algorithms for incremental, decremental and fully dynamic scenarios.  An \emph{incremental} algorithm allows only to add terminals, a \emph{decremental} algorithm supports removing terminals, whereas a final {\em fully dynamic} algorithm supports both these operations.
Our results are summarized in Table~\ref{tabelka}. The overall approximation ratio of the algorithms we obtain is $(6+\eps)$ for general graphs and $(2+\eps)$ for planar graphs.
In particular, we can maintain a fully dynamic $(6+\eps)$-approximate tree in $\tilde{O}(\sqrt{n}\log D)$ amortized time per update in an arbitrary weighted graph,  where $\strecz$ is the stretch of the metric induced by the graph.
This extended abstract aims to present a brief overview of this result.
The result is a composition and a consequence of many ideas that are of independent interest. We believe that the strength of this paper lies not only in the algorithms we propose, but also in the byproducts of our construction. We outline these additional results below.

\paragraph{Dynamic vertex-color distance oracles}
The algorithms for online Steiner tree assume that the entire metric (i.e., the distances between any pair of vertices) is given explicitly. This assumption is not feasible in case of a metric induced by a graph.
Hence, to obtain the necessary distances efficiently, we develop a data structure called \emph{dynamic vertex-color distance oracle}.
This oracle is given a weighted undirected graph and maintains an assignment of colors to vertices.
While the graph remains fixed, the colors may change.
The oracle can answer, among other queries, what is the nearest vertex of color $c$ to a vertex $v$.
We develop two variants of approximate vertex-color distance oracles: incremental and fully dynamic.
In the first variant, each vertex is initially given a distinct color and the color sets (i.e., two sets representing vertices of the same color) can be merged.
The fully dynamic oracles additionally support (restricted) operations of splitting color sets.\footnote{Because of these two operations we believe that it is more natural to change slightly the previously used vocabulary and assign \emph{colors} instead of \emph{labels} to vertices.}
Note that these update operations are much more general than the operation of changing the color of a single vertex, which was considered in earlier works \cite{vertextolabel,chechik}.

For planar graphs we propose two $(1+\eps)$-approximate oracles.
The incremental oracle supports all operations in $O(\eps^{-1} \log^2 n \log D)$ amortized time (in expectation), whereas the fully dynamic oracle supports operations in worst case time $\bigo(\eps^{-1}\sqrt{n} \log^{2} n \log D)$.
For general graphs we introduce a $3$-approximate fully dynamic oracle that works in $O(\sqrt{n} \log n)$ expected time. Our construction of oracles is generic, that is we show how to extend oracles that may answer vertex-to-vertex queries and satisfy certain conditions into dynamic vertex-color oracles.
For that we introduce the concept of a \emph{generic distance oracle}, which captures the common properties of many distance oracles and allows us to use a uniform approach for different oracles for planar and general graphs.




\paragraph{Online Steiner tree}
We show an online algorithm that decrementally maintains $(2+\eps)$-approximate Steiner tree, applying after each terminal deletion $O(\eps^{-1})$ changes to the tree (in amortized sense).
This improves over the previous $4$-approximate algorithm.
In addition to that, we show a fully dynamic $(2+\eps)$-approximate online algorithm, which makes $O(\eps^{-1} \log D)$ changes to the tree (in amortized sense) after each operation. One of the new techniques
used to obtain these results is the lazy handling of high degree Steiner nodes for arbitrary degree threshold. This improves over an algorithm by Imase and Waxman~\cite{ImaseW91}, which makes $O(t^{3/2})$ changes to process $t$ addition or removal operations, and maintains a $4$-approximate tree.
An algorithm performing a smaller number of changes was given in~\cite{GuK14}, but, as we discuss in the next section, it departs slightly from the classical Imase-Waxman model.

\paragraph{Query Steiner tree}\appref{Cor 5.20-5.21\\Section 5.4}
In the query model, as defined in~\cite{CyganKMPS10}, for a fixed graph $G$, we are asked queries
to compute an approximate Steiner tree for a set $S \subseteq V$ as fast as possible.
This models a situation when many sets of users want to setup a new multicast tree. We obtain an algorithm, which after preprocessing in $O(\sqrt{n}(m + n \log n))$ expected time uses $O(n\sqrt{n} \log n)$ space and computes a $6$-approximate Steiner tree in $O(|S| \sqrt{n} \log n)$ expected time.
In the planar case, we can compute $(2+\eps)$-approximate tree in $O(|S| \eps^{-1} \log^2 n \log D)$ expected time, using $O(\eps^{-1} n \log^2 n \log D)$ preprocessing time and space. In other words, we show a more efficient solution for computing many multicast trees in one fixed graph than 
computing each tree separately. 
This preprocessing problem is related to the study initiated in~\cite{export} where compact multicast routing schemes were shown.
In comparison with~\cite{export}, our schemes are not only compact but efficient as well.

\paragraph{Nonrearrangeable incremental Steiner tree}\appref{Thm 5.22\\Section 5.4}
In the nonrearrangeable incremental Steiner tree problem one has to connect arriving terminals to the previously constructed tree without modifying it.
We show how to implement the $O(\log n)$-approximate online algorithm for this problem given by Imase and Waxman~\cite{ImaseW91} so that it runs in $O(r \sqrt{n} \log n)$ expected time for non-planar graphs and $O(r \log^2 n \log D)$ expected time for planar graphs.
Here $r$ denotes the final number of terminals. This gives an improvement over the naive execution of this algorithm that
requires $O(r^2)$ time and resolves one of the open problems in~\cite{industry}.

\paragraph{Bipartite emulators}
As an interesting side result, we also show a different, simple approach to dynamic Steiner tree, which exposes a trade-off between the approximation ratio and the running time. \appref{Section 7} It is based on \emph{bipartite emulators}: low-degree bipartite graphs that approximate distances in the original graph.
We run the dynamic MSF algorithm on top of a bipartite emulator to obtain sublinear time dynamic algorithms.
In particular, we obtain a 12-approximate algorithm for general graphs that processes each update in $\tilde{O}(\sqrt{n})$ expected amortized time and a $(4+\eps)$-approximate algorithm for planar graphs processing updates in $\tilde{O}(\eps^{-1} \log^6 n)$ amortized time.
While our emulators are constructed using previously known distance oracles~\cite{Thorup05,Thorup04},
our contribution lies in introducing the concept of bipartite emulators, whose properties make it possible to solve the Steiner tree problem with a modification of the dynamic MSF algorithm~\cite{DBLP:journals/jacm/HolmLT01}.

\medskip

We want to stress that the construction of our algorithms for the Steiner tree, in particular the approach that combines online Steiner tree algorithm with a distance oracle, is highly modular.
Not only any improvement in the construction of the vertex-color distance oracles will result in better Steiner tree algorithms, but the vertex-color distance oracles themselves are constructed in a generic way out of distance oracles in~\cite{Thorup04,Thorup05}.
The approximation factor of $(6+\eps)$ for Steiner tree in general graphs comes from using $3$-approximate oracles combined with a $2$-approximation of the Steiner tree given by the MST in the metric closure and $(1+\eps)$-approximate online MST algorithm.
In other words we hit two challenging bounds: in order to improve our approximation factors, one would need either to improve the approximation ratio of the oracles which are believed to be optimal, or devise a framework not based on computing the MST.
The second challenge would require to construct simple and fast (e.g., near-linear time) approximation algorithms for Steiner tree that would beat the MST approximation ratio of $2$. Constructing such algorithms is a challenging open problem.

\subsection{Related results}

The problems we deal with in this paper and related have received a lot of attention in the literature. We present a brief summary in this section.

\paragraph{Vertex-color distance oracles}
Our vertex-color distance oracles fall into the model studied in the literature under the name of vertex-label distance oracles.
Dynamic vertex-color oracles for general graphs have been introduced by Hermelin et. al~\cite{vertextolabel} and improved by Chechik~\cite{chechik}.
These oracles allow only to change a color (label) of a single vertex, as opposed to our split and merge operations.
The oracle by Chechik~\cite{chechik} has expected size $\tilde{O}(n^{1+1/k})$, and reports $(4k-5)$-approximate distances in $\bigo(k)$ time.
This oracle can support changes of vertices' colors in $\bigo(n^{1/k} \log n)$ time.
Our results have much better approximation guarantee and more powerful update operations, at the cost of higher query time. A vertex-color oracle for planar graphs has been shown by Li, Ma and Ning~\cite{planarvl}, but it does not support updating colors.
The incremental variant of our oracle allows merging colors, and has only slightly higher running time.

\paragraph{Online Steiner tree}
There has been an increasing interest in the online Steiner tree and the related online MST problem in recent years, which started with a paper by Megow et al.~\cite{Wiese}.
They showed that in the incremental case one can maintain an approximate online MST in $\mclo{G}[S]$ (and consequently an approximate Steiner tree) with only a constant number of changes to the tree per terminal insertion (in amortized sense), which resolved a long standing open problem posed in~\cite{ImaseW91}.
The result of~\cite{Wiese} was improved to worst-case constant by Gu, Gupta and Kumar~\cite{GuGK13}.
Then, Gupta and Kumar~\cite{GuK14} have shown that constant worst-case number of changes is sufficient in the decremental case.
Their paper also shows a fully dynamic algorithm that performs constant number of changes in amortized sense, but in a slightly different model.
In~\cite{GuK14}, a newly added terminal is treated as a new vertex in the graph, with new distances given \emph{only} to all \emph{currently active} (not yet deleted) terminals; the remaining distances are assumed implicitly by the triangle inequality.
However, in the classical Imase-Waxman model, the entire algorithm runs on a fixed host graph that is given at the beginning, and terminals are only activated and deactivated.
After a terminal is deleted, it may still be used in the maintained tree as a Steiner node.
In particular, in our algorithms it is crucial that the newly added terminal may be connected directly to a Steiner node. This is not allowed in~\cite{GuK14}, so it seems that their analysis cannot be used in the Imase-Waxman model that is studied here.


While the online algorithms for Steiner tree have been studied extensively, our paper is the first one to show that they can be turned into algorithms with low running time. Moreover, in the decremental case and fully dynamic case we are the fist to show $(2+\eps)$-approximate algorithms.

\subsection{Organization of the paper}
The remainder of the paper is organized as follows. Section~\ref{sec:preliminaries} introduces the notation and recalls some results used further in the paper. In Section~\ref{sec:overview} we describe a shortened version of our results, which is supposed to sketch the ideas without diving into technicalities and proofs of correctness. 
Section~\ref{sec:edge_replacements} presents online algorithms for Steiner tree.
It covers our improvements in this area and analyses the algorithms in the setting where we only have access to approximate distances in the graph.
Section~\ref{sec:oracle_constr} introduces dynamic approximate vertex-color distance oracles, starting with a generic oracle, through various constructions of particular oracles, ending with the summary of our results in this area.
In Section~\ref{sec:algorithms} we show how to combine the online algorithms with distance oracles to obtain efficient algorithms for dynamic Steiner tree.
In Section~\ref{sec:enum-apply} we show alternative algorithms for dynamic Steiner tree problem based on bipartite emulators. These algorithms are slightly faster compared to the previously presented, but the approximation ratio doubles.
Finally, Section~\ref{sec:conclusions} concludes the paper and suggests directions for future research.

\section{Preliminaries}
\label{sec:preliminaries}
Let $G = (V, E, \dlug_G)$ be a graph.
Throughout the paper, we assume that each graph is undirected and has positive edge weights.
By $V(G)$ and $E(G)$ we denote the vertex and edge sets of $G$, respectively.
For $u,v \in V$, we write $\dist_G(u,v)$ to denote the distance between $u$ and $v$ in $G$.
The graph we refer to may be omitted if it is clear from the context.
For $v \in V$, $\Gamma_G(v) = \{u \in V: uv \in E\}$ is the set of neighbors of $v$ in $G$.

Whenever we give some statement with accuracy parameter $\eps > 0$, we implicitly
assume that $\eps$ is a small fixed constant. In particular, we assume $\eps < 1$.

We define $\mclo{G}=(V,\binom{V}{2},\dlug_{\mclo{G}})$ to be the metric closure of $G$, i.e., for $u,v \in V$, $\dlug_{\mclo{G}}(u,v) = \dist_G(u,v)$.
The \emph{stretch} of the metric induced by $G$ is the ratio between the longest and the shortest edge of $\mclo{G}$.
Moreover, let $\indu{G}{X}$ denote the subgraph of $G$ induced by $X \subseteq V$ and let $\mst(G)$ stand for the minimum spanning tree in $G$.
For a given graph $G$ and a terminal set $S$ we define $\st(G,S)$ to be an optimum Steiner tree in $G$ that spans $S$.
Our algorithms are based on the following well-known fact.

\begin{lemma}
\label{lem:two-apx-st}
Let $G = (V, E, \dlug_G)$ be a graph and $S \subseteq V$.
Then $\dlug_G(\st(G,S)) \leq \dlug_{\mclo{G}}(\mst(\indu{\mclo{G}}{S})) \leq 2\dlug_G(\st(G,S))$.
\end{lemma}

In each of our algorithms the ultimate goal is to maintain a good approximation of $\mst(\indu{\mclo{G}}{S})$.
Note that, as we focus on $\mst(\indu{\mclo{G}}{S})$, our algorithms can also maintain a solution to Subset TSP
on the same (dynamic) terminal set with the same approximation guarantee.
Since we work with a metric closure of $G$, the solution that we maintain is actually a collection of (approximately) shortest paths in $G$.
As we precompute these paths during initialization, our algorithms can be extended to report each such path in time that is linear in its length.

We distinguish three dynamic scenarios: \emph{incremental}, \emph{decremental} and \emph{fully dynamic}. 
All three scenarios assume that the input graph remains the same, while the set $S$ of terminals changes over time.
In the incremental scenario we consider how to efficiently update the tree after a vertex is added to $S$ (it is
announced to be a terminal), in the decremental scenario the update operation is removing a vertex from $S$,
and the fully dynamic scenario allows both these operations intermixed with each other.
As discussed already in the introduction, we are interested in performing one update operation in sublinear (amortized) time,
as a near-linear time algorithm follows easily from the works on dynamic minimum spanning forest~\cite{DBLP:journals/jacm/HolmLT01}.

\paragraph{Dynamic MSF algorithm.}
In a number of our algorithms we use an algorithm that may dynamically maintain a minimum spanning forest of a graph, subject to edge insertions and removals.
In the following, we call it a dynamic MSF algorithm.

\begin{theorem}[\cite{DBLP:journals/jacm/HolmLT01}]\label{thm:Thorup_full}
There exists a fully dynamic MSF algorithm, that for a graph on $n$ vertices
supports $m$ edge additions and removals in $O(m \log^4 n)$ total time.
\end{theorem}

Whenever we use the algorithm of Theorem~\ref{thm:Thorup_full} on some graph, we implicitly make sure that all the edges have pairwise distinct weights.
We ensure this by breaking ties between weights using the timestamp of the addition, and some arbitrary total ordering of edges present at initialization.
In particular, if we compare two edges of the same weight, we prefer to use the older one (i.e., we treat it as slightly cheaper).
If, after an edge $e$ is added to the graph, the maintained forest is altered, then the new forest includes $e$, and either the cost of the forest strictly decreased in this operation or the insertion of $e$ connected two connected components of the graph.

\section{Overview of our algorithms}\label{sec:overview}

In this section we present a high-level overview of our algorithm for dynamic Steiner tree.
This is done in three steps.
First, we describe the online algorithms that we use (section~\ref{ss:over:replacements}).
Then, we introduce dynamic vertex-color distance oracles (section~\ref{sec:approx_distance_oracles}), and finally (section~\ref{sec:overview:algorithms}) we show how to combine both these tools to obtain efficient algorithms.

\subsection{Online algorithms}\label{ss:over:replacements}
We start with presenting online algorithms for Steiner tree.
In the online setting the goal is to maintain an approximate Steiner tree and make small number of changes to the tree after a terminal is added or deleted.
We later use these algorithms to obtain efficient algorithms with low approximation factors.

By Lemma~\ref{lem:two-apx-st}, in order to maintain an approximate Steiner tree, we may maintain a tree $T$ that spans the current set of terminals $S$ in the metric closure of the graph.
Although in this section we assume that we work with a complete graph, it should not be thought of as a metric closure of $G$, but instead as its approximation.
In later sections, this approximation will be given by a distance oracle.
A $\GDapx$-approximate oracle for a graph $G=(V, E, \dlug_G)$ yields a complete graph $\GD = (V, \binom{V}{2}, \dlug_\GD)$, such that $\dist_G(u,v) \leq \dlug_\GD(u,v) \leq \GDapx \dist_G(u,v)$.
We view $\GD$ as a complete graph, but it should be noted that it is not metric, since it does not satisfy triangle inequality.
Therefore, we call $\GD$ a $\GDapx$-near metric space.

We essentially follow the core ideas of Imase and Waxman~\cite{ImaseW91}, with more modern improvements of~\cite{Wiese,GuGK13,GuK14}: to maintain good approximation ratio, it suffices to (a) as long as we can replace an edge $e'$ in the tree with a another one $e$ of significantly lower cost, proceed with the replacements; (b) defer deletions of large-degree nonterminal vertices from the tree.
However, we need some technical work to formally state these properties, especially in the context of \emph{near} metric spaces.
Furthermore, in~\cite{ImaseW91} it is only argued that deferral of deletion of non-terminal vertices of degree at least three does not influence the approximation factor much;
we generalize this result to arbitrary threshold, in particular showing that larger thresholds give better approximation guarantee.

Following~\cite{ImaseW91}, when a terminal $\tnew$ is added to $S$, we can update the current tree by connecting $\tnew$ with $\drzewo$ using the cheapest connecting edge.
This alone does not lead to a good approximation of $\mstx{\indu{\GD}{{S \cup \{ \tnew \}}}}$ because some other edges incident to $\tnew$ could possibly be used to replace expensive edges in $\drzewo$.
In our algorithms, we repeatedly replace tree edges with non-tree edges (assuring that after the replacement, we obtain a tree spanning $S$) until we reach a point at which we can be sure that the current tree is a good approximation of $\mstx{\indu{\GD}{S}}$.
We now formalize the notion of replacement.

Let $\drzewo$ be a tree in $\GD$. For any edge $e$ with both endpoints in $V(\drzewo)$, we say that
an edge $e_{\drzewo} \in E(\drzewo)$ is a {\em{friend}} of $e$ (with respect to $\drzewo$)
if $e_{\drzewo}$ lies on the unique path between the endpoints of $e$.
For a friend $e_{\drzewo}$ of $e$ with regards to $\drzewo$,
$\drzewo' := (\drzewo \setminus \{e_{\drzewo}\}) \cup \{e\}$ is a tree that spans $V(\drzewo)$ as well,
with cost $\dlug_\GD(\drzewo) - (\dlug_\GD(e_{\drzewo}) - \dlug_\GD(e))$. We say that $\drzewo'$ is created
from $\drzewo$ by {\em{replacing}} the edge $e_{\drzewo}$ with $e$, and $(e,e_{\drzewo})$ is
a {\em{replacement pair}} in $\drzewo$.

\begin{definition}[heavy, efficient and good replacement]
Let $c > \eps \geq 0$ be constants and
let $\drzewo$, $e$ and $e_\drzewo$ be defined as above. We say that the $(e,e_\drzewo)$ pair is a
\begin{enumerate}
\item {\em{$\eps$-heavy replacement}} if $\dlug_\GD(e_\drzewo) > \eps \dlug_\GD(\drzewo) / |V|$;
\item {\em{$\eps$-efficient replacement}} if $(1+\eps)\dlug_\GD(e) < \dlug_\GD(e_\drzewo)$;
\item {\em{$(\eps,c)$-good replacement}} if it is both $\eps$-efficient and $(\eps/c)$-heavy.
\end{enumerate}
\end{definition}

\begin{figure}[tb]
\centering
\begin{subfigure}{.5\textwidth}
\centering
\includegraphics{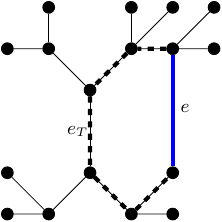}
\caption{}
\label{fig:friend}
\end{subfigure}%
\begin{subfigure}{.5\textwidth}
 \centering
 \includegraphics{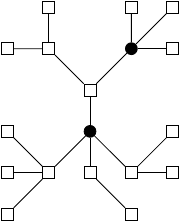}
 \caption{}
 \label{fig:degth}
\end{subfigure}%
\caption{In panel (a), all dashed edges of the tree, in particular $e_T$, are friends of the blue thick edge $e$.
  Panel (b) illustrates a tree with some nonterminal vertices that are of large degree.}
\label{fig:friend-degth}
\end{figure}

Standard arguments show that if no $(\eps,c)$-good replacement is present in the tree, then it is close to a minimum spanning tree
of its vertex set.

\begin{lemma}\label{lem:repl:over}
Let $c > \eps \geq 0$ be constants and let $\drzewo$ be a tree in a
complete graph $\GD$. If $\drzewo$ does not admit any $(\eps,c)$-good replacement,
then it is a $c(1+\eps)/(c-\eps)$-approximation of minimum spanning tree of $V(\drzewo)$.
In particular, if $\drzewo$ does not admit any $\eps$-efficient replacement,
then it is a $(1+\eps)$-approximation of minimum spanning tree on $V(\drzewo)$.
\end{lemma}

In the next lemma we observe that having nonterminals of high degree does not influence the approximation ratio significantly.
\begin{lemma}\label{lem:degth:over}
Let $\eps \geq 0$ be a constant, $\degth \geq 2$, $\mclo{G}=(V,\binom{V}{2},\dlug_{\mclo{G}})$ be a complete weighted graph and $S \subseteq V$.
Let $\drzewo_{\mst}=\mstx{\indu{\mclo{G}}{S}}$ be a MST of $S$ and let $T$ be a tree in $\mclo{G}$ that spans $S \cup N$ and does not admit $(1+\eps)$-efficient replacements.
Furthermore, assume that each vertex of $V(T) \cap N$ is of degree larger than $\degth$ in $T$.
Then $\dlug_{\mclo{G}}(\drzewo) \leq (1+\eps)\frac{\degth}{\degth-1}\dlug_{\mclo{G}}(\drzewo_{\mst})$.
\end{lemma}

The special case $\degth=2$ of Lemma~\ref{lem:degth:over} has been proven by Imase and Waxman~\cite{ImaseW91}; we provide a different, simpler proof of this result
that additionally extends to arbitrary $\degth \geq 2$.

\subsubsection{Decremental online algorithm}

In this section we describe the decremental scheme, i.e., how to handle the case when the terminals are removed from the terminal set.
In the decremental scheme, the main idea is to maintain the minimum spanning tree on the terminal set, but to postpone the deletion of terminals that are of degree above some fixed threshold $\degth \geq 2$ in this spanning tree.

When a vertex $v$ of degree $s \leq \degth$ is deleted from the tree $\drzewo$, the tree breaks into $s$ components $\ccomp^v = \{\drzewo_1,\drzewo_2,\ldots,\drzewo_s\}$ that need to be reconnected.
The natural idea is to use for this task a set of edges of $\GD$ of minimum possible total weight.
That is, we first for each $1 \leq i < j \leq s$ identify an edge $e_{ij}$ of minimum possible cost in $\GD$ among edges between $\drzewo_i$ and $\drzewo_j$.
Then, we construct an auxiliary complete graph with vertex set $\ccomp^v$ and edge cost $\dlug(\drzewo_i\drzewo_j) := \dlug_\GD(e_{ij})$, find a minimum spanning tree $\drzewo^\circ$ of this graph, and use edges of this tree to reconnect $T$.

Take $\degth = 1+\lceil \eps^{-1} \rceil$.
Lemma~\ref{lem:degth:over} implies that the decremental scheme maintains a $2(1+\degeps)\GDapx$-approximation of minimum Steiner tree of $G$.
In the course of the first $r$ deletions, it removes a vertex from the tree at most $r$ times and each removed vertex has degree at most $\degth = \bigo(\degeps^{-1})$ in the currently maintained tree.

\subsubsection{Incremental online algorithm}

We now consider the case when new terminals are added to the terminal set and we are to update the currently maintained tree.
Following Imase and Waxman~\cite{ImaseW91} and Megow et al~\cite{Wiese}, we first connect the new terminal to the tree we already have, and then try to apply all occurring $(\eps/2,1+\eps)$-good replacement pairs, for some $\eps > 0$.

At one step, given a tree $\drzewo$ and a new terminal vertex $v \in V \setminus V(\drzewo)$, we:
\begin{enumerate}
\item add an edge to $\drzewo$ connecting $v$ with $V(\drzewo)$
of cost $\min_{u \in V(\drzewo)} \dlug_\GD(uv)$;
\item apply a sequence of $\eps$-efficient replacement pairs $(e,e')$
that satisfy the following additional property:
$e'$ is a friend of $e$ in the currently maintained tree of maximum possible cost;
\item after all the replacements, we require that there does not exist a $(\eps/2,1+\eps)$-good
replacement pair $(e,e')$ {\em{with $e$ incident to $v$}}.
\end{enumerate}

We emphasize here the the incremental algorithm requires only to care about replacement edges incident to the newly added vertex;
we show that no other efficient replacement pair can be introduced in the above procedure.

Lemma~\ref{lem:repl:over} implies that the incremental algorithm for some $\eps > 0$ maintains $2(1+\eps)\GDapx$ approximation of the Steiner tree.
With the help of the recent analysis of~\cite{GuGK13}, we show that, in the course of the first $r$ additions, it performs $O(r \eps^{-1} (1+\log \GDapx))$ replacements.

\begin{figure}[tb]
\centering
\begin{subfigure}{.5\textwidth}
\centering
\includegraphics{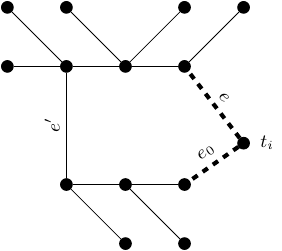}
\caption{}
\label{fig:addvert}
\end{subfigure}%
\begin{subfigure}{.5\textwidth}
 \centering
 \includegraphics{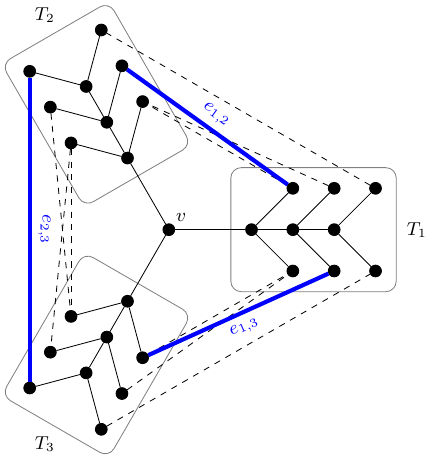}
 \caption{}
 \label{fig:delvert}
\end{subfigure}%
\caption{Panel (a) illustrates the addition step: when a new terminal $t_i$ is added, it is first connected to
  the closest terminal with an edge $e_0$, and then a replacement $(e,e')$ is applied.
  Panel (b) illustrates the deletion step: when a vertex $v$ is deleted, the tree splits into subtrees
  $T_i$, and between every pair of subtrees the shortest reconnecting edge $e_{i,j}$ is found (depicted as blue thick edge).
  Then, we find the set of reconnecting edges by computing a minimum spanning tree of the auxiliary graph $G_c$ with vertex set $\{T_1,T_2,\ldots,T_s\}$ and edges
  $e_{i,j}$ connecting $T_i$ and $T_j$.
}
\label{fig:add-del}
\end{figure}

\subsubsection{Fully dynamic online algorithm}
In order to obtain a scheme for a fully dynamic algorithm, we merge the ideas of two previous sections.
We fix an accuracy parameter $\eps > 0$.
We aim at maintaining a $(1+\eps)^2$-approximation of $\mstx{\indu{\GD}{S}}$ where $S$ is the current set of terminals.
As in the decremental step, we define $\degth = 1+ \lceil \eps^{-1} \rceil = \bigo(\eps^{-1})$, that is, $\degth$ is the minimum positive integer with $\frac{\degth}{\degth-1} \leq 1+\eps$.

In a deletion step, we behave in exactly the same manner as in the decremental scheme.
If we want to delete $v$ from the terminal set, we mark it as non-terminal and, if its degree is at most $\degth$, it is removed from the tree.

In an addition step, we perform similarly as in the incremental scheme, but there are two significant differences.
First, we do not have the guarantee that the cost of the tree will not decrease much in the future, so we cannot stop replacing edges at some cost threshold: although the low cost edges may contribute only a little to the weight of the tree, this may change after many terminals are removed and the weight of the tree drops.

Second, we need to watch out for non-terminal vertices whose degree may drop to the threshold $\degth$ as a consequence of a replacement.
Formally, to add a vertex $v$ to the terminal set $S$ we perform the following operations on the currently maintained tree $\drzewo$.
\begin{enumerate}
\item If $v \in V(\drzewo)$, mark $v$ as a terminal and finish.
\item Otherwise, connect $v$ to any vertex of $\drzewo$ using the cheapest connection in $\GD$.
\item Apply a sequence of $\eps$-efficient replacement pairs, where for each such pair $(e,e')$ we require that $e'$ has the maximum possible cost among the friends of $e$.
\item Once all replacement pairs are applied, we require that there exists no $\eps$-efficient replacement pair $(e,e')$ \emph{with $e$ incident to $v$};
\item At the end remove all non-terminal vertices whose degree dropped to at most $\degth$ due to replacements.
\end{enumerate}

Observe that, again as in the incremental scheme, we show that it suffices to care only about replacement pairs $(e,e')$ with $e$ incident to the newly added vertex $v$.

Lemmas~\ref{lem:repl:over} and~\ref{lem:degth:over} imply that the scheme maintains a  $2(1+\eps)^2\GDapx$ approximation of the Steiner tree.
In the course of the first $r$ operations, where $r_+$ of this operations are additions and $r_- = r - r_+$ are deletions,  the algorithm performs $\bigo(\eps^{-1}(r_+ \log \strecz + r_-(1+\log \GDapx)))$ replacements and removes at most $r_-$ vertices, each of which has degree at most $\degth = \bigo(\eps^{-1})$ in the currently maintained tree.

Although our problem is similar to the one considered in~\cite{GuK14}, it does not seem clear if the analysis of incremental algorithm (\cite{GuGK13}, Theorem~\ref{thm:gugk}) is applicable to our fully dynamic scheme.
The usage of this analysis in~\cite{GuK14} strongly relies on the assumption that newly added terminals are new vertices in the graph, and their distance to previously deleted terminals is inferred from the triangle inequality.
However, in our model of a fixed host graph $G$ this assumption is no longer valid, hence we need to fall back to a different analysis here.

\subsection{Dynamic vertex-color distance oracles}
\label{sec:approx_distance_oracles}
In this section we introduce dynamic vertex-color distance oracles that we later use to implement the online algorithms.
The oracles' interface is tailored to the construction of the algorithms.
In fact, the algorithms will access the graphs only by the oracles that we describe.
We consider two variants of the oracle: an \emph{incremental} and a more general \emph{fully dynamic}.
The former will be used for incremental algorithms, whereas the latter for decremental and fully dynamic ones.

\subsubsection{Oracle interface}

Let $G=(V,E)$ be a graph.
Each oracle maintains a partition of $V$ into sets $C_1, \ldots, C_k$, i.e., $\bigcup_{i=1}^k C_i=V$ and $C_i \cap C_j = \emptyset$ for all $i \neq j, 1 \leq i,j \leq k$.
In other words we assign colors to the vertices of $G$.
At any given point each color might be \emph{active} or not.
We say that a vertex is active if the associated color is active.
The oracle allows us to find the distance from a given vertex to the nearest vertex of a given color or to the nearest active vertex.
The updates supported by the oracle alter only the sets $C_i$, while the underlying graph $G$ is never changed.

In the incremental variant of the oracle, the only operation that allows us to change colors is merging two colors.
To be more precise, the incremental oracle supports the following set of operations:
\begin{itemize}
\item $\operdistance(v, i)$ -- compute the approximate distance from $v$ to the nearest vertex of color $i$,
\item $\opernearest(v, k)$ -- compute the approximate distance from $v$ to the $k$-th nearest active color (for a constant $k$),
\item $\operactivate(i)$ -- activate set (color) $C_i$,
\item $\opermerge(i,j)$ -- merge two different active sets $C_i$ and $C_j$ into one active set $C_l$, where $l \in \{ i,j\}$.
\end{itemize}

\begin{figure}[bt]
\centering
\includegraphics[width=.95\linewidth]{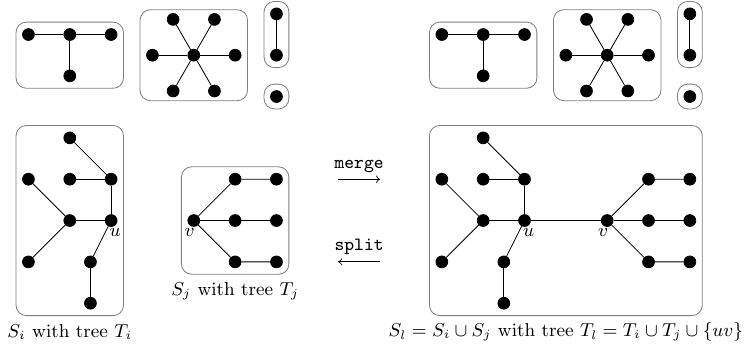}
\caption{$\opermerge$ and $\opersplit$ operations in the fully dynamic oracle.}
\label{fig:merge-split}
\end{figure}
The fully dynamic variant also supports splitting sets $C_i$, but only in a restricted way.
The oracle associates with every set $C_i$ a tree $T_i$ that spans $C_i$. Trees $T_i$ may be arbitrary trees
spanning vertices of $C_i$, and their edges are not necessarily present in $G$.
Splitting a set is achieved by specifying an edge $e$ in $T_i$.
The edge $e$ is removed from $T_i$ and the two connected components that are created specify how the set should be split (see also Figure~\ref{fig:merge-split}).
Formally speaking, the fully dynamic oracle supports the following operations:
\begin{itemize}
\item $\operdistance(v, i)$ -- compute the approximate distance from $v$ to the nearest vertex or color $i$,
\item $\opernearest(v, k)$ -- compute the approximate distance from $v$ to the $k$-th nearest active color (for a constant $k$),
\item $\operactivate(i)$ -- activate set (color) $C_i$,
\item $\operdeactivate(i)$ -- deactivate set (color) $C_i$,
\item $\opermerge(i,j,u,v)$ -- merge active sets $C_i$ and $C_j$ associated with trees $T_i$ and $T_j$ into an active color set $C_l$, $l \in \{i,j\}$, and associate it with $T_l:=T_i \cup T_j \cup \{uv\}$,
\item $\opersplit(l,u,v)$ -- if $uv$ is an edge of a spanning tree $T_l$ of an active color set $C_l$, split $C_l$ into active color sets $C_i$
 and $C_j$ associated with the two connected components of $T_l \setminus \{ uv \}$.
\end{itemize}

We would like to remark here that in our construction of the oracles,
the $\opernearest(v,k)$ query returns the distance to the $k$-th nearest color,
but we use the distances measured by the $\operdistance$ query.
That is, it returns the $k$-th smallest value of $\operdistance(v,i)$ over all active colors $i$.

\subsubsection{Generic distance oracles}
Our oracles are based on static distance oracles that may answer
vertex-to-vertex distance queries.
In addition to supporting the scheme listed above, our oracles have some structural properties important for our algorithms.
These properties are summarized
in the following definition and we explain why we need them right below.
The definition is illustrated in Fig.~\ref{fig:oracle}.

\begin{definition}\label{def:oracle}
Let $G = (V,E,\dlug_G)$ be a weighted graph and $\alpha \geq 1$.
An $\alpha$-approximate generic oracle for $G$ associates with
every $v \in V$:
\begin{itemize}
 \item a set of \emph{portals}, denoted by $\portals(v)$,
    which is a subset of $V$,
 \item and a family of \emph{pieces}, denoted by $\pieces(v)$,
    were each element (referred to as piece) is a weighted planar\footnote{The requirement that pieces are planar is not essential for the definition, but we have it here for convenience. In our oracles, pieces are either planar graphs or simply trees.} graph on a subset of $V$.
\end{itemize}

We say that a vertex $w$ is \emph{piece-visible} from $v$ if there is a piece $\piece \in \pieces(v)$
such that $w \in V(\piece)$.
The piece distance from $v$ to $w$ is defined as $R_{v,w}=\min_{\piece \in \pieces(v)} \dist_\piece(v,w)$,
and $R_{v,w} = +\infty$ if $w$ is not piece-visible from $v$.
The oracle additionally stores, for each $v \in V$ and $p \in \portals(v)$,
an approximate distance $D_{p,v} \geq \dist_G(p,v)$\footnote{Ideally, we would
like $D_{p,v}$ to be equal to $\dist_G(p,v)$ and $R_{v,w}$
equal to $\dist_G(v,w)$.
However, in some cases $D_{p,v}$ and $R_{v,w}$ might be greater.}.
For every pair $v,w \in V$ we require that $R_{v,w}=R_{w,v}$, 
and either $R_{v,w}=\dist_G(v,w)$
or 
there is a portal $p \in \portals(v) \cap \portals(w)$, such that
there is a walk from $v$ to $w$ of length at
most $\alpha \cdot \dist_G(v,w)$ that goes through $p$,
     and $D_{v,p} + D_{w,p} \leq \alpha \dist_G(v,w)$.
The \emph{portal distance} between $v$ and $w$ is
$\min_{p \in \portals(v) \cap \portals(w)} D_{v,p} + D_{w, p}$.
\end{definition}
\begin{figure}[bt]
\centering
\begin{subfigure}{.5\textwidth}
\centering
\includegraphics{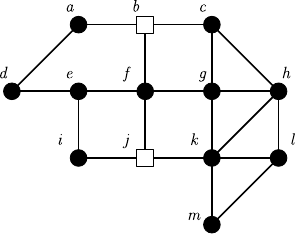}
\caption{}
\label{fig:oracle-graph}
\end{subfigure}%
\begin{subfigure}{.5\textwidth}
 \centering
 \includegraphics{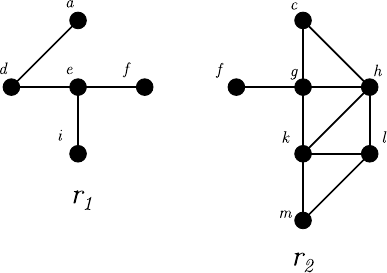}
 \caption{}
 \label{fig:oracle-pieces}
\end{subfigure}%
\caption{A generic distance oracle. Panel (a) contains an unweighted graph with two distinguished portals (marked with squares). We have that $\portals(v) = \{b, j\}$ for all vertices $v$. In panel (b) there are two pieces. In this example we assume that a piece $r_i \in \pieces(v)$ if $v \in V(r_i)$. For every $u, v \in V$, the minimum of piece and portal distances between $u$ and $v$ is at most $2\dist(u,v)$. For example, the piece distance between $e$ and $g$ is infinite, but the portal distance is $4$, which is twice as much as the exact distance.}
\label{fig:oracle}
\end{figure}

Let us now introduce the notation related to oracles.
We denote the set of all portals as $\porset=\bigcup_{v \in V} \portals(v)$,
and refer to its elements as portals.
We denote the family of pieces as $\pieset=\bigcup_{v \in V} \pieces(v)$.
The pieces in $\pieces(v)$ are supposed to represent distances from $v$ to some
vertices visible from it. For every vertex $w$, either one of piece of $v$ contains a path to $w$ of length $\dist_G(v,w)$, or there is a walk from $v$ to $w$ through a portal of both $v$ and $w$, whose length gives an approximation of $\dist_G(v,w)$.
Consider a complete graph $\GD=(V,\binom{V}{2},\dlug_\GD)$, where $\dlug_\GD(v,w)$ is defined
as the minimum over the piece distance and the portal distance between $v$ and $w$.
By Definition \ref{def:oracle} it is a well defined simple graph with finite weights.
This is the graph of distances seen by the oracle. Even though oracle's answers are approximate
with respect to $G$, they are exact with respect to $\GD$, i.e., the oracle always returns the cheapest
edge in $\GD$ between a vertex and a set of vertices of interest.
We need this property for our algorithms to work.
In particular, $\operdistance(v,i)$ returns the cheapest edge in $\GD$ between
$v$ and $C_i$, and $\opernearest(v,k)$ returns the cheapest edge in $\GD$ between
$v$ and an active vertex of $V \setminus (C_{i_1} \cup \dots \cup C_{i_{k-1}})$, where
$C_{i_1} \cup \dots \cup C_{i_{k-1}}$ are the $k-1$ active colors that are closest to $v$ (w.r.t. $\GD$).

We define a \emph{cluster} of a portal $p \in \porset$ to be set of all vertices, for which it is a portal,
i.e., $\cluster(p) = \{v \in V : p \in \portals(v)\}$.
If $p$ is not a portal, $\cluster(p) = \emptyset$.
One can view the connections between vertices $V(G)$ and portals $\porset$ as a bipartite graph
on $V(G) \uplus \porset$, where there is an edge between $p\in\porset$ and $v\in V(G)$ iff $p\in \portals(v)$
or equivalently $v \in \cluster(p)$. The number of edges in this graph is denoted as
$\portot = \sum_p |\cluster(p)| = \sum_v |\portals(v)|$.

\subsubsection{From a generic oracle to a dynamic vertex-color distance oracle}
We show that given a generic oracle, we may obtain both an incremental and fully dynamic vertex-color distance oracles.
Their efficiency is expressed in terms of oracle parameters, such as the total size of all pieces of each vertex or the total number of portals of all vertices.

We first show an incremental construction.
For every portal $p$ we maintain the distance to the nearest vertex of every color in $\cluster(p)$.
This requires only as much space as the oracle itself.
When two colors $i$ and $j$ are merged, we find the color assigned to a smaller number of vertices.
Assume it is color $i$.
Then, we change colors of all vertices of color $i$ to $j$.
This way, a vertex changes colors only $O(\log n)$ times and each time this happens, it only needs to update information in all its portals.
To find the nearest vertex of color $i$ from a vertex $v$, we first check distances to all vertices of color $i$ that are piece-visible from $v$ and then examine paths going through portals of $v$.
For a given portal, we find the distance to the nearest vertex of color $i$, which immediately yields a candidate path from $v$ to a vertex of color $i$ going through a portal.
The minimum of all candidate lengths gives the desired approximate distance.

A fully dynamic construction is more involved.
We describe a simplified version here.
The ultimate goal is also to maintain, for every portal, the distance to the nearest vertex of every color.
Recall that a fully dynamic oracle associates with every color $C_i \subseteq V$ a tree $T_i$ spanning $C_i$.
For every portal $p$ we maintain a dynamic tree $ET^p_i$, which has the same topology as $T_i$ and associates with every vertex $v \in V(T_i) \subseteq V$ a key equal to $\dist_G(p, v)$.
The dynamic trees support link and cut operations and finding minimum of all keys in a tree.
To handle an update, we iterate through all portals and link/cut the corresponding tree $ET^p_i$.
This allows us to maintain, for every portal, the nearest vertex of every color.
Answering queries can then be done as in the incremental oracle.

Next, we adapt some existing distance oracles for general and planar graphs in order to obtain generic oracles.
Basing on an oracle by Thorup and Zwick~\cite{Thorup05}, we construct a generic oracle of stretch $3$.
This requires enforcing some additional properties, that may be of independent interest.
From the generic oracle, by applying our generic construction, we immediately obtain a fully dynamic $3$-approximate vertex-color distance oracle that handles operations in $\tilde{O}(\sqrt{n})$ expected time.
In case of planar graphs, we adapt the construction by Thorup~\cite{Thorup04}, and in the end obtain an incremental $(1+\eps)$-approximate vertex-color distance oracle handling operations in $O(\eps^{-1}\log^2 n \log D)$ expected amortized time, where $D$ is the stretch of the metric induced by the input graph.
We also get a fully dynamic $(1+\eps)$-approximate vertex-color distance oracle for planar graphs that handles operations in $\tilde{O}(\eps^{-1}\sqrt{n} \log D)$ time.

\subsection{Implementing the online algorithms efficiently}
\label{sec:overview:algorithms}
We show how to use the vertex-color distance oracles to implement the online algorithms.
A $\GDapx$-approximate nearest neighbor oracle for a graph $G=(V, E, \dlug_G)$ yields a $\GDapx$-near metric space, which approximates $G$.
In other words, it works as if it was an exact oracle, but the distances between vertices were given by the edge lengths in $\GD$.
However, note that the edge lengths in $\GD$ may not satisfy triangle inequality.

\subsubsection{Efficient decremental algorithm}
In the decremental scenario we maintain an approximate Steiner tree as terminals are deleted.
The algorithm is based on vertex-color distance oracles, but it uses not only the operations provided by the oracle, but also the sets of portals and pieces of each vertex.
In particular, the algorithm uses the concepts of portal and piece distances.

For a fixed $\eps > 0$ we plan to maintain a $(1+\eps)$-approximation of $\mst(\indu{\GD}{S})$.
Following the decremental scheme, we set $\degth = 1 + \lceil \eps^{-1} \rceil = O(\eps^{-1})$.
We start with the tree $T$ being equal to the MST on the terminals.
When a vertex $v$ of degree more than $\degth$ is removed, we mark it as a nonterminal vertex, but do nothing more.
Otherwise, we need to remove $v$ and reconnect the connected components that emerge into a new tree.
The main challenge lies in finding the reconnecting edges efficiently.

In order to do that, we use the dynamic MSF algorithm (see Theorem~\ref{thm:Thorup_full}) on a graph $H$ that we maintain.
Since our goal is to maintain a tree $T$ which is an MST of the set of terminal and some nonterminal vertices, we assure that $H$ is a subgraph of $\GD$ that contains this MST (for simplicity, we assume here that the MST is unique).
Moreover, the only nonempty (i.e., containing edges) connected component of $H$ is composed exactly of terminals and nonterminal vertices which used to be terminals.
As a result, the dynamic MSF algorithm will find the tree $T$.

The simplest approach to maintaining $H$ would be to add, for every two $u, w \in S$, an edge $uw$ of length $\dlug_\GD(u,w)$.
This, however, would obviously be inefficient, i.e., work in linear time.
Instead of that, we use the structure of the vertex-color distance oracle.
For every two $u, w \in V(T)$ which are mutually piece-visible we add to $H$ an edge $uw$ of length being equal to the piece distance.
Moreover, we will be adding some edges corresponding to portal distances, but they will be chosen in a careful way to ensure that the number of such edges is low.

We initialize the edge set of $H$ to be the set of all edges of the initial $\drzewo$ and all edges between mutually piece-visible terminals.
Our algorithm uses a single instance $\DO$ of the fully dynamic vertex-color distance oracle.
We update the vertex-color distance oracle $\DO$, so that it has a single active color associated with a tree equal to $T$.
All other vertices have distinct, inactive colors.
Any change in the spanning forest of $H$ results in a constant number of modifications to $\DO$.

The main challenge is what happens if a vertex of degree $s \leq \degth$ is removed.
Then the tree $T$ decomposes into trees $T_1, \ldots, T_s$ that we need to reconnect with a set of edges of minimum total cost.
We rely on the dynamic MSF algorithm here and make sure that $H$ contains a superset of the desired edges.
First of all, $H$ surely contains all edges corresponding to piece distances, as this is maintained as an invariant of our algorithm.
Thus, we only need to assure that $H$ also contains the necessary edges corresponding to portal distances.
It is easy to observe that it suffices to add to $H$ the MST of edges corresponding to portal distances between trees $T_i$.
More formally, we consider a complete graph $G_c$ over trees $T_i$, in which the edge length between two trees is the portal distance between the nearest vertices in these trees.
We add to $H$ only edges corresponding to $\mst(G_c)$.

This MST can be simply computed by generating the graph $G_c$ explicitly.
We go through all portals and for each portal we generate edges corresponding to all portal distances that use this portal.
Note that we use the assumption that $\degth$ and the total number of portals are relatively small.

The main factors that influence the running time are: the number of pairs of mutually piece-visible terminals (all such edges are initially added to $H$), the time for computing the reconnecting edges corresponding to portal distances, and the time for updating the oracle, when a vertex of degree at most $\degth$ is removed.
Note that $r$ delete operations may cause at most $r$ vertices to be removed.

In practice, the update time is dominated by the time needed by the oracles for executing $\degth = O(\eps^{-1})$ $\opermerge$ and $\opersplit$ operations caused by a vertex removal.
In general graphs we obtain a $(6+\eps)$-approximate algorithm handling updates in $O(\eps^{-1}\sqrt{n} \log n)$ expected amortized time.
For planar graphs, the approximation ratio is $2+\eps$, whereas the amortized running time of a single update amounts to $\tilde{O}(\eps^{-1.5} \sqrt{n} \log D)$.

\subsubsection{Efficient incremental algorithm}
In the incremental scenario, the goal is to approximate the Steiner tree, as the terminals are added.
Denote the terminals that are added by $t_1, t_2, t_3, \ldots$.
Let $S_i = \{t_1, \ldots, t_i\}$ and let $T_i$ denote the tree maintained by the algorithm after adding $i$ terminals.
We maintain an approximate minimum spanning tree of $S_i$ in graph $\GD$.

Fix a constant $\eps > 0$.
We assume that the edge lengths in $\GD$ are powers of $1+\eps$.
This can be achieved by rounding up edge lengths.
Assume that all edge lengths of $\GD$ belong to the interval $[1, D]$.
We define the \emph{level} of an edge $uv$, as $\poziom(uv) := \log_{1+\eps} \dlug_\GD(uv)$.
The algorithm represents the current tree $T_i$ by maintaining $h = \lfloor \log_{1+\eps} D\rfloor$ \emph{layers} $L_1, \ldots, L_{h}$.
Layer $L_j$ contains edges of the current tree $T_i$ whose level is at most $j$.
In particular, layer $L_h$ contains all edges of $T_i$.

For each layer $L_j$, we maintain an incremental vertex-color distance oracle $\DO_j$ on the graph $G$ (see Figure~\ref{fig:DO}).
The partition of $V$ into colors corresponds to the connected components of a graph $(V,L_j)$.
A color is active if and only if its elements belong to $S_i$.

\begin{figure}[bt]
\centering
\includegraphics[width=.99\linewidth]{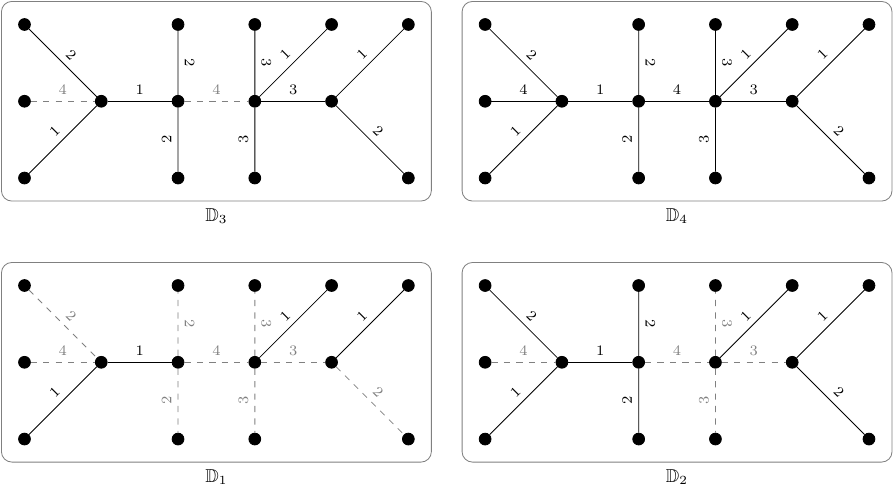}
\caption{Trees maintained by the oracles $\DO_j$ corresponding to layers $L_j$ of the tree $T_i$.
The small numbers above edges denote their levels.}
\label{fig:DO}
\end{figure}

The process of adding a terminal $t_i$ consists of two steps.
First, we find the shortest edge in $\GD$ connecting $t_i$ to any of $t_1, \ldots, t_{i-1}$ and add this edge to $T_{i-1}$ to obtain tree $T_i$.
Then, we apply a sequence of $\eps$-efficient replacements $(e, e')$ to $T_i$ in order to decrease its weight.
Recall that we may only consider replacements in which $e$ is incident to the newly added terminal $t_i$.

In order to find the replacements, we use the vertex-color distance oracles.
Fix a layer number $j$ and assume that the colors of the oracle $\DO_j$ reflect the layer structure, including the newly added terminal $t_i$.
We want to find a replacement pair $(e, e')$ such that the $\poziom(e') > j$ and $\poziom(e) \leq j$.
Denote by $C$ the connected component of the graph $(V,L_j)$ that contains $t_i$ (or a color of $t_i$ in $\DO_j$).
By definition of $C$, it consists of edges of level at most $j$ and the path in $T_i$ from $t_i$ to every $t \not\in C$ contains an edge of level $> j$.
We find the vertex $t \not\in C$ which is the nearest to $t_i$ by querying $\DO_j$.
Since $t_i$ has color $C$, we issue a $\opernearest(t_i, 2)$ query to find the second nearest color from $t_i$.
Assume that we find a vertex $t'$.
If $\poziom(t_it') \leq j$, we have found a replacement pair.
Let $e'$ be the heaviest edge on the path from $t_i$ to $t'$ in $T_i$ and $e = t_it'$.
Clearly, $(e, e')$ is a replacement pair, and since the weight of $e$ is lower than the weight of $e'$ and the weights can differ at least by a factor of $1+\eps$, this replacement is $\eps$-efficient.

The running time of the algorithm described above depends on $\log D$.
In order to circumvent this dependency, we observe that we do not need to replace edges whose length is at most $\eps/(2n)$ fraction of the weight of the current tree, as their total weight is negligible.
This allows us to limit the number of levels that we consider to roughly $O(\log n)$.

One of the factors influencing the running time is the number of replacements, but it depends only on $\eps$ and $\GDapx$.
For each replacement, we issue a constant number of oracle operations on each level, so the running time is, roughly speaking, a product of the number of replacements, the number of layers (roughly $O(\log n)$) and the running time of a single oracle operation.
In the end, we obtain $(6+\eps)$-approximate algorithm for general graphs handling insertions in $\tilde{O}(\eps^{-1}\sqrt{n})$ expected amortized time and a $(2+\eps)$-approximate algorithm for planar graphs, which processes updates in $O(\eps^{-2} \log^2 n \log (n/\eps) \log D)$ expected amortized time.

\subsubsection{Efficient fully dynamic algorithm}
The fully dynamic algorithm is obtained by appropriately merging the ideas of the decremental and incremental algorithms.
In fact we maintain both the invariants of the incremental and decremental algorithms.

Fix $\eps > 0$ and set $\degth = 1 + \lceil \eps^{-1} \rceil$.
We assume that all edges' lengths in $\GD$ are powers of $1 + \eps$, which may be assured by rounding their lengths up.
We maintain a tree $T$ that spans the set of terminals and nonterminals of degree more than $\degth$.
The tree $T$ is maintained by using dynamic MSF algorithm on a graph $H$ that we update.
It is an MST of the set of terminals and high-degree nonterminal vertices, and, since the length of every edge is a power of $1 + \eps$, this is equivalent to the fact that it does not admit any $\eps$-efficient replacements.

Assume that the edge lengths in $\GD$ belong to $[1, D]$.
Let $h = \lfloor \log_{1+\eps} D \rfloor$.
We maintain a collection of fully dynamic nearest neighbor oracles $\DO_1, \ldots, \DO_h$.
Recall that the \emph{level} of an edge $uv$ is $\poziom(uv) := \log_{1+\eps} \dlug_\GD(uv)$.
The oracle $\DO_h$ contains a single active color associated with a tree equal to $T$.
All other colors are inactive and assigned to isolated vertices.
In general, oracle $\DO_i$ reflects the forest consisting of edges of $T$ of level at most $i$.
Each active color of $\DO_i$ corresponds to one tree of this forest.
In particular the trees assigned to colors are the same as the trees in the forest.
The oracles $\DO_i$ are counterparts of the oracles from the incremental algorithm.
Since $\DO_h$ contains all edges of the current tree, it may be used as the oracle $\DO$ used in the decremental algorithm.

When a terminal $v$ is deleted, we use exactly the same procedure as in the decremental algorithm.
The only difference is that we need to update all oracles $\DO_i$, so that they reflect the changes made to the tree.

On the other hand, if a vertex $v$ is added, we use a procedure similar to the incremental algorithm.
Namely, we connect $v$ to $T$ using the shortest edge in $\GD$ and then apply all occurring $\eps$-replacements.
In order to find all replacements, we need to consider replacements at all possible $O(\log D)$ levels.
When a replacement $(e, e')$ is detected, we add edge $e$ to $H$.
Since every replacement that we find is $\eps$-efficient, the dynamic MSF algorithm will surely update the maintained tree accordingly.
Note that $e'$ may be left in $H$.
The vertex-color distance oracles are updated to reflect the changes done by the algorithm.

After all replacements are made, the tree $T$ is a MST of the set of terminals and nonterminal vertices of degree more than $\degth$.
The last step is to add to $H$ all edges $vx$ for all vertices $x$ which are piece-visible from $v$, as this is an invariant of the decremental algorithm.

Regarding efficiency, not that every edge that is added may be replaced by an edge that is shorter by a factor of $(1+\eps)$.
In general, an edge insertion may trigger $O(\log_{1+\eps} D)$ replacements, in amortized sense.
Each replacement requires us to update oracles on each of $O(\log_{1+\eps} D)$ levels.
Thus, for each added edge, we perform $O(\eps^{-2} \log^2 D)$ $\opermerge$ and $\opersplit$ operations.
Moreover, since for every vertex we need to maintain edges to all vertices piece-visible by it, we have to add or remove all such edges when a vertex is added to or removed from the tree.

For general graphs, we are able to maintain $(6+\eps)$-approximation of the Steiner tree, handling each update in $\tilde{O}(\eps^{-2} \sqrt{n} \log^2 D)$ expected amortized time.
In case of planar graphs, the approximation ratio is $(2+\eps)$, whereas the amortized expected running time of a single operation amounts to $O(\eps^{-2} \sqrt{n} \log^{2.5} D)$.
To deal with the big dependency on $\log D$, we show a general technique for decreasing the dependency on $\log D$ from $\log^c D$ to $\log D$, at the cost of adding some additional $\eps^{-1}$ and $\log n$ factors.

\section{Online algorithms}
\label{sec:edge_replacements}

In this section we describe our online algorithms for Steiner tree:
we show how to maintain a low-weight tree spanning the terminals, using small number of changes to the tree.
In other words, this section provides all details of the results mentioned
in Section~\ref{ss:over:replacements}.
For completeness, some parts of Section~\ref{ss:over:replacements} are repeated here.

Let us now give a short overview of this section.
First, we observe that the distance oracles, used with conjunction
with our the online algorithms, give only approximate distances and, consequently,
the distances seen by the online algorithm may satisfy the triangle inequality
with some slackness. To cope with that, in Section~\ref{ss:swap:near}
we introduce the notion of \emph{$\GDapx$-near metric space} and prove
its few properties that allow us to handle them almost as typical metric spaces.

Then, in Section~\ref{ss:swap:swaps} we identify which properties
of the maintained tree are needed to guarantee low weight.
Here we essentially follow the core ideas of Imase and Waxman~\cite{ImaseW91},
with a more modern improvements of~\cite{Wiese,GuGK13,GuK14}:
to maintain good approximation ratio, it suffices to
(a) as long as we can replace an edge $e'$ in the tree with a new one $e$
of significantly lower cost, proceed with the replacement;
(b) defer deletion of large-degree nonterminal vertices from the tree.
However, we need some technical work to formally state these properties,
especially in the context of \emph{near} metric spaces.
We also improve the approximation ratio in comparison to previous works.

We proceed with the actual online algorithms in subsequent sections:
the decremental scheme is treated in Section~\ref{ss:swap:dec},
the incremental in Section~\ref{ss:swap:incr},
and we merge the ideas of these two algorithms
to obtain a fully dynamic one in Section~\ref{ss:swap:fully}.

Finally, in Section~\ref{ss:swap:bootstrap} we discuss how to bootstrap our algorithm
to avoid superfluous dependency on the stretch of the metric in the fully dynamic setting.

In the previous works~\cite{ImaseW91,Wiese,GuGK13,GuK14} the main goal was to optimize
the number of changes to the tree, regardless of the actual time needed to identify them.
Here we take a different view: as our ultimate goal is to minimize the \emph{total running time}
of each operation, in amortized sense, we prefer simple and non obtrusive requirements
for the online algorithms, at the cost of (slight) increase in the actual number
of changes. Moreover, we put more stress on optimizing the final approximation ratio of the
algorithm, and identify some tradeoffs that can be made.

Consequently, in all presented online algorithms
we allow some additional slackness, comparing to the algorithms
of~\cite{ImaseW91,Wiese,GuGK13,GuK14}. For example,
the newly added terminal is not necessarily connected to the \emph{closest} present terminal,
but \emph{approximately closest}; a similar slackness is allowed when performing
edge replacements.
This relaxation, although increases the number of possible steps of the algorithm
by a constant factor, allows application of approximate distance oracles, and, consequently,
yields fast implementations of the considered dynamic algorithms.
Together with the presence of near metric spaces, the introduced relaxations add some technical difficulties
to our proofs.

We would like also to remark that, although the recent analysis of the
incremental algorithm of Imase and Waxman~\cite{ImaseW91} due to Gu, Gupta and Kumar~\cite{GuGK13}
perfectly fits into our incremental scheme, we were not able to apply the same ideas
to our fully dynamic scheme, as it is done in~\cite{GuK14}.
The reason is that our fully dynamic model differs in a few significant aspects from the one
of~\cite{GuK14}.
In~\cite{GuK14}, a newly added terminal is treated as a new vertex in the graph,
with new distances given \emph{only} to all \emph{currently active} (not yet deleted) terminals; the remaining
distances are assumed implicitly by the triangle inequality. In this manner, the analysis
of~\cite{GuGK13} is directly applicable, as the terminal closest to the newly added terminal
is always an active one. However, in our case such an interpretation
does not make sense, as we assume the entire algorithm runs on a fixed host graph, given
at the beginning. It does not seem clear whether the analysis of~\cite{GuGK13} is applicable
to our model at all.

\subsection{$\GDapx$-near metric spaces}\label{ss:swap:near}
As already discussed, in our dynamic algorithms we do not access the underlying
graph directly,
but instead use approximate distance oracles that can answer distance queries.
The oracles we use provide us with an {\em{approximation}} of the metric closure of
the input graph, that itself may not be metric, but it is always close to being
metric. In the following definition we formalize this notion of closeness.

\begin{definition}[$\GDapx$-near metric space]
We say that a complete graph $\GD = (V,\binom{V}{2},\dlug_\GD)$ is an \emph{$\GDapx$-near metric space}
if there exists a metric space $\mclo{G} = (V,\binom{V}{2},\dlug_{\mclo{G}})$ on $V$ such that for any $u,v \in V$ we have
$$\dlug_{\mclo{G}}(uv) \leq \dlug_\GD(uv) \leq \GDapx\dlug_{\mclo{G}}(uv).$$
We additionally say that {\em{$\mclo{G}$ is approximated by $\GD$}}.
\end{definition}
Intuitively, a $\GDapx$-near metric space behaves almost like a metric space, but a factor of $\GDapx$ pops up if we use the triangle inequality.
\begin{lemma}\label{lem:eps-near-triangle}
Let $\GD = (V,\binom{V}{2},\dlug_\GD)$ be a $\GDapx$-near metric space.
For any $k \geq 1$ and any sequence $v_0,v_1,v_2,\ldots,v_k \in V$
it holds that
$$\dlug_\GD(v_0v_k) \leq \GDapx \sum_{i=1}^k \dlug_\GD(v_{i-1}v_i).$$
\end{lemma}
\begin{proof}
Assume $\GD$ approximates metric space $\mclo{G} = (V,\binom{V}{2},\dlug_{\mclo{G}})$.
Then
$$\dlug_\GD(v_0v_k) \leq \GDapx \dlug_{\mclo{G}}(v_0v_k) \leq \GDapx\sum_{i=1}^k \dlug_{\mclo{G}}(v_{i-1}v_i) \leq \GDapx\sum_{i=1}^k \dlug_\GD(v_{i-1}v_i).$$
\maybeqed\end{proof}

The analysis of the approximation factors we obtain relies on Lemma~\ref{lem:two-apx-st}.
So the ultimate goal is to maintain a good approximation of $\mstx{\indu{\mclo{G}}{S}}$.

However, as the algorithm is given access only to a graph $\GD$, which is an approximation of $G$,
in subsequent sections we instead
present how to maintain a good approximation of $\mstx{\indu{\GD}{S}}$
by maintaining some tree $\drzewo$ spanning $S$
in $\GD$. Note that, as $\drzewo$ is allowed to use some non-terminal vertices,
it may happen that actually $\dlug_\GD(\drzewo) < \dlug_\GD(\mstx{\indu{\GD}{S}})$. 
In what follows, we say that $\drzewo$ is a {\em{$\apxfactor$-approximation of $\mstx{\indu{\GD}{S}}$}}
if $\dlug_\GD(\drzewo) \leq \apxfactor \dlug_\GD(\mstx{\indu{\GD}{S}})$.

Let us now formally show that approximating $\mstx{\indu{\GD}{S}}$ is sufficient for our needs.
\begin{lemma}\label{lem:GDvsmclo}
Assume $\GD =(V,\binom{V}{2},\dlug_\GD)$ is a $\GDapx$-near metric space approximating
a metric space $\mclo{G} = (V,\binom{V}{2},\dlug_{\mclo{G}})$.
Let $S \subseteq V$ be a terminal set, and let $\drzewo$ be a tree spanning $S$
and being a $\apxfactor$-approximation of the minimum spanning tree of $\indu{\GD}{S}$ in the complete graph $\GD$
(i.e., with regards to weights $\dlug_\GD$).
Then $\drzewo$ is a $2\apxfactor\GDapx$-approximation of a minimum Steiner tree
connecting $S$ in $\mclo{G}$ (i.e., with regards to weights $\dlug_{\mclo{G}}$).
\end{lemma}
\begin{proof}
Let $\overline{\drzewo}_{\mst}$ be a minimum spanning tree of $\indu{\mclo{G}}{S}$
and let $\overline{\drzewo}$ be a minimum Steiner tree connecting $S$ in $\mclo{G}$
(i.e., both $\overline{\drzewo}_{\mst}$ and $\overline{\drzewo}$ are minimum with regards to weights $\dlug_{\mclo{G}}$).
As $\mclo{G}$ is a metric, by Lemma~\ref{lem:two-apx-st}, we have $\dlug_{\mclo{G}}(\overline{\drzewo}_{\mst}) \leq 2\dlug_{\mclo{G}}(\overline{\drzewo})$.
Let $\drzewo_{\mst}$ be a minimum spanning tree of $\indu{\GD}{S}$ (i.e., with regards to weights $\dlug_\GD$).
Using the definition of a $\GDapx$-near metric space we obtain:
$$\dlug_{\mclo{G}}(\drzewo) \leq \dlug_\GD(\drzewo) \leq \apxfactor \dlug_\GD(\drzewo_{\mst})
  \leq \apxfactor \dlug_\GD(\overline{\drzewo}_{\mst})
  \leq \apxfactor\GDapx \dlug_{\mclo{G}}(\overline{\drzewo}_{\mst})
  \leq 2\apxfactor\GDapx \dlug_{\mclo{G}}(\overline{\drzewo}).$$
\maybeqed\end{proof}
Lemma~\ref{lem:GDvsmclo} allows us to henceforth focus on
a weighted complete graph $\GD$ and approximate $\mstx{\indu{\GD}{S}}$, in most cases forgetting
about the underlying approximated metric space $\mclo{G}$.

\subsubsection{Analysis of the greedy approach in near metric spaces}
As discussed at the beginning of this section, we would like to apply the recent analysis of~\cite{GuGK13} of the greedy approach for incremental
algorithm to our case. 
However, the analysis of~\cite{GuGK13} is performed in the presence of a metric space,
  and our algorithm perceives the input graph via distance oracles, and hence
  operates on a \emph{near} metric space. In this short section,
  we adapt the results of~\cite{GuGK13} to near metric spaces.

In the metric case, the following is proven in~\cite{GuGK13}.
\begin{theorem}[\cite{GuGK13}]\label{thm:gugk-orig}
Consider a set of terminals $S = \{t_1,t_2,\ldots,t_n\}$ in a metric space
$\mclo{G} = (V,\binom{V}{2},\dlug_{\mclo{G}})$.
For any $2 \leq k \leq n$, let $g_k$ be the minimum-cost (w.r.t. $\dlug_{\mclo{G}}$)
edge connecting $t_k$ with $\{t_1,t_2,\ldots,t_{k-1}\}$.
Then
$$\prod_{k=2}^n \dlug_{\mclo{G}}(g_k) \leq 4^{n-1} \prod_{e \in E(\mstx{\indu{\mclo{G}}{S}})} \dlug_{\mclo{G}}(e).$$
\end{theorem}
As shown below, in the near-metric setting, the constant $4$ increases to $4\GDapx$.
\begin{theorem}\label{thm:gugk}
Consider a set of terminals $S = \{t_1,t_2,\ldots,t_n\}$ in a $\GDapx$-near metric space
$\GD = (V,\binom{V}{2},\dlug_\GD)$.
For any $2 \leq k \leq n$, let $f_k$ be the minimum-cost (w.r.t. $\dlug_\GD$)
edge connecting $t_k$ with $\{t_1,t_2,\ldots,t_{k-1}\}$.
Then
$$\prod_{k=2}^n \dlug_\GD(f_k) \leq (4\GDapx)^{n-1} \prod_{e \in E(\mstx{\indu{\GD}{S}})} \dlug_\GD(e).$$
\end{theorem}
\begin{proof}
Assume $\GD$ approximates metric space $\mclo{G} = (V,\binom{V}{2},\dlug_{\mclo{G}})$.
For any $2 \leq k \leq n$, let $g_k$ be defined as in Theorem~\ref{thm:gugk-orig},
    that is, $g_k$ is the edge connecting $t_k$
with $\{t_1,t_2,\ldots,t_{k-1}\}$ of minimum possible cost in the metric $\dlug_{\mclo{G}}$.
By the definition of a $\GDapx$-near metric space and the choice of $f_k$ we have
\begin{equation}\label{eq:gugk:1}
\dlug_\GD(f_k) \leq \dlug_\GD(g_k) \leq \GDapx\dlug_{\mclo{G}}(g_k).
\end{equation}
Moreover, again by the definition of a $\GDapx$-near metric space we obtain:
\begin{equation}\label{eq:gugk:2}
\prod_{e \in E(\mstx{\indu{\GD}{S}})} \dlug_\GD(e) \geq \prod_{e \in E(\mstx{\indu{\GD}{S}})} \dlug_{\mclo{G}}(e) \geq \prod_{e \in E(\mstx{\indu{\mclo{G}}{S}})} \dlug_{\mclo{G}}(e).
\end{equation}
The theorem follows by combining \eqref{eq:gugk:1} and \eqref{eq:gugk:2}
with Theorem~\ref{thm:gugk-orig}.
\maybeqed\end{proof}

\subsection{Properties guaranteeing good approximation factors}\label{ss:swap:swaps}

Our goal is to maintain a tree $T$ that spans the current set of terminals $S$
in the near-metric space $\GD$.
In this section we identify the properties of the tree $T$
that guarantee good approximation factor.

Following~\cite{ImaseW91},
when a terminal $\tnew$ is added to $S$, we can update the current tree by connecting $\tnew$ with $\drzewo$ using the cheapest connecting edge.
This does not lead to a good approximation of $\mstx{\indu{\GD}{{S \cup \{ \tnew \}}}}$ because some other edges incident to $\tnew$ could possibly be used to replace expensive edges in $\drzewo$.
In our algorithms, we repeatedly replace tree edges with non-tree edges (assuring that after the replacement, we obtain a tree spanning $S$) until we reach a point at which we can be sure that the current tree is a good approximation of $\mstx{\indu{\GD}{S}}$.
We now formalize the notion of replacement.

Let $\drzewo$ be a tree in $\GD$. For any edge $e$ with both endpoints in $V(\drzewo)$, we say that
an edge $e_{\drzewo} \in E(\drzewo)$ is a {\em{friend}} of $e$ (with respect to $\drzewo$)
if $e_{\drzewo}$ lies on the unique path in $\drzewo$ between the endpoints of $e$.
For a friend $e_{\drzewo}$ of $e$ with regards to $\drzewo$,
$\drzewo' := (\drzewo \setminus \{e_{\drzewo}\}) \cup \{e\}$ is a tree that spans $V(\drzewo)$ as well,
with cost $\dlug_\GD(\drzewo) - (\dlug_\GD(e_{\drzewo}) - \dlug_\GD(e))$. We say that $\drzewo'$ is created
from $\drzewo$ by {\em{replacing}} the edge $e_{\drzewo}$ with $e$, and $(e,e_{\drzewo})$ is
a {\em{replacement pair}} in $\drzewo$.

\begin{definition}[heavy, efficient and good replacement]
Let $c > \effeps \geq 0$ be constants and
let $\drzewo$, $e$ and $e_\drzewo$ be defined as above. We say that the $(e,e_\drzewo)$ pair is a
\begin{enumerate}
\item {\em{$\effeps$-heavy replacement}} if $\dlug_\GD(e_\drzewo) > \effeps \dlug_\GD(\drzewo) / |V|$;
\item {\em{$\effeps$-efficient replacement}} if $(1+\effeps)\dlug_\GD(e) < \dlug_\GD(e_\drzewo)$;
\item {\em{$(\effeps,c)$-good replacement}} if it is both $\effeps$-efficient and $(\effeps/c)$-heavy.
\end{enumerate}
\end{definition}

Lemmas \ref{lem:no-good-replacement_full} and \ref{lem:deg3-assignment_full} that follow
will later be used to expose the properties of the tree we need to maintain in order to
obtain a constant approximation of Steiner tree.
For their proof, we need the following variant of the classical Hall's theorem.
\begin{theorem}[Hall's theorem]
Let $H$ be a bipartite graph with bipartition classes $A \uplus B = V(H)$.
Assume that there exist two positive integers $p,q$ such that for each $X \subseteq A$
it holds that $|\Gamma_H(X)| \geq \frac{p}{q} |X|$.
Then there exists a function $\pi : A \times \{1,2,\ldots,p\} \to B$
such that (i) $a\pi(a,i) \in E(H)$ for every $a \in A$ and $1 \leq i \leq p$, and (ii) $|\pi^{-1}(b)| \leq q$ for each $b \in B$.
\end{theorem}

\begin{lemma}\label{lem:no-good-replacement_full}
Let $c > \effeps \geq 0$ be constants and let $\drzewo$ be a tree in a
complete graph $\GD$. If $\drzewo$ does not admit any $(\effeps,c)$-good replacement,
then it is a $c(1+\effeps)/(c-\effeps)$-approximation of minimum spanning tree of $V(\drzewo)$.
In particular, if $\drzewo$ does not admit any $\effeps$-efficient replacement,
then it is a $(1+\effeps)$-approximation of minimum spanning tree on $V(\drzewo)$.
\end{lemma}
\begin{proof}
Let $\drzewo_{\mst}$ be a minimum spanning tree of $V(\drzewo)$.
We claim that there exists a bijection $\pi: E(\drzewo) \to E(\drzewo_{\mst})$
such that $e$ is a friend of $\pi(e)$ with regards to $\drzewo$, for any $e \in E(\drzewo)$.
Indeed,
for any $X \subseteq E(\drzewo)$, there are $|X|+1$ connected components
of $\drzewo \setminus X$ and, as $\drzewo$ and $\drzewo_{\mst}$ spans the
same set of vertices, at least $|X|$ edges of $\drzewo_{\mst}$ connect two distinct connected components of $\drzewo \setminus X$.
As each such edge has a friend in $X$, the claim follows from Hall's theorem.

Since $\drzewo$ does not admit a $(\effeps,c)$-good replacement,
for any $e \in E(\drzewo)$, either $\dlug_\GD(e) \leq \effeps \dlug_\GD(\drzewo) / (c|V|)$
or $\dlug_\GD(e) \leq (1+\effeps)\dlug_\GD(\pi(e))$.
The cost of the edges of the first type sum up to at most $\effeps \dlug_\GD(\drzewo) /c$,
whereas the cost of the edges of the second type sum up to at most
$(1+\effeps)\dlug_\GD(\drzewo_{\mst})$.
We infer
$$\dlug_\GD(\drzewo) \leq \frac{\effeps}{c} \dlug_\GD(\drzewo) + (1+\effeps)\dlug_\GD(\drzewo_{\mst})$$
and the lemma follows.
\maybeqed\end{proof}

We now move to the core arguments needed for a deletion step.
Imase and Waxman~\cite{ImaseW91} proved that deferring the deletion
of non-terminal vertices from the tree $T$ as long as they have degree at least three
incurs only an additional factor of two in the approximation guarantee.
We present here a new, simplified proof of a generalization of this fact
to arbitrary degree threshold.
\begin{lemma}\label{lem:deg3-assignment_full}
Let $\degth \geq 2$ be an integer and
let $\drzewo$ be a tree spanning $N \cup S$ where $\degree{\drzewo}{v} > \degth$
for any $v \in N$.
Let $\drzewo_S$ be any tree spanning $S$. Then there exists a function
$\pi: E(\drzewo) \times \{1,2,\ldots,\degth-1\} \to E(\drzewo_S)$ such that
\begin{enumerate}
\item $e$ is a friend of $\pi(e,i)$ with respect to $\drzewo$,
for any $e \in E(\drzewo)$ and $1 \leq i \leq \degth-1$; and
\item
$|\pi^{-1}(e')| \leq \degth$ for any edge $e'$ of $\drzewo_S$.
\end{enumerate}
\end{lemma}

\begin{proof}
We use Hall's theorem. Let $X \subseteq E(\drzewo)$; we need to prove that
there are at least $\frac{\degth-1}{\degth}|X|$ edges of $E(\drzewo_S)$ that have a friend in $X$.
Let $\ccomp$ be the family of connected components of
the forest $\drzewo \setminus X$,
    and $\ccomp_{\neg S} \subseteq \ccomp$ be
the family of these components of $\ccomp$ that do not contain any vertex of $S$.
By the properties of the set $N$, each component of $\ccomp_{\neg S}$
is adjacent to at least $\degth+1$ edges of $X$.
Consequently, since $\drzewo$ is a tree, we have
$$\left|\ccomp_{\neg S}\right| < \left|\ccomp\right| / \degth = (|X|+1)/\degth.$$
Hence, at least $\frac{\degth-1}{\degth}|X|+1$ components of $\ccomp$ contain a vertex
of $S$.
To connect these vertices, $\drzewo_S$ needs to contain at least $\frac{\degth-1}{\degth}|X|$ edges
with endpoints in different components of $\drzewo \setminus X$, and each such
edge has at least one friend in $X$. The lemma follows.
\maybeqed\end{proof}

With a similar strategy, we provide a proof Lemma~\ref{lem:emu:cut-leaves}, used later on in Section~\ref{sec:enum-apply}, but proven here for convenience.
\begin{proof}[Proof of Lemma~\ref{lem:emu:cut-leaves}]
Recall that $T_{MST} = \mst(\indu{\emuclo{G}}{S})$.
Define a tree $\widehat{T}$ as follows: for each edge $uw \in E(T_{MST})$,
take its corresponding edges $ux,wx \in E(B)$, $\dlug_B(ux) + \dlug_B(wx) = \faladist_B(uw)$,
and add them to $\widehat{T}$. In this way we obtain a tree $\widehat{T}$ in $B$
of length at most $\faladist_B(T_{MST})$ that spans a superset of $S$.
Denote $\widehat{S} = V(\widehat{T})$.
It suffices to prove that $\dlug_B(T) \leq 2\dlug_B(\widehat{T})$.

Let $\mathcal{L}_{N \setminus \widehat{S}}$ be the set of leaf vertices of $T'$
that are contained in $N \setminus \widehat{S}$, and let
$T_{\widehat{S}} = T' \setminus \mathcal{L}_{N \setminus \widehat{S}}$. 
Clearly, $T$ is a subtree of $T_{\widehat{S}}$,
  so it suffices to prove $\dlug_B(T_{\widehat{S}}) \leq 2\dlug_B(\widehat{T})$.

To this end, we show that there exists a function $\pi:E(T_{\widehat{S}}) \to E(\widehat{T})$ such that
\begin{enumerate}
\item $e$ is a friend of $\pi(e)$ with respect to $T_{\widehat{S}}$, for any $e \in E(T_{\widehat{S}})$;
\item $|\pi^{-1}(e')| \leq 2$ for every $e' \in E(\widehat{T})$.
\end{enumerate}
Observe that $T_{\widehat{S}}$ is a minimum spanning tree of $\indu{B}{\widehat{S}}$, since
$T'$ is a minimum spanning tree as well. Hence, for such a function
$\pi$ we would have $\dlug_G(e) \leq \dlug_B(\pi(e))$ for any $e \in E(T_{\widehat{S}})$ and,
  consequently, $\dlug_B(T_{\widehat{S}}) \leq 2\dlug_B(\widehat{T})$,
concluding the proof of the lemma.

By Hall's theorem, it suffices to prove that for any $X \subseteq E(T_{\widehat{S}})$,
at least $|X|/2$ edges of $E(\widehat{T})$ have a friend in $X$ (with respect to $T_{\widehat{S}}$).
To this end, consider any $X \subseteq E(T_{\widehat{S}})$ and let $\ccomp$ be the family
of connected components of the forest $T_{\widehat{S}} \setminus X$,
and $\ccomp_{\neg \widehat{S}} \subseteq \ccomp$ be
the family of these components of $\ccomp$ that do not contain any vertex of $\widehat{S}$.
Observe that each component $Q$ of $\ccomp_{\neg \widehat{S}}$ is either adjacent to at least
three edges of $X$, or consists of a single vertex of $N$ and is adjacent in $T$ to exactly
two vertices of $S \subseteq \widehat{S}$.
In the latter case, we contract $Q$ together with its two neighboring components,
decreasing $|\ccomp|$ by $2$ and $|\ccomp_{\neg \widehat{S}}|$ by one.
After $k$ contractions, we have as in Lemma~\ref{lem:deg3-assignment_full}
$$|\ccomp_{\neg \widehat{S}}| - k < (|\ccomp| - 2k)/2.$$
Consequently, $|\ccomp_{\neg \widehat{S}}| < |\ccomp|/2 = (|X|+1)/2$.
Thus, at least $|X|/2$ edges of $\widehat{T}$ are needed to connect the components
of $\ccomp \setminus \ccomp_{\neg \widehat{S}}$, and all such edges have some friend in $X$.
This finishes the proof of an existence of the function $\pi$, and concludes the proof of the lemma.
\end{proof}

Lemma \ref{lem:apx_full} that follows summarizes the properties we
need to maintain a good tree. With respect to terminal addition, it
says that it suffices to perform good replacements. With respect to
terminal removal, it
allows us to postpone the deletion of a terminal from the approximate Steiner tree
as long as it has degree larger that some fixed integer threshold $\degth$.


\begin{lemma}\label{lem:apx_full}
Let $\effeps \geq 0$ be a constant,
$\degth \geq 2$ be an integer,
    $\GD=(V,\binom{V}{2},\dlug_\GD)$ be a complete weighted graph
and $S \subseteq V$.
Let $\drzewo_{\mst}=\mstx{\indu{\GD}{S}}$ be a minimum spanning tree of $S$
and let $\drzewo$ be a tree spanning $S \cup N$ such that
for any $v \in N$ we have $\degree{T}{v} > \degth$. Then the following hold.
\begin{enumerate}
\item If $N = \emptyset$ and $\drzewo$ does not admit a $(\effeps/2, 1+\effeps)$-good replacement,
then $\dlug_\GD(\drzewo) \leq (1+\effeps)\dlug_\GD(\drzewo_{\mst})$.\label{lem:apx:inc_full}
\item If $\drzewo$ is a minimum spanning tree of $S \cup N$,
then $\dlug_\GD(\drzewo) \leq \frac{\degth}{\degth-1}\dlug_\GD(\drzewo_{\mst})$.\label{lem:apx:dec_full}
\item If $\drzewo$ does not admit a $\effeps$-efficient replacement then
$\dlug_\GD(\drzewo) \leq \frac{\degth}{\degth-1}(1+\effeps)
  \dlug_\GD(\drzewo_{\mst})$.\label{lem:apx:full_full}
\end{enumerate}
\end{lemma}
\begin{proof}
Point \ref{lem:apx:inc_full} follows from Lemma \ref{lem:no-good-replacement_full}.
To prove the remaining points, invoke Lemma \ref{lem:deg3-assignment_full} on trees $\drzewo$ and $\drzewo_{MST}$ with threshold $\degth$ to obtain
a mapping $\pi : E(\drzewo) \times \{1,2,\ldots,\degth-1\} \to E(\drzewo_{MST})$.
If $\drzewo$ is a minimum spanning tree of $S \cup N$, $\dlug_\GD(e) \leq \dlug_\GD(\pi(e,i))$ for any $e \in E(\drzewo)$ and $1 \leq i \leq \degth-1$; Point \ref{lem:apx:dec_full} follows
by summation over the entire domain of $\pi$.
If $\drzewo$ does not admit an $\effeps$-efficient replacement, $\dlug_\GD(e) \leq (1+\effeps)\dlug_\GD(\pi(e,i))$ for any $e \in E(\drzewo)$ and $1 \leq i \leq \degth-1$
and Point \ref{lem:apx:full_full} follows similarly.
\maybeqed\end{proof}

Knowing which properties we need our tree to satisfy,
we may now investigate how to maintain these properties
when confronted with addition and deletion of a terminal to and from the terminal set.
In the next subsections we show how to modify the tree to preserve the desired properties.

\subsection{Decremental online algorithm}\label{ss:dec-scheme_full}\label{ss:swap:dec}

In the decremental scheme, we study the case when terminals are removed from the terminal set. The main idea is to use
Lemma~\ref{lem:deg3-assignment_full} and
to maintain the minimum
spanning tree on the terminal set,
but to postpone the deletion of terminals that are of degree above some fixed
threshold $\degth \geq 2$ in this spanning tree.

When a vertex $v$ of degree $s \leq \degth$ is deleted from the tree $\drzewo$,
the tree breaks into $s$ components
$\ccomp^v = \{\drzewo_1,\drzewo_2,\ldots,\drzewo_s\}$
that need to be reconnected.
The natural idea is use for this task a set of edges of $\GD$ of minimum possible total weight.
That is, we first for each $1 \leq i < j \leq s$ identify
an edge $e_{ij}$ of minimum possible cost in $\GD$ among edges between
$\drzewo_i$ and $\drzewo_j$.
Then, we construct an auxiliary complete graph $H_v$ with vertex set $\ccomp^v$
and edge cost $\dlug(\drzewo_i\drzewo_j) := \dlug_\GD(e_{ij})$
and find a minimum spanning tree $\mst(H_v)$ of this graph.
We define $F^v := \{e_{ij}: \drzewo_i\drzewo_j \in E(\mst(H_v))\}$
to be the set of reconnecting edges and $\drzewo' := (\drzewo \setminus \{v\}) \cup F^v$
to be the reconnected tree $\drzewo$.
We now proceed with a formal argumentation that this natural idea indeed works as expected.

\begin{lemma}\label{lem:del-smalldeg}
Let $\drzewo$ be a tree in a complete graph $\GD=(V,\binom{V}{2},\dlug_\GD)$,
let $\ccomp$ be a family of pairwise vertex-disjoint subtrees of $\drzewo$
such that there exists a subtree $\drzewo^c$ of $\drzewo$ with
$|V(\drzewo^c) \cap V(\widehat{\drzewo})| = 1$ for any $\widehat{\drzewo} \in \ccomp$,
and let $F$ be a set of edges of $\GD$ such that (a) each edge of $F$ connects
two different trees of $\ccomp$, and (b)
  $\drzewo' := F \cup \bigcup \ccomp$ is a tree.
Let $(e,e')$ be a replacement pair in $\drzewo'$. Then one of the following holds:
\begin{enumerate}
\item $(e,e')$ is a replacement pair in $\drzewo$ as well;
\item $e' \in F$; or
\item there exists $f \in F$ such that
$(e,f)$ is a replacement pair in $\drzewo'$ and
$(f,e')$ is a replacement pair in $\drzewo$.
\end{enumerate}
\end{lemma}
\begin{proof}
Let $P$ be the unique path between the endpoints of $e$ in the tree $\drzewo'$.
By the construction of $\drzewo'$, $P$ consists of subpaths
$P_1,P_2,\ldots,P_s$, where $P_i$ is a path in $\drzewo_i \in \ccomp$,
and edges $e_2,e_3,\ldots,e_s \in F$, where each $e_i$ connects the
endpoint of $P_{i-1}$ with the starting point of $P_i$.

If $s = 1$ then $P = P_1$ and hence $(e,e')$ is a replacement pair in $\drzewo$ as well.
If $e' = e_i$ for some $2 \leq i \leq s$ then $e' \in F$ and the lemma is proven.
Otherwise, as $e'$ lies on $P$, $e' \in E(P_j)$ for some $1 \leq j \leq s$.
Observe that for any $2 \leq i \leq s$, the pair $(e,e_i)$ is a replacement pair in $\drzewo'$.
To finish the proof of the lemma it suffices to show that either $(e,e')$ or $(e_i,e')$
for some $2 \leq i \leq s$ is a replacement pair in $\drzewo$.

Let $x$ be the first vertex of $P_j$, $y$ be the last vertex of $P_j$, and let $z$ be the unique
vertex of $V(\drzewo^c) \cap V(\drzewo_j)$. Since $e' \in E(P_j)$, the edge $e'$ lies either
on the unique path between $x$ and $z$ in $\drzewo_j$, or on the unique path between
$y$ and $z$.
In the first case, observe that $(e_j,e')$ is a replacement pair in $\drzewo$ if $j > 1$,
and $(e,e')$ is a replacement pair in $\drzewo$ if $j=1$.
In the second case, symmetrically, observe that $(e_{j+1},e')$ is a replacement
pair in $\drzewo$ if $j < s$ and $(e,e')$ is a replacement pair in $\drzewo$ if $j=s$.
\maybeqed\end{proof}

\begin{lemma}\label{lem:del-deg}
Let $\drzewo$ be a tree in a complete graph $\GD=(V,\binom{V}{2},\dlug_\GD)$ that does not admit a $\effeps$-efficient replacement
and let $v \in V(\drzewo)$.
Let $\ccomp$ be the family of connected components
of $\drzewo \setminus \{v\}$ and define
the reconnecting edges $F^v$ and the reconnected tree
$\drzewo' = (\drzewo \setminus \{v\}) \cup F^v$
as at the beginning of this section.
Then $\drzewo'$ does not admit a $\effeps$-efficient replacement as well.
\end{lemma}
\begin{proof}
Let $(e,e')$ be a replacement pair in $\drzewo'$; our goal is to prove
that it is not $\effeps$-efficient.
Apply Lemma~\ref{lem:del-smalldeg}
to the pair $(e,e')$, the family $\ccomp$,
tree $\drzewo^c = \indu{\drzewo}{\{v\}\cup\Gamma_\drzewo(v)}$
and set $F^v$, and consider the three possible outcomes.

If $(e,e')$ is a replacement pair in $\drzewo$ as well, then clearly it is not $\effeps$-efficient.
Otherwise, observe first that
for every $g \in F^v$ such that $(e,g)$ is a replacement pair in $\drzewo'$,
it follows from the choice of $F^v$ that $\dlug_\GD(e) \geq \dlug_\GD(g)$.
Hence, if $e' \in F^v$ then $\dlug_\GD(e) \geq \dlug_\GD(e')$ and $(e,e')$ is not $\effeps$-efficient.

Consider now an edge $f \in F^v$ promised in the remaining case of Lemma~\ref{lem:del-smalldeg}.
As $(e,f)$ is a replacement pair in $\drzewo'$, we have from the previous reasoning
that $\dlug_\GD(e) \geq \dlug_\GD(f)$.
As $(f,e')$ is a replacement pair in $\drzewo$, it is not $\effeps$-efficient.
Therefore, $(e,e')$ is also not $\effeps$-efficient and the lemma is proven.
\maybeqed\end{proof}

We are ready to describe our dynamic algorithm. Fix an accuracy parameter
$\degeps > 0$ and denote $\degth := 1 + \lceil \degeps^{-1} \rceil$, that is,
$\degth$ is the smallest positive integer for which $\frac{\degth}{\degth-1} \leq 1+\degeps$.
Thus, by Lemma \ref{lem:apx_full}, Point \ref{lem:apx:dec_full},
to obtain a $(1+\degeps)$-approximation of $\mstx{\indu{\GD}{S}}$
where $S$ is the set of currently active terminals, it is sufficient to maintain a tree $\drzewo$
such that
\begin{enumerate}
\item $S \subseteq V(\drzewo)$;
\item $\drzewo$ is a minimum spanning tree of $V(\drzewo)$ (i.e., $\drzewo$ does not admit
    any $0$-efficient replacement);
\item any non-terminal vertex of $\drzewo$ is of degree larger than $\degth$ in $\drzewo$.
\end{enumerate}
To achieve this goal, we describe a procedure $\remove(v)$ that checks if a vertex $v$ should be removed from the currently maintained tree $\drzewo$ and, if this is the case, performs the removal.
\begin{enumerate}
\item If $v$ is a terminal, $v \notin V(\drzewo)$ or $\degree{\drzewo}{v} >\degth$, do nothing.
\item Otherwise, delete $v$ and its incident edges from $\drzewo$,
  and compute the reconnecting edges $F^v$
  and the reconnected tree $\drzewo' = (\drzewo \setminus \{v\}) \cup F^v$
  as described in the beginning of this section.
  Replace $\drzewo$ with $\drzewo'$.
  Then, recursively invoke $\remove$ on all ex-neighbors of $v$ in $\drzewo$.
\end{enumerate}
If we delete $v$ from the terminal set, we first mark $v$ as non-terminal and then call $\remove(v)$.
Lemma \ref{lem:del-deg} for $\effeps=0$ ensures that,
 if we initialize $\drzewo$ as the minimum spanning tree on the terminal set,
then $\drzewo$ satisfies all aforementioned three properties and is a $(1+\degeps)$-approximation
of $\mstx{\indu{\GD}{S}}$, where $S$ is the current set of terminals.
By Lemma~\ref{lem:GDvsmclo}, we obtain in this manner a $2(1+\degeps)\GDapx$-approximation
of minimum Steiner tree on $S$ in the metric $\mclo{G}$ that is approximated by $\GD$.

\begin{lemma}\label{lem:dec-eff_full}
Assume that an algorithm implements the decremental scheme.
Then, it maintains a $2(1+\degeps)\GDapx$-approximation
of minimum Steiner tree of $S$ in $G$.
In the course of the first $r$ deletions, the algorithm performs at most $r$ calls to the procedure $\remove$
that result in a modification of the maintained tree.
Moreover, each vertex removed by the procedure $\remove$
is of degree at most $\degth = O(\degeps^{-1})$ in the currently maintained tree.
\end{lemma}
\begin{proof}
The approximation ratio follows from the discussion above.
The claim about efficiency follows from an observation that $\remove$ modifies the maintained tree only when called upon an non-terminal
vertex in the tree of degree at most $\degth$,
and each deletion operation introduces only one non-terminal vertex.
\maybeqed\end{proof}

\subsection{Incremental online algorithm}\label{ss:inc-scheme_full}\label{ss:swap:incr}

We consider here the case when new terminals are added to the terminal set and we are to update the currently maintained tree.
Following Imase and Waxman~\cite{ImaseW91} and Megow et al~\cite{Wiese},
we would like to first connect the new terminal to tree we already have, and then
try to apply all occurring $(\effeps/2,1+\effeps)$-good replacement pairs, for some $\effeps > 0$.
The main technical contribution
of this section is a proof
that it is sufficient to
look only at replacement pairs $(e,e')$ where $e$ is incident
to the newly added vertex.

However, as we would like to use our scheme in conjunction with
approximate distance oracles, we need to introduce some additional relaxation
in terms of efficiency of conducted replacements.
We would like to require that the algorithm does not leave any
$(\effeps/2,1+\effeps)$-good replacement pair, but any replacement made by the algorithm
is only $\stepeps$-efficient for some $0 < \stepeps \leq \effeps/2$.
The first requirement yields the approximation guarantee,
while the latter controls the number of performed replacements.

We now proceed with formal description of the incremental scheme.
Assume the algorithm operates on a $\GDapx$-near metric space $\GD=(V,\binom{V}{2},\dlug_\GD)$
and we are additionally equipped with parameters $\effeps \geq 2\stepeps > 0$.
Define $c = 2\GDapx(1+\effeps)^2$.

At one step, given a tree $\drzewo$ and a new terminal vertex $v \in V \setminus V(\drzewo)$, we:
\begin{enumerate}
\item add an edge to $\drzewo$ connecting $v$ with $V(\drzewo)$
of cost $\min_{u \in V(\drzewo)} \dlug_\GD(uv)$;
\item apply a sequence of $\stepeps$-efficient replacement pairs $(e,e')$
that satisfy the following additional property:
$e'$ is a friend of $e$ in the currently maintained tree of maximum possible cost;
\item after all the replacements, we require that there does not exist a $(\effeps/2,c)$-good
replacement pair $(e,e')$ {\em{with $e$ incident to $v$}}.
\end{enumerate}

To argue about the efficiency of incremental scheme we need to use
the fact that graph $\GD$ we work on is an $\GDapx$-near metric space.
The next lemma shows that the cost of the optimal tree cannot decrease
much during the course of the algorithm.

\begin{lemma}\label{lem:inc-apx_full}
Let $\drzewo, \drzewo'$ be two trees in a $\GDapx$-near metric space
$\GD=(V,\binom{V}{2},\dlug_\GD)$,
such that $V(\drzewo) \subseteq V(\drzewo')$ and $\drzewo$ is a $\apxfactor$-approximation of
a minimum spanning tree on $V(\drzewo)$, for some $\apxfactor \geq 1$.
Then $\dlug_\GD(\drzewo') \geq \dlug_\GD(\drzewo) / (2\apxfactor\GDapx)$.
\end{lemma}
\begin{proof}
Let $\mclo{G}=(V,\binom{V}{2},\dlug_{\mclo{G}})$ be a metric space approximated by $\GD$.
Let $\overline{\drzewo}_{\mst}=\mstx{\indu{\mclo{G}}{V(\drzewo)}}$ 
be the minimum spanning tree of $V(\drzewo)$
in $\mclo{G}$ (i.e., with regards to weights $\dlug_{\mclo{G}}$)
and let $\overline{\drzewo}_S$ be the minimum Steiner tree
on the terminal set $S=V(\drzewo)$ in $\mclo{G}$. Let $\drzewo_{\mst}=\mstx{\indu{\GD}{V(\drzewo)}}$
be the minimum spanning tree of $V(\drzewo)$ in $\GD$.
Note that $\overline{\drzewo}_{\mst}$ is a $2$-approximation of the minimum Steiner tree on
$V(\drzewo)$, whereas
$\drzewo'$ is a (not necessarily optimal) Steiner tree of $V(\drzewo)$. Hence
\begin{align*}
\dlug_\GD(\drzewo') &\geq \dlug_{\mclo{G}}(\drzewo') \geq \dlug_{\mclo{G}}(\overline{\drzewo}_S) \geq
\dlug_{\mclo{G}}(\overline{\drzewo}_{\mst}) / 2 \geq
\dlug_\GD(\overline{\drzewo}_{\mst}) / (2\GDapx) \\
  &\geq \dlug_\GD(\drzewo_{\mst})/(2\GDapx)
\geq \dlug_\GD(\drzewo) / (2\apxfactor\GDapx).
\end{align*}
\maybeqed\end{proof}

In what follows we prove that it is sufficient to find replacement edges adjacent to the
new terminal.

\begin{lemma}\label{lem:inc-proof_full}
If an incremental algorithm satisfies the scheme above in an $\GDapx$-near
metric space $\GD$ then
after each step tree $\drzewo$ does not admit a $(\effeps/2, 1+\effeps)$-good
replacement pair and, consequently,
is a $(1+\effeps)$-approximation of $\mstx{\indu{\GD}{S}}$ and a $2(1+\effeps)\GDapx$-approximation
of minimum Steiner tree spanning $S$ in $G$.
\end{lemma}

\begin{proof}
Let $t_1,t_2,t_3,\ldots$ be a sequence of terminals added to the terminal set, $S_k = \{t_1,t_2,\ldots, t_k\}$ and $\drzewo_k$
be the tree maintained by the algorithm after $k$-th terminal addition, i.e., $V(\drzewo_k) = S_k$.
We inductively prove that $\drzewo_k$ does not admit a $(\effeps/2, 1+\effeps)$-good replacement pair;
the second part of the lemma follows directly from Lemmas \ref{lem:apx_full} and \ref{lem:GDvsmclo}. The claim for $k \leq 1$ is obvious;
let us assume that $k \geq 2$ and that the statement is true for the trees $\drzewo_1,\drzewo_2,\ldots,\drzewo_{k-1}$.

First, note that by Lemma \ref{lem:no-good-replacement_full}, any tree $\drzewo_i$ for $i < k$ is a $(1+\effeps)$-approximation
of a minimum spanning tree of $\indu{\GD}{S_i}$. By Lemma \ref{lem:inc-apx_full},
$\dlug_\GD(\drzewo_k) \geq \dlug_\GD(\drzewo_i) / (2\GDapx(1+\effeps))$ for any $i < k$.

Assume that there exists a $(\effeps/2)$-efficient replacement pair $(e,e')$ in $\drzewo_k$.
Assume $e = t_it_j$ where $1 \leq i < j \leq k$. Let $e_j$ be an edge of maximum
cost on the unique path
between $t_i$ and $t_j$ in $\drzewo_j$.
We now prove the following.

\begin{claim}\label{cl:inc-proof:1}
For any $j \leq j' \leq k$, the unique path between
$t_i$ and $t_j$ in $\drzewo_{j'}$ does not contain edges
of cost greater or equal to $\dlug_\GD(e_j)$.
\end{claim}
\begin{proof}
We use induction on $j'$.
In the base case $j' = j$ the claim follows from the choice of $e_j$.
In the inductive step, consider a single
replacement pair $(f,f')$ applied by the algorithm.
Denote by $\drzewo^1$ (resp. $\drzewo^2$) the maintained tree before (resp. after)
the replacement pair was applied, and
by $P_1$ (resp. $P_2$) the unique path connecting $t_i$ and $t_j$ in $\drzewo^1$ (resp. $\drzewo^2$).
Assume that no edge on $P_1$ is of cost greater on equal to $\dlug_\GD(e_j)$;
we would like to prove the same statement for $P_2$. If $P_1=P_2$ there is nothing to prove
so assume otherwise.
We have $f \in E(P_2) \setminus E(P_1)$, $f' \in E(P_1) \setminus E(P_2)$
and, by the efficiency of $(f,f')$, we have $\dlug_\GD(f) \leq \dlug_\GD(f') \leq \dlug_\GD(e_j)$.
Moreover, observe that
each edge $f'' \in E(P_2) \setminus (\{f\} \cup E(P_1))$ is a friend of
$f$ in the tree $\drzewo^1$. By the choice of $f'$, we have
$\dlug_\GD(f'') \leq \dlug_\GD(f') \leq \dlug_\GD(e_j)$ and the claim follows.
\cqed\end{proof}

By Claim~\ref{cl:inc-proof:1}, we infer that $\dlug_\GD(e') \leq \dlug_\GD(e_j)$, as $e'$
lies on the unique path between $t_i$ and $t_j$ in $\drzewo_k$.
Hence, $(e,e_j)$ is a $(\effeps/2)$-efficient replacement pair in $\drzewo_j$ and,
moreover, $e$ is incident with $t_j$. By the requirements we impose on the algorithm,
  $(e,e_j)$ is not $(\effeps/(2c))$-heavy in $\drzewo_j$, that is:
$$\dlug_\GD(e_j) \leq \frac{\effeps/2}{c} \cdot \frac{\dlug_\GD(\drzewo_j)}{|V|} \leq
\frac{\effeps/2}{2\GDapx(1+\effeps)^2} \cdot \frac{2\GDapx(1+\effeps)\dlug_\GD(\drzewo_k)}{|V|} =
\frac{\effeps/2}{1+\effeps} \cdot \frac{\dlug_\GD(\drzewo_k)}{|V|}.$$
As $\dlug_\GD(e') \leq \dlug_\GD(e_j)$, we infer that $(e,e')$ is not
$(\effeps/2,1+\effeps)$-good in $\drzewo_k$, and the lemma is proven.
\maybeqed\end{proof}

We now apply the analysis of~\cite{GuGK13}, i.e. Theorem~\ref{thm:gugk},
to obtain a bound on the number of replacements performed by the incremental scheme.
\begin{lemma}\label{lem:inc-eff_full}
Assume that an algorithm implements the incremental scheme for some choice
of parameters $\effeps \geq 2\stepeps > 0$ with $\stepeps = O(1)$.
Then, in the course of the first $r$ additions, the algorithm performs
$O(r \stepeps^{-1} (1+\log \GDapx))$ replacements.
\end{lemma}
\begin{proof}
Let $t_1,t_2,t_3,\ldots, t_r$ be a sequence of terminals added to the terminal set, $S_k = \{t_1,t_2,\ldots, t_k\}$ and $\drzewo_k$
be the tree maintained by the algorithm after $k$-th terminal addition, i.e., $V(\drzewo_k) = S_k$.
The set of replacement pairs applied by the algorithm can be arranged into a set of $r-1$ sequences of the form $(e^2_k,e^1_k)$, $(e^3_k,e^2_k)$, \ldots, $(e^{s(k)}_k, e^{s(k)-1}_k)$,
$2 \leq k \leq r$, where $e^{s(k)}_k \in E(\drzewo_r)$ and $e^1_k$ is the edge added to connect
$t_k$ to $\drzewo_{k-1}$ at the beginning of the $k$-th addition step.
As each replacement is $\stepeps$-efficient, we have
\begin{equation}\label{eq:inc-eff:1}
s(k)-1 \leq \log_{1+\stepeps} \left( \dlug_\GD(e^1_k) / \dlug_\GD(e^{s(k)}_k) \right).
\end{equation}

Consider now $\mstx{\indu{\GD}{S_r}}$ and let $e_2,e_3,\ldots,e_r$ be such a permutation
of $E(\mstx{\indu{\GD}{S_r}})$ such that $e_k$ is a friend of $e^{s(k)}_k$ w.r.t.
$\mstx{\indu{\GD}{S_r}}$,
for every $2 \leq k \leq r$.
By the properties of a minimum spanning tree,
\begin{equation}\label{eq:inc-eff:3}
\dlug_\GD(e^{s(k)}_k) \geq \dlug_\GD(e_k).
\end{equation}

Altogether, we can bound the number of replacements, equal to $\sum_{k=2}^r s(k)-1$, as follows.
\begin{align*}
\sum_{k=2}^r s(k)-1 &\leq \sum_{k=2}^r \log_{1+\stepeps} \frac{\dlug_\GD(e^1_k)}{\dlug_\GD(e^{s(k)}_k)} & \textrm{by~\eqref{eq:inc-eff:1}} \\
  &\leq \sum_{k=2}^r \log_{1+\stepeps} \frac{\dlug_\GD(e^1_k)}{\dlug_\GD(e_k)} & \textrm{by~\eqref{eq:inc-eff:3}}\\
  & \leq O(\stepeps^{-1}) \left( r + \log \frac{\prod_{k=2}^r \dlug_\GD(e^1_k)}{\prod_{k=2}^r \dlug_\GD(e_k)} \right) & \textrm{as }\stepeps = O(1)\\
  & \leq O(\stepeps^{-1}) \left(r + \log \left((4\GDapx)^{r-1}\right)\right) &\textrm{by Theorem~\ref{thm:gugk}}\\
  & \leq O(r\stepeps^{-1} (1+ \log \GDapx)). &
\end{align*}
\maybeqed\end{proof}

\subsection{Fully dynamic online algorithm}\label{ss:fully_full}\label{ss:swap:fully}
We now merge the ideas of two previous sections to obtain a scheme for a fully dynamic algorithm.
We fix three accuracy parameters $\effeps,\stepeps,\degeps > 0$,
   where $\effeps$ controls the efficiency
of replacement pairs that are allowed to remain in the tree,
$\stepeps \leq \effeps$ controls the efficiency of the replacement pairs actually made,
and hence the total number of replacements made by the algorithm,
   whereas
$\degeps$ controls the loss caused by postponing the deletion of high-degree
non-terminal vertices from the tree,
 and hence the degree threshold at which we delete non-terminal vertices from the tree.
 We aim at $(1+\effeps)(1+\degeps)$-approximation
of $\mstx{\indu{\GD}{S}}$ for $S$ being the current set of terminals.
As in the decremental step, we define
$\degth = 1+ \lceil \degeps^{-1} \rceil = O(\degeps^{-1})$,
that is, $\degth$ is the minimum positive integer with $\frac{\degth}{\degth-1} \leq 1+\degeps$.
In the algorithm, we maintain a tree $T$ spanning a set of terminals and nonterminals.
The degree of every nonterminal is greater than $\degth$ and the tree does not admit any $\effeps$-efficient replacements.

In a deletion step, we behave in exactly the same manner as in the decremental
scheme in Section \ref{ss:dec-scheme_full}.
If we want to delete $v$ from the terminal set, we mark it as non-terminal
and call the procedure $\remove(v)$
that tries to remove vertex $v$ from the tree $\drzewo$ if $v$ is not a terminal
and its degree is at most $\degth$.

In an addition step, we perform similarly as in the incremental scheme, but there are two
significant differences. First, we do not have the guarantee that the cost of the tree will not
decrease much in the future (as in Lemma \ref{lem:inc-apx_full}), so we cannot stop replacing edges
at some cost threshold. Second, we need to watch out for non-terminal vertices
whose degree may drop to the threshold $\degth$ as a consequence of a replacement.
Formally, to add a vertex $v$ to the terminal set $S$ we perform the following operations on the currently maintained tree $\drzewo$.
\begin{enumerate}
\item If $v \in V(\drzewo)$, mark $v$ as a terminal and finish.
\item Otherwise, connect $v$ to any vertex of $\drzewo$, initialize an auxiliary set $R$ as $R=\emptyset$; and
\item apply a sequence of $\stepeps$-efficient
replacement pairs, where for each such pair $(e,e')$ we require that
$e'$ has the maximum possible cost among the friends of $e$;
moreover, for each such pair we insert both endpoints
of $e'$ into the set $R$;
\item once all replacement pairs are applied, we require that there exists no $\effeps$-efficient replacement pair $(e,e')$
\emph{with $e$ incident to $v$};
\item at the end, invoke $\remove$ on all vertices in $R$ in an arbitrary order.
\end{enumerate}

We claim the following.
\begin{lemma}\label{lem:fully:apx_full}
If a fully dynamic algorithm satisfies the scheme above,
after each step the tree $\drzewo$ does not admit a $\effeps$-efficient replacement pair and each non-terminal
vertex of $\drzewo$ has degree larger than $\degth$ in $\drzewo$.
Consequently, $\drzewo$ is a $(1+\effeps)(1+\degeps)$-approximation of
$\mstx{\indu{\GD}{S}}$ and a $2(1+\effeps)(1+\degeps)\GDapx$-approximation
of minimum Steiner tree spanning $S$ in $G$, where $S$ is the current set of terminals.
\end{lemma}

\begin{proof}
The second claim of the lemma follows directly from the first claim and lemmas~\ref{lem:apx_full} and~\ref{lem:GDvsmclo}.
The claim that each non-terminal vertex of $\drzewo$ has degree larger than $\degth$ in $\drzewo$ is obvious, as we invoke
the procedure $\remove$ on a vertex whenever its degree drops in $\drzewo$.

Consider now a step of adding a vertex $v$ to the terminal set. Let $\drzewo$ be the tree before this step,
and $\drzewo'$ be the tree obtained after connecting $v$ to the $\drzewo$ and performing all replacements,
but before the first call to the procedure $\remove$ on an element of $R$.
We claim that if $\drzewo$ does not admit a $\effeps$-efficient replacement pair,
neither does $\drzewo'$.

Assume the contrary. By the definition of $\drzewo'$, if $\drzewo'$ admits
a $\effeps$-efficient replacement pair $(e,e')$, then $e$ is not incident with $v$
(such a pair is henceforth called {\em{forbidden replacement pair}}).
Let $\drzewo^0$ be the tree $\drzewo$ with $v$ connected to $\drzewo$,
and $\drzewo^1,\drzewo^2,\ldots,\drzewo^s=\drzewo'$ be the sequence
of trees computed by the algorithm; $\drzewo^j$ is constructed from $\drzewo^{j-1}$ by
applying a $\stepeps$-efficient replacement pair $(e_j,e_j')$.
Note that $\degree{{\drzewo^0}}{v} = 1$
and it follows from Lemma \ref{lem:del-deg} for the vertex $v$ that $\drzewo^0$
does not admit a forbidden replacement pair.
By our assumption, $\drzewo^s=\drzewo'$ contains a forbidden replacement pair;
let $i$ be the smallest integer for which $\drzewo^i$ contains a forbidden replacement pair, and let $(e,e')$ be any one of them.
Note that $i > 0$.

We apply Lemma~\ref{lem:del-smalldeg} to the tree $\drzewo^{i-1}$, the family $\ccomp$
being the two trees of $\drzewo^{i-1} \setminus \{e_i'\}$, the tree $\drzewo^c$ consisting of a single
edge $e_i'$, and set $F = \{e_i\}$.
By the choice of $i$, $(e,e')$ is not a replacement pair in $\drzewo^{i-1}$,
as no forbidden replacement pair exists in $\drzewo^{i-1}$.

Consider now the case $e' \in F$, that is, $e' = e_i$.
As $(e,e')$ is $\effeps$-efficient and $(e_i,e_i')$ is $\stepeps$-efficient,
we have that $(1+\effeps)\dlug_\GD(e) < \dlug_\GD(e_i')$ and, in particular, $e \neq e_i'$.
As $e$ connects the two connected components of $\drzewo^{i-1} \setminus \{e_i'\}$,
$(e,e_i')$ is a $\effeps$-efficient replacement pair in $\drzewo^{i-1}$.
Since $e$ is not incident to $v$, $(e,e_i')$ is a forbidden replacement pair in $\drzewo^{i-1}$, a contradiction to the choice of $i$.

In the remaining case, $(e_i,e')$ is a replacement pair in $\drzewo^{i-1}$.
As we require $e_i'$ to be an edge of maximum cost among the friends of $e_i$ in $\drzewo^{i-1}$,
we infer that $\dlug_\GD(e_i') \geq \dlug_\GD(e') > \dlug_\GD(e)/(1+\effeps)$,
  where the last inequality follows from the $\effeps$-efficiency of $(e,e')$.
In particular, $e \neq e_i'$,
and $(e,e_i')$ is a $\effeps$-efficient replacement pair in $\drzewo^{i-1}$.
Since $e$ is not incident to $v$, $(e,e_i')$ is a forbidden replacement pair, a contradiction to the choice of $i$.

To finish the proof of the lemma, observe that
Lemma~\ref{lem:del-deg} ensures that a call to the procedure $\remove$
cannot introduce a $\effeps$-efficient replacement pair, if it was not present prior to the call.
\maybeqed\end{proof}

We conclude this section with a note on the efficiency of our fully dynamic scheme.
\begin{lemma}\label{lem:full-eff_full}
Assume that an algorithm implements the fully dynamic scheme for some choice
of parameters $\effeps,\stepeps,\degeps > 0$ with $\stepeps \leq \effeps$
and $\stepeps = O(1)$,
and is run on an $\GDapx$-near metric space $\GD = (V,\binom{V}{2},\dlug_\GD)$ with stretch $\strecz$, and
a sequence of $r$ operations, where $r_+$ of this operations are additions and $r_- = r - r_+$ are deletions.
Then the algorithm performs $O(\stepeps^{-1}(r_+ \log \strecz + r_-(1+\log \GDapx)))$ replacements,
and at most $r_-$ calls to $\remove$ procedure that result in a modification of the maintained tree.
Moreover, each vertex removed by the procedure $\remove$
is of degree at most $\degth = O(\degeps^{-1})$ in the currently maintained tree.
\end{lemma}
\begin{proof}
The second and third bounds follow directly from the fact that $\remove$ modifies the tree only if it is called upon a non-terminal vertex of degree at most $\degth$ in the tree,
and each deletion operation introduces only one non-terminal vertex.
For the first claim, let $\dlug_{min}$ be the minimum cost of an edge in $\GD$, and consider
a potential $\Phi(\drzewo) = \sum_{e \in E(\drzewo)} \log_{1+\stepeps} (\dlug_\GD(e)/\dlug_{min})$.
Clearly, this potential
is zero at the beginning and never becomes negative. Moreover, each $\stepeps$-efficient replacement
drops the potential by at least one.

If we introduce a new edge, we increase the potential
by at most $\lceil \log_{1+\stepeps} \strecz \rceil = O(\stepeps^{-1} \log \strecz)$.
Consider now a call to the procedure $\remove$ that replaces a vertex $v$ of degree $s \leq \degth$ in $\drzewo$
with a set $F^v$ of $s-1$ edges, obtaining the reconnected tree $\drzewo'$.
Let $\ccomp^v = \{\drzewo_1,\drzewo_2,\ldots,\drzewo_s\}$ be the connected components of $\drzewo \setminus \{v\}$
and for every $1 \leq i \leq s$, let $e_i$ be the unique edge connecting $V(\drzewo_i)$ with $v$ in the tree $\drzewo$, and let $u_i$ be the endpoint of $e_i$ that belongs to $V(\drzewo_i)$.
Without loss of generality assume $\dlug_\GD(e_1) \geq \dlug_\GD(e_2) \geq \ldots \geq \dlug_\GD(e_s)$.
Observe that, by Lemma~\ref{lem:eps-near-triangle},
 for every $1 \leq i < s$ we have $\dlug_\GD(u_iu_{i+1}) \leq \GDapx(\dlug_\GD(e_i) + \dlug_\GD(e_{i+1})) \leq 2\GDapx\dlug_\GD(e_i)$.
Moreover, the set $F' := \{u_i u_{i+1}: 1 \leq i < s\}$ is a valid candidate for the set $F^v$
and hence there exists a bijective mapping $\pi:F^v \to F'$ such that $\dlug_\GD(\pi(e)) \geq \dlug_\GD(e)$ for any $e \in F^v$.
We infer that
$$\Phi(\drzewo') - \Phi(\drzewo) \leq \sum_{i=1}^{s-1} \log_{1+\stepeps} \left( \dlug_\GD(u_i u_{i+1})/\dlug_\GD(e_i) \right) \leq \log_{1+\stepeps}(2\GDapx) = O(\stepeps^{-1}(1 + \log \GDapx)).$$
The lemma follows.
\maybeqed\end{proof}

As discussed at the beginning of this section, it does not seem clear if
the analysis of incremental algorithm of~\cite{GuGK13} (i.e., Theorem~\ref{thm:gugk})
is applicable to our fully dynamic scheme.
The usage of this analysis in~\cite{GuK14} strongly relies on the assumption that newly added terminals
are new vertices in the graph, with distances to terminals already deleted being implicitly defined by the triangle inequality.
However, in our model of a fixed host graph $G$ this assumption is no longer valid, hence we need to fall back
to the more straightforward analysis of Lemma~\ref{lem:full-eff_full}.

\subsubsection{Decreasing the dependence on the stretch of the metric}\label{ss:fully:strecz_full}\label{ss:swap:bootstrap}

\newcommand{\lcrel}[1]{\mathcal{R}^{#1}}
\newcommand{\lcset}[1]{\mathcal{V}^{#1}}
\newcommand{\lcv}{\mathbf{v}}
\newcommand{\lccc}{\mathbf{c}}
\newcommand{\lcvarg}[2]{\lcv^{#1}(#2)}
\newcommand{\lcdlug}[1]{\dlug^{#1}}
\newcommand{\lcG}[1]{G^{#1}}
\newcommand{\lcA}[2]{\mathcal{A}^{#1}_{#2}}

\newcommand{\ljump}{\mathfrak{b}}
\newcommand{\ljumpup}{\mathfrak{a}}

The fully dynamic scheme presented above yields algorithms depending on the stretch of the metric.
The dependence in Lemma \ref{lem:full-eff_full} is proportional to $\log \strecz$.
However, in some implementations the dependence on $\log \strecz$ may be bigger.
In this section we present a generic method of avoiding such a dependency, at the cost of additional
polylogarithmic factors in $n=|V|$.

The main idea is similar as in the incremental algorithm: the edges of cost $\eps \dlug_\GD(\drzewo) / n$
are almost for free, and, consequently, the window of interesting edge costs 
is of stretch $O(n/\eps)$.
However, in the fully dynamic algorithm we cannot simply ignore what happens to the tree
on low-cost edges, as they may become important in the future, when the cost of the tree decreases significantly.
Thus, we adapt a different strategy: for each possible window, we maintain a solution that considers only edges
of levels from this window (and assumes the cheaper ones are for free). Formally, we prove the following.

\begin{thm}\label{thm:strecz-decrease_full}
Let $\mathcal{G}$ be a minor-closed graph class and $\apxfactor,\eps > 0$ be constants.
Assume we are given a fully dynamic algorithm that, given a graph $G \in \mathcal{G}$ with nonnegative
edge weights, $n$ vertices and metric stretch $\strecz$,
initializes in $f_{ini}(n,\log \strecz)$ time, handles updates (terminal additions and deletions) in
$f_{upd}(n,\log \strecz)$ time, handles queries on the cost of minimum Steiner tree in $f_q(n,\log \strecz)$
time with approximation guarantee $\apxfactor$, and uses $f_{mem}(n, \log \strecz)$ memory.
Assume additionally that all $f_{ini}, f_{upd}, f_q$ are non-decreasing with respect to their parameters
and that $f_{ini}$ and $f_{mem}$ are convex and at least linear in $n$.

Then there for some $\Delta = O(\eps^{-1} \log (\eps^{-1}n))$ there exists
a fully dynamic algorithm that initializes in $O(f_{ini}(n, \Delta) \eps^{-1} \log \strecz)$ time,
handles updates in $O(f_{upd}(n, \Delta) \eps^{-1} \log \strecz)$ time,
handles queries in $f_q(n, \Delta) + O(1)$ time with approximation guarantee $(\apxfactor+\eps)$
and uses $O( f_{mem}(n, \Delta) \eps^{-1} \log \strecz)$ memory.
\end{thm}

Before we dive into the depths of the proof of Theorem~\ref{thm:strecz-decrease_full},
we would like to clarify that this theorem operates on the actual graph $G$ and the metric space $\mclo{G}$
being the closure of the edge weights of $G$.
In other words, all distance oracles and subsequent $\GDapx$-near metric spaces approximating $\mclo{G}$
are hidden inside the assumed fully dynamic algorithm.

\begin{proof}
Let $\mclo{G} = (V,\binom{V}{2},\dlug_{\mclo{G}})$ be the metric closure of $G$, recall that $n = |V|$.
Let us introduce levels of edge weights in $\mclo{G}$ as follows.
Let $\beta = \beta(\eps) = 1+\eps$ and
for any $u,v \in V$, we define the level of $uv$,
$\poziom(uv)$, as $\lfloor \log_\beta \dlug(uv) \rfloor$,
i.e., $\beta^{\poziom(uv)} \leq \dlug_G(uv) < \beta^{\poziom(uv)+1}$.
Let $\ljumpup$ be the smallest integer
such that if for any two edges $e,e'$ we have $\poziom(e') - \poziom(e) \geq \ljumpup$ then $\dlug_G(e) < \dlug_G(e') / n$.
Similarly we define $\ljump$ to be the minimum difference in levels that guarantee a ratio of costs of at least $\eps/n$.
Note that $\ljumpup = O(\eps^{-1} \log n)$ and $\ljump = O(\eps^{-1} \log (\eps^{-1}n))$.
We fix $\Delta = \ljumpup + \ljump$.

For an integer $i$, we define a binary relation $\lcrel{i}$ on $V$ as follows:
$(u,v) \in \lcrel{i}$ iff there exists a path between $u$ and $v$ in $\mclo{G}$ whose each edge is of level at most $i$.
Note that, equivalently, we may require that this path resides in $G$, as the weights of edges in $G$ are nonnegative.
It is easy to see that $\lcrel{i}$ is an equivalence relation; let $\lcset{i}$ be the set of its equivalence classes
and, for any $v \in V$, $\lcvarg{i}{v}$ is the equivalence class of $v$ in $\lcrel{i}$.

For a fixed level $i$, let $\lcG{i}$ be a graph constructed from $G$ as follows: we contract all edges of $G$
of level at most $i-\ljump$, and we delete all edges $uv$, where $(u,v) \notin \lcrel{i+\ljumpup}$.
Note that $V(\lcG{i}) = \lcset{i-\ljump}$ and each connected component of $\lcG{i}$ corresponds to one equivalence class
of $\lcrel{i+\ljumpup}$. Moreover, $\lcG{i}$ is a minor of $G$ and can be constructed in linear time.

Let $\lcdlug{i}$ be the metric on $V(\lcG{i}) = \lcset{i-\ljump}$ induced by the weights of the edges of $\lcG{i}$.
Let $\lcv$ and $\lcv'$ be two distinct vertices of the same connected
component of $\lcG{i}$. As $\lcv$ and $\lcv'$ can be connected with a path with edges of level at most $i+\ljumpup$,
we have $\lcdlug{i}(\lcv\lcv') < (n-1) \beta^{i+1+\ljumpup}$ and therefore $\poziom(\lcv\lcv') \leq i + O(\eps^{-1} \log n)$
in the metric $\lcdlug{i}$. As we contracted all edges of level at most $i-\ljump$ in the process of obtaining the graph $\lcG{i}$,
we have $\poziom(\lcv\lcv') > i-\ljump$. Since $\ljump = O(\eps^{-1} \log (\eps^{-1}n))$, we infer
that the logarithm of the stretch of the metric $\lcdlug{i}$ is at most $\Delta = O(\eps^{-1} \log (\eps^{-1}n))$.

Intuitively, the graph $\lcG{i}$ contains all the information we need to construct the approximate Steiner
tree for a given set of terminals, if the most expensive edge of such a tree would have level roughly $i$.

We run several copies of the assumed fully dynamic algorithm:
for each level $i$ out of $\log_{1+\eps} \strecz = O(\eps^{-1} \log \strecz)$ levels
and for each $\lccc \in \lcset{i+\ljumpup}$
we create a copy of the fully dynamic algorithm $\lcA{i}{\lccc}$ created on the connected component
$\lccc$ of the graph $\lcG{i}$.
At any time, if $S$ is the current set of terminals in $G$, a vertex $\lcv \in V(\lcG{i})$ is a terminal
in $\lcA{i}{\lccc}$ where $\lcv \subseteq \lccc$ if $\lcv \cap S \neq \emptyset$.
Denote $S^i = \{\lcv \in V(\lcG{i}): S \cap \lcv \neq \emptyset\}$.
Note that the memory usage of the structures is as promised due to convexity of $f_{mem}$.

At initialization, note that the construction of the graph $\lcG{i}$ takes time linear in $|V(G)| + |E(G)|$, and, therefore
is subsumed by the time taken to initialize the fully dynamic algorithms. The initialization time follows from the assumptions
on the convexity of $f_{ini}$.

To handle updates (additions and deletions from the terminal set), for each level $i$ and for each $\lcv \in V(\lcG{i}) = \lcset{i-\ljump}$
we keep a counter how many terminals of $S$ are contained in $\lcv$. An addition or deletion of a vertex $v$ from the terminal set
results in a change of the counter at each of the $O(\eps^{-1} \log \strecz)$ levels; if at any level the counter changes
from $0$ to $1$ or from $1$ to $0$, an operation (addition or deletion) is performed on $\lcA{i}{\lccc}$, where $\lcv \subseteq \lccc \in \lcset{i+\ljumpup}$.
The promised update time follows.

To handle queries, first, for each level $i$, we define $F_i$ to be an (arbitrarily chosen) spanning forest
of the subgraph of $G$ that contains exactly those edges that have level at most $i$. Note that $F_i$ can be constructed
in linear time for each $i$ at the time of initialization.
Let us define $i(S)$ to be the minimum level where all terminals are contained in the same equivalence
class $\lccc(S) \in \lcset{i(S)}$. Assume $\lccc(S) \subseteq \lccc'(S) \in \lcset{i(S)+\ljumpup}$, that is,
let $\lccc'(S)$ be the equivalence class of $\lcrel{i(S)+\ljumpup}$ that contains $S$.
Let $\drzewo$ be the tree maintained by $\lcA{i(S)}{\lccc'(S)}$ and let $H$ be the edge-induced subgraph
of $G$ that contains all edges of $\drzewo$ as well as $F_{i-\ljump}$. Clearly, $H$ contains a connected component
that spans $S$ in $G$. We claim that the total cost of all edges of $H$ is a good approximation of the cost of the minimum spanning tree of $S$
in $G$ and, consequently, a good approximation of the cost of the minimum Steiner tree of $S$ in $G$. Let $\drzewo_{OPT}$ be any such
minimum Steiner tree of $S$ in $G$.

By the definition of $i(S)$, any tree spanning $S$ in $G$ needs to contain at least one edge of level at least $i$. Hence, $\dlug_G(\drzewo_{OPT}) \geq \beta^i$.
On the other hand, by the definition of $\ljump$, $\dlug_G(F_{i-\ljump}) \leq (n-1) \cdot \eps \beta^i / n < \eps \beta^i \leq \eps \dlug_G(\drzewo_{OPT})$.

As $S \subseteq \lccc(S) \in \lcset{i(S)}$, $F_i$ contains a connected component that spans $S$ and, consequently,
$\dlug_G(\drzewo_{OPT}) \leq \dlug_G(F_i) < (n-1)\beta^{i+1}$. By the definition of $\ljumpup$, $\dlug_G(\drzewo_{OPT}) < \beta^{i+\ljumpup}$,
so $\drzewo_{OPT}$ does not contain any edge of level $i+\ljumpup$ or higher. We infer that $V(\drzewo_{OPT}) \subseteq \lccc'(S)$
and, by construction of $\lcG{i}$, the cost of a minimum Steiner tree of $S^i$ in $\lcG{i}$ is at most $\dlug_G(\drzewo_{OPT})$.
Therefore, the tree maintained by $\lcA{i(S)}{\lccc'(S)}$ has cost at most $\apxfactor\dlug_G(\drzewo_{OPT})$.
We infer that $\dlug_G(H) \leq (\apxfactor+\eps)\dlug_G(\drzewo_{OPT})$.

We note that the value of $i(S)$ can be updated after each addition or deletion operation by counting the number of terminals
in each connected component of each graph $\lcG{i}$. Thus, to handle a query we need to query one algorithm $\lcA{i(S)}{\lccc'(S)}$
and to look up the precomputed cost of the forest $F_{i(S)-\ljump}$. The query time follows.
\maybeqed\end{proof}

\section{Dynamic vertex-color distance oracles}\label{sec:oracle_constr}
In this section we construct dynamic vertex-color distance oracles that we later use to implement algorithms from Section~\ref{sec:edge_replacements}.
Their interface, as well as some basic definitions have been introduced in Section~\ref{sec:approx_distance_oracles}.
This section covers the construction that allows us to extend any generic oracle to a vertex-color distance oracle.
Moreover, we show how to build generic oracles for general and planar graphs.
In the end, we obtain an incremental vertex-color distance oracle for planar graphs, and fully dynamic oracles, both for planar and general graphs.

\subsection{Generic construction}
\label{sec:generic-oracles}
In the Section~\ref{sec:approx_distance_oracles} we defined a list of queries and update operations that our oracles should support. In this section
we present generic implementations of the two schemes introduced in the previous section. The constructions presented here
are generic in a sense that no specific choice of portals or pieces is imposed. Instead, the running times and space requirements
are measured in terms of a few general parameters:
\begin{itemize}
 \item $\portot = \sum_{v \in V} |\portals(v)|$ is the (expected) total number of vertex-portal connections,
 \item $\pornum = |\porset|$ is the (expected) total number of portals,
 \item $\pormax = \max_{v \in V} |\portals(v)|$ is an upped bound on the (expected) number of portals assigned to a vertex,
 \item $\pietot = \sum_{r \in \pieset} |r|$ is the (expected) total size of all pieces,
 \item $\pienum = |\pieset|$ is the (expected) total number of different pieces,
 \item $\piemax = \max_{r \in \pieset} |r|$ is an upper bound on the (expected) size of a piece,
 \item $\piemaxvass = \max_{v \in v} |\pieces(v)|$ is an upper bound on the (expected) number of pieces assigned to a vertex,
 \item $\acttot = \sum_{v \in A} |\portals(v)|$, where $A$ is a set of vertices activated by the oracle,
    is the (expected) total number of vertex-portal connections
    over vertices that were activated by the oracle.
\end{itemize}
Note that some of our constructions are randomized, and then we measure the quality of the oracle
in expectation.
We also assume that all these values are bounded polynomially in $n$, and hence the logarithm
of any of it can be bounded by $O(\log n)$.

The problem one deals with in our oracles is finding the 
vertex-set portal distances. The piece distances
are easy to find: upon a query it suffices to run
single source shortest paths \cite{Henzinger-97} algorithm from $v$ in 
every piece of $\pieces(v)$ and compare the distances. 
This section is mostly devoted to designing
a data structure able to find portal distances. To that end, 
for a portal $p$ and color $i$ we define $N^p_i$ to be
the distance between $p$ and the nearest vertex in $\cluster(p)$ 
that has color $i$.
Every portal $p$ holds a heap $H_p$, where it stores the values $N^p_i$ 
that correspond to active colors in $\cluster(p)$.
The main goal of this section is to describe how to update the 
values $N^p_i$ and heaps $H_p$ as colors of vertices are modified.

\begin{lemma}
\label{lem:generic_incremental}
Assume we are given an $\alpha$-approximate generic distance oracle for $G$.
One can maintain expected constant-time access to the values $N^p_i$ and deterministic constant-time access to heaps $H_p$
in $O(\acttot \log n)$ 
total expected time under an arbitrary sequence of $\operactivate(i)$
and $\opermerge(i, j)$ operations.
The data structure requires $O(\portot)$ space and can be initialized 
in $O(\portot \log n)$ time.

The data structure can be derandomized, at the cost of increase to total $O(\acttot \log n \log \log n)$ running time,
and $O(\log \log n)$ time to access a single value of $N_i^p$.
\end{lemma}

\begin{proof}
Every update operation affects vertices of some color and all their portals.
Thus, for every color $i$ we define a set of pairs $L(i) = \{ (v, p) | v \in C_i, p \in \portals(v)\}$
and maintain it as a list.
Note that the total size of all those lists is exactly $\portot$
and at the beginning all the colors are inactive.
Recall that colors have to be activated first in order to be merged.
The total size of the lists of colors that will be activated during the runtime of the oracle is $\acttot$.
On the remaining lists no operations will ever be executed.

As mentioned before, each portal $p \in P$ holds a heap $H_p$ storing
distances $N^p_i$ to all active colors $i$ in $\cluster(p)$.
Together with $N^p_i$ we store on $H_p$ the vertex that provides $N^p_i$,
i.e., the vertex $v$ in $\cluster(p)$ of color $i$ whose distance from $p$ is $N^p_i$.
We define $\overline{N}^p_i=(N^p_i,v)$ to be such a pair. We store pairs $\overline{N}^p_i$ in a heap, sorted according to the first coordinate only.

At the beginning, all colors are
inactive. Every portal $p$ initializes $H_p$ to be an empty heap.
In addition to $H_p$, every portal maintains two dictionaries,
where one can add/remove \emph{(key,value)} pairs and look up keys in
time $O(t_\text{dict})$ per operation; we address the choice of the dictionary, and henceforth determine the value of $t_\text{dict}$, at the end of the proof.
A dictionary of inactive distances, $I_p$,
stores pairs $(i,\overline{N}^p_i)$ for every inactive color $i$ which has
a nonempty intersection with $\cluster(p)$. Initially, every portal iterates over its cluster
and computes pairs $(i,\overline{N}^p_i)$ based on stored distances $D_{p,v}$.
Every portal
also holds a dictionary of active colors $A_p$, which stores pairs $(i,q)$ for every active color $i$ in $\cluster(p)$,
where $q$ is a pointer to a position on $H_p$ where $\overline{N}^p_i$ is stored. Initially $A_p$ is empty.

To perform $\operactivate(i)$ it suffices to update all relevant heaps $H_p$ with information about color $i$.
We iterate over elements of $L(i)$ and for each pair $(v, p)$ we remove $(i,\overline{N}^p_i)$ from $I_p$ (if it is there)
and add $\overline{N}^p_i$ to $H_p$ and $(i,q)$ to $A_p$, where $q$ is a pointer to $\overline{N}^p_i$ on $H_p$.
This is done at most once for each of the $O(\acttot)$ pairs, so the total time of this operation is
$O(\acttot (\log n+t_\text{dict}))$.

Now, consider a $\opermerge(i, j)$ operation, and keep in mind that both $i$ and $j$ are active.
Without loss of generality, we may assume that $|L(i)| \leq |L(j)|$.
We change the color of all vertices in $C_i$ to $j$.
To be more precise, we iterate over $L(i)$ and for each pair $(v,p)$ we:
\begin{align*}
&\textrm{remove } i,j \textrm{ from } A_p \textrm{, set } \overline{N}^p_j = \min(\overline{N}^p_i, \overline{N}^p_j)
\textrm{ and insert } \overline{N}^p_j \textrm{ to } H_p \textrm{ and } A_p & \textrm{if } i,j \in A_p
  \\ &\textrm{we remove } (i,q) \textrm{ from } A_p \textrm{ and add } (j,q) \textrm{ to } A_p, & \textrm{if } i\in A_p, j \notin A_p
  \\ &\textrm{do nothing } & \textrm{if } i \notin A_p
\end{align*}
Finally, we append list $L(i)$ to $L(j)$ in $O(|L(i)|)$ time.

Note that for each appended pair, the length of the list it belongs to doubles.
Since the active list length is bounded by $\acttot$, each of $\acttot$ pairs will change its list at most
$O(\log \acttot)=O(\log n)$ times.
Each such list change of a pair may trigger a dictionary update and a heap update.
The total time for list merging, iterating and dictionary modifications is therefore $O(\acttot t_\text{dict} \log n)$.
We will bound the number of operations on heaps separately.
Note
that only when $i,j \in A_p$ the merge is effective in $\cluster(p)$ and $O(\log n)$ operation on a heap is triggered.
Otherwise, if $j \notin A_p$, only $A_p$ is modified and not $H_p$.
For every portal $p$ we effectively merge colors in its cluster at
most $|\cluster(p) \cap A|-1$ times, where $A$ is the set of activated vertices. So the number of times
we modify a heap is $\sum_{v \in A} |\portals(v)| = O(\acttot)$.
This gives the total running time of $O(\acttot t_\text{dict} \log n)$.

Observe that each heap $H_p$ can be accessed in constant time, and each value $N_i^p$ in $O(t_{\textrm{dict}})$ time.

Finally, we address the choice of the dictionaries. If use hash maps to implement dictionaries, 
then we obtain $t_\text{dict}=O(1)$ in expectation.
As an alternative, observe that every used dictionary uses color numbers as keys.
Hence, one might use Y-fast-tries~\cite{Y-fast-tries} yielding a slightly 
worse but deterministic $t_\text{dict}=O(\log \log n)$ worst case time.
\end{proof}

\begin{lemma}
\label{lem:generic_full}
Assume we are given an $\alpha$-approximate generic distance oracle for a weighted graph $G=(V,E,\dlug)$.
Let $n = |V|$.
One can maintain constant-time access to the values $N^p_i$ and heaps $H_p$, under a sequence 
of $\operactivate(i)$, $\operdeactivate(i)$, $\opermerge(i,j)$ or $\opersplit(l, i, j)$ operations.
Each operation requires $O(\pornum \log n)$ worst case time.
The data structure uses $O(\portot \log n)$ space and can be initialized in $O(\portot \log n)$ time
with arbitrary forest $T_1, \dots, T_k$ spanning $V(G)$ corresponding to colors $C_1, \dots, C_k$.
\end{lemma}

\begin{proof}
In the proof we use Euler tour trees introduced in \cite{Tarjan97dynamictrees}, called further the ET-trees.
An ET-tree represents an Euler tour of an arbitrary unrooted tree $T$, i.e., a tour that traverses each edge of the tree twice, 
once in each direction.
More formally, each edge of $T$ is replaced with two oppositely 
directed edges and a directed self loop is added in every vertex.
We represent the Euler tour of the resulting graph as a balanced binary tree that allows join and split operations.
The tree has a single node for every edge of the tour, and the nodes are arranged in the same order as the tour.

The trees represented in this way can be linked by adding an edge or split by removing some edge (i.e., they support link and cut operations), and the corresponding ET-trees may be updated by joining and splitting the representing binary trees.
Each of these operations can be implemented in $O(\log n)$ time.
ET-trees can also support typical extensions of binary search trees, including associating keys with vertices and finding the minimum over 
all the keys in the tree.
Note that the keys of vertices of $T$ are associated with the self loops in the Euler tour representation.
Thus, for every node of $T$, the ET-tree has a single node holding its key.
On the other hand, some nodes of the ET-tree may not correspond to any node of $T$.

In our data structure, for each color $C_i$ we maintain the corresponding tree $T_i$ as an ET-tree $ET_i$.
We associate a key of size $O(\pornum)$ with every node of $ET_i$.
Hence, each operation on $ET_i$ requires $O(\pornum \log n)$ time.

Let us now describe the keys in detail.
Consider a color $i$, a node $x$ of $ET_i$ and a portal $p$.
The subtree of $ET_i$ rooted at $x$ corresponds naturally to some set of vertices of color $i$.
Denote this set by $C_i(x)$.
In particular, the entire $ET_i$ corresponds to $C_i$.
We define $N^p_i(x)$ to be the distance between $p$ and the nearest vertex that belongs both to $C_i(x)$ and $\cluster(p)$.

A key associated with $x$ is a list $L(x)$ of pairs $(p, N^p_i(x))$.
The list is sorted according to some (arbitrary) linear ordering of portals.
This allows us to compute the key of any node in $ET_i$ in $O(\pornum)$ time, using the keys stored in the child nodes.

In order to save space, if a subtree of $ET_i$ rooted at a node $x$ does not contain any nodes from $\cluster(p)$, we do not store 
the list element for portal $p$ in the list $L(x)$.
Let us bound the total length of all lists $L(x)$ in all ET-trees.
Consider a vertex $v$.
It belongs to exactly one node $x$ of some ET-tree and may cause that $O(|\portals(v)|)$ elements appear in every list on the path 
from $x$ to the root of the corresponding ET-tree.
Since the depth of each tree representing an ET-tree is $O(\log n)$, the total length of all lists $L(x)$, as well as the space usage of all ET-trees 
is $O(\sum_{v \in V} |\portals(v)| \log n) = O(\portot \log n)$.
They can also be initialized, given any forest $T_1 \dots T_k$ spanning $V(G)$ and representing colors $C_1 \dots C_k$, 
in $O(\portot \log n)$ time.

Clearly, each value $N^p_i$ can be extracted from the root of the tree $ET_i$.
In order to enable constant
lookup of values $N^p_i$, useful for efficient distance queries, for each color $i$ 
we store a deterministic dictionary of Hagerup et al.~\cite{Hagerup200169}.
A dictionary with $k$ entries can be initialized in $O(k \log k)$, and allows constant time access.
It uses $O(k)$ space.
Since every operation affects at most two colors, we can afford to rebuild 
the affected dictionaries after every operation 
in $O(\pornum \log n)$ time. 

A $\opermerge(i,j)$ or $\opersplit(l, i, j)$ operation can be handled with a single operation on one of the ET-trees.
Each such operation may affect at most two elements of each heap $H_p$.
Thus, updating both the ET-trees and the heaps $H_p$ requires $O(\pornum \log n)$ time.

The $\operactivate(i)$ and $\operdeactivate(i)$ operations do not affect the values $N^p_i$, but they may affect the heaps $H_p$.
Again, it takes $O(\pornum \log n)$ time to update them.

We remark here that, contrary to the previous lemma, we have $O(\pornum \log n)$ deterministic worst case time
for each operation, as we can afford to use always Y-fast-tries to maintain pointers
to all colors present in a heap $H_p$.
\end{proof}

In our data structures we also need to find vertices of a given color in a particular piece.
Using the data structures from the above lemma, for a given vertex, we may determine its color in $O(\log n)$ time.
This can be achieved by simply maintaining a pointer from each vertex to the corresponding node in an ET-tree.
Once we find this node, we need to find the root of its ET-tree which requires $O(\log n)$ time.
However, if the total number of distinct pieces is small, we may use a slightly faster approach.

\begin{lemma}
\label{lem:generic_listing}
Assume we are given an $\alpha$-approximate generic distance 
oracle for a graph $G = (V,E,\dlug)$.
Let $n = |V|$. 
There exists a data structure that allows to list all vertices 
of a given color in a given piece, and list all active vertices 
in a given piece.
Both lists are computed in time that is linear in their size.

The data structure can be initialized in $O(\pietot \log n)$ 
time and 
updated after every $\operactivate(i)$, $\operdeactivate(i)$, $\opermerge(i,j)$ 
or $\opersplit(l, i, j)$ operation in 
$O(\pienum \log n )$ worst 
case time.
It requires $O(\pietot \log n)$ space.
\end{lemma}

\begin{proof}
For each color $i$ we maintain
yet another ET-tree $\widetilde{ET}_i$ corresponding to 
tree 
$T_i$ associated with $C_i$.
Consider a piece $\piece \in \pieset$ and a node $x$ of 
$\widetilde{ET}_i$.
We define $x$ to be $\piece$-\emph{heavy} if either $x$ 
corresponds to a vertex $v \in \piece$,
$x$ is a root of $\widetilde{ET}_i$, or $x$ has 
$\piece$-heavy nodes in both subtrees.

For a node $y$ of $\widetilde{ET}_i$ (not necessarily an 
$\piece$-heavy one) we extend the key associated with $y$ with 
a list $L(y)$ of triples $(\piece, left_\piece(y), right_\piece(y))$, 
where $\piece$ is a piece.
Let $z$ and $z'$ be 
the highest $\piece$-heavy vertices in left 
and right subtree of $y$ respectively.
Then, $left_\piece(y)$ and $right_\piece(y)$
are pointers to the position of $\piece$ in $L(z)$ and $L(z')$
respectively.  
We do not store such a triple in $L(y)$ if none of the 
subtrees contains an $\piece$-heavy node.
The list $L(y)$ is sorted according to some arbitrary 
linear order on pieces.

Note that we are able to calculate values of $left_\piece(y)$ 
and $right_\piece(y)$ using the keys stored in the sons of $y$.
This allows us to compute the key of a node of $\widetilde{ET}_i$ 
in $O(\pienum)$ time, so each ET-tree operation requires 
additional $O(\pienum \log n)$ time.

Consider $\piece$-\emph{heavy} vertices of $\widetilde{ET}_i$.
Together with $left_\piece$ and $right_\piece$ pointers they form a 
binary tree, whose leaves correspond to vertices of piece $\piece$.
Since the size of a binary tree is linear in the number 
of its leaves, the pointers allow us to iterate over all 
vertices of color $i$ in a piece $\piece$ in linear time.

We store information about an arbitrary vertex $v$ of a 
color $i$ in a piece $\piece$ only on a path between a 
corresponding vertex $x$ of $\widetilde{ET}_i$ and a root 
of $\widetilde{ET}_i$ 
using $O(\log n)$ space.
Summing over all vertices of all pieces, we get a space 
usage bounded by $O(\pietot \log n)$.

To allow listing all active vertices in a~piece for each 
piece $\piece$ we maintain a set (represented as a balanced 
tree) containing all active colors of vertices in piece $\piece$.
After each update operation we update all sets in 
$O(\pienum \log n)$ time.
To list all active vertices in $\piece$ we iterate over active 
colors in $\piece$ and list all their vertices in $\piece$.

Lastly, we note that we need to access elements of lists 
$L(y)$ in the roots of all ET-trees in constant time.
In order to achieve that, after some ET-tree is updated, 
we build a deterministic dictionary~\cite{Hagerup200169} in 
$O(\pienum \log \pienum) = O(\pienum \log n)$ time.
\end{proof}

In the next two lemmas we show that our template for distance oracles
is sufficient to implement the schemes of Section~\ref{sec:edge_replacements}.

\begin{lemma}
\label{lem:generic_incremental_oracle}
Assume we are given an $\alpha$-approximate generic 
distance oracle for a weighted graph $G=(V,E,\dlug)$.
Let $n=|V|$. One can implement the incremental scheme 
in $O(\pormax \log n)$ expected amortized time per update
and $O(\piemaxvass \piemax + \pormax)$ 
expected time per 
query operation.
The data structure requires $O(\portot)$ space and can 
be initialized in $O(\portot \log n)$ time.

The algorithm can be derandomized at the cost of increase
to $O(\pormax \log n \log\log n)$ amortized time per update
and $O(\piemaxvass \piemax + \pormax \log \log n)$ worst case time per query operation.
\end{lemma}

\begin{proof} 
According to Lemma~\ref{lem:generic_incremental},
there is a data structure that can maintain heaps $H_p$
with distances $N^p_i$ for every portal $p \in \porset$ 
in total expected time 
$O(\acttot \log n)$ (or $O(\acttot \log n \log \log n)$ deterministic time), using $O(\portot)$ space and
$O(\portot \log n)$ time for initialization. 
The incremental oracle starts with all vertices
assigned to distinct inactive colors. Hence,
the number of update operations is at least $|A|$,
where $A$ is the set of vertices active at the end.
Clearly, $\acttot \log n = O(|A| \pormax \log n)$ 
and hence the amortized time per update 
is $O(\pormax \log n)$.  

Now we move on to handling distance queries. 
Note that in the proof of 
Lemma~\ref{lem:generic_incremental} we iterated
through all vertices that changed color upon every
$\opermerge$ operation. So it is possible to maintain a flag
in each vertex which allows checking the color of a vertex
in constant time. 

Consider $\operdistance(v,i)$ query. In order
to find a piece distance from $v$ to the closest
vertex of color $i$, we run
the linear time single source shortest paths algorithm
for planar graphs \cite{Henzinger-97} in every
piece in $\pieces(v)$ with a source at $v$.
We then check the colors of all the vertices 
visible from $v$
and choose a minimum
distance among the vertices of color $i$.
In order to find a portal distance we iterate
through $\portals(v)$ and compute 
$\min_{p \in \portals(v)} (D_{p,v}+N^p_i)$.
Recall that the dynamic dictionary
we use allows portals to access values $N^p_i$
in expected constant time (or $O(\log \log n)$ deterministic).
Therefore answering $\operdistance(v,i)$ query
takes expected $O(\piemaxvass \piemax + \pormax)$ time (or $O(\piemaxvass \piemax + \pormax \log \log n)$ deterministic time).

Consider now $\opernearest(v,k)$ query. In order 
to find the shortest piece distance from $v$ to some active color,
we again iterate through all vertices that are piece-visible from $v$,
and among these choose the closest to $v$ active vertex.
To find the shortest portal distance, we compute
$\min_{p \in \portals(v)} (D_{p,v}+top(H_p))$,
where $top(H_p)$ is the minimum distance stored on $H_p$.
We take the minimum over piece and portal distance,
and based on the vertex associated with this distance 
(recall that on heaps $H_p$ we store pairs $(N^p_i,v)$)
we decide which active color is the closest to $v$.
This takes $O(\piemaxvass \piemax + \pormax)$ time.
Next, we repeat this procedure $k-1$ times,
and in $l$'th iteration we ignore the distances to $l-1$ closest colors we found already.
So if at $l$'th iteration $i_1,\dots,i_{l-1}$  are the first $l-1$ closest colors, then we look in this iteration for an active vertex in set $\bigcup_i C_i \setminus (C_{i_1} \cup \dots \cup C_{i_{l-1}})$.
To that end, we iterate through vertices that are piece-visible from $v$ and find the closest to $v$ vertex of a color we are interested in.
Then we iterate through portals and for every $p \in \portals(v)$ we find $l$ smallest elements in a heap $H_p$.
Since $l \leq k = O(1)$, this can be done in constant time.
From these $l$ elements, we select the minimum element that does not correspond to one of  $i_1,\dots,i_{l-1}$.
The procedure is repeated $k$ times and, if $k = O(1)$, every step takes as much time as finding the nearest active vertex.
Hence the total time to answer $\opernearest(v,k)$ is $O(\piemaxvass \piemax + \pormax)$ if $k$ is constant.
\end{proof}

\begin{lemma}
\label{lem:generic_full_oracle}
Assume we are given an $\alpha$-approximate generic distance oracle for 
$G=(V,E,\dlug)$ and let $n=|V|$.
One can implement the fully dynamic scheme in 
$O((\pornum+\pienum) \log n)$ 
worst case time per update
and $O(\piemaxvass \piemax + \pormax)$ worst case time per query operation. 
The data structure requires $O((\portot+\pietot) \log n)$ space and can be initialized 
in $O((\portot+\pietot) \log n)$ time.
\end{lemma}

\begin{proof}
Here we employ the data structure of Lemma~\ref{lem:generic_full} to maintain
heaps $H_p$ and for a constant time access to $N^p_i$ for each portal
$p \in \porset$. Recall, that this data structure requires 
$O(\portot \log n)$ space and the same initialization time. It allows
updates in $O(\pornum \log n)$ worst case time.

Hence, the only thing left to do is to handle distance queries.
This is done in the same manner as in the proof of
Lemma~\ref{lem:generic_incremental_oracle} with one difference.
The update operations do not iterate through all vertices that change
color, so we cannot maintain a color flag in each vertex so that it 
knows what color it is. We could access the color of a vertex via 
ET-tree $ET_i$ containing it, but that takes $O(\log n)$ time per vertex.
Lemma~\ref{lem:generic_listing} comes in useful here.
We maintain a structure which requires $O(\pietot \log n)$ space
and the same initialization time, and can be updated in $O(\pienum \log n)$
worst case time.
This structure allows us to list all vertices active
colors within a given piece. 
Recall that in order to answer distance queries we first find the
related piece distance and compare it with a portal distance.
In order to find the piece distance, 
for each piece $\piece \in \pieces(v)$ we  
list all active vertices in this piece and mark 
their colors in some global array. 
We run the linear time single source shortest paths algorithm~\cite{Henzinger-97} on $\piece$
with the source at $v$ in order to find the distances from $v$ to vertices in $\piece$.
Then we iterate through vertices $w \in \piece$, and knowing their color from the global array, we choose the closest vertex we are interested in.
This takes linear time per piece by Lemma~\ref{lem:generic_listing}.
Finding the portal distances is essentially the same as in the 
proof of 
Lemma~\ref{lem:generic_incremental_oracle}, with one
small difference: the heaps $H_p$ 
can be easily modified to store triples $(N^p_i,w,i)$,
where $w$ is the vertex corresponding to distance $N^p_i$ and $i$
is its color, so that when we know the colors of vertices in the heaps $H_p$. The Query time of $O(\piemaxvass \piemax + \pormax)$ 
follows.
\end{proof}

\begin{corollary}\label{cor:generic_full_oracle}
Assume we are given an $\alpha$-approximate generic distance oracle 
for $G=(V,E,\dlug)$ and let $n=|V|$.
One can implement the fully dynamic scheme in 
$O(\pornum \log n)$ 
worst case time per update
and $O(\piemaxvass \piemax \log n + \pormax)$ worst case time per query 
operation. 
The data structure requires $O(\portot \log n)$ space and can be 
initialized 
in $O(\portot \log n)$ time. 
\end{corollary}
\begin{proof}
Instead of using structure of Lemma~\ref{lem:generic_listing},
we pay an additional logarithmic factor per distance query,
as while iterating through a vertex seen by the queried vertex $v$,
we additionally check its color using the ET-trees.
\end{proof}

\begin{table}
\begin{center}
\begin{tabular}{|l|l|l|l|l|}
\hline
 &  Initialization & Query & Update\\
\hline
Incremental &   $O(\portot \log n)$ & 
\begin{tabular}{@{}l@{}}
$O(\piemaxvass \piemax + \pormax)$ exp. \\ $O(\piemaxvass \piemax + \pormax\log\log n)$ det.\end{tabular}
& $O(\pormax \log n)$\\
Fully dynamic 1 &  linear & $O(\piemaxvass \piemax + \pormax)$ & $O((\pornum+\pienum) \log n)$ \\
Fully dynamic 2 &  linear & $O(\piemaxvass \piemax \log n + \pormax)$ & $O(\pornum \log n)$ \\
\hline
\end{tabular}
\end{center}
\caption{Summary of our generic constructions of vertex-color oracles.}\label{tab:oracles_summary}
\end{table} 

\subsection{3-approximate oracles for general graphs}

In the previous section we showed how to efficiently implement
the incremental and fully dynamic vertex-color distance oracle scheme.
These constructions were based on generic oracles, which provide
an choice and assignment of portals and pieces to vertices in the graph and
store partial information about the distances in the graph. In
other words, we showed how to implement both schemes given 
portals and pieces.

Here we show how to choose portals and pieces, i.e., how to construct
generic oracles. In principle we want to keep the right balance between
the parameters introduced in the previous section in order to obtain the 
best running times. 

The first generic oracle we present here is a $3$-approximate generic oracle
for general weighted graphs. It is based on ideas by Thorup and 
Zwick~\cite{Thorup05}.
In their construction, however, the procedure that measures distance 
between $u$ and $v$ may give different results when the two vertices are 
swapped. This does not comply with Definition~\ref{def:oracle}, which
requires symmetry in answering queries, that is, if the portal distance between two vertices $u$ and $w$ is finite, it has to be the same in $\pieces(u)$ and $\pieces(w)$.
Also, our data structures measure distance 
between some vertex and the nearest vertex of some given color.
We extend the construction, so that the oracle works on portal and 
piece distances, which are invariant under swapping vertices.
Some additional properties that we obtain are summarized in 
Section~\ref{sec:additional-properties}.

\begin{lemma}
\label{lem:generic_general_full}
Let $G = (V,E,\dlug)$ be a weighted graph, $n = |V|$ and $m = |E|$.
In $O(\sqrt{n}(m + n \log n))$ expected time we can build a $3$-approximate 
generic distance oracle for $G$ that uses expected $O(n\sqrt{n})$ space.
The set of portals $\porset \subseteq V$ of size $O(\sqrt{n})$
is the same for every vertex, i.e., $\portals(v) = \porset$ for every $v \in V$.
Moreover, for each $v \in V$, $\pieces(v)$ contains a single graph that is 
a star on $O(\sqrt{n})$ vertices (in expectation), whose center is $v$.
The parameters of the oracle are as in Table \ref{tab:general_oracle_param}.
\end{lemma}

\begin{proof}
We choose $\porset \subseteq V$ uniformly at random with
probability $1/\sqrt{n}$, so in expectation it is a subset of $V$ 
of size $\sqrt{n}$.
For each $v \in V$ we set $\portals(v) = \porset$, so the number of portals is $\pornum = O(\sqrt{n})$ in expectation
and the expected number of portals assigned to a vertex
is $\pormax = O(\sqrt{n})$. Also, the sum of sizes of all clusters is $\portot = O(n \sqrt{n})$ in expectation.
For each $w \in V$ we define $p_w$ to be the portal which is closest 
to $w$, i.e. $\dist(w, p_w) = \min_{p \in \porset} \dist(w, p)$.
By $B_w$ we denote the set of vertices that are closer to $w$ than any 
portal, i.e. $B_w = \{ v \in V : \dist(w, v) < \dist(w, p_w) \}$.
In principle our goal is to create a single piece per vertex (i.e., $\piemaxvass=1$).
A star $w \times (B_w \setminus \{ w \})$ satisfies all the requirements but the symmetry, that is the piece distance between $u$ and $w$ may be different when measured in $\pieces(u)$ from the one given by $\pieces(w)$.
In order to ensure the symmetry, we reduce $B_w$ to 
$B^\text{red}_w \subseteq B_w$.
The reduction, for every vertex $w$, 
removes from $B_w$ vertices $u \in B_w$ such that $w \notin B_u$.
For every $w \in V$ we declare $\pieces(w)$ to contain a single graph $r$
with the vertex set $V(r)=B^\text{red}_w$, the edge set $E(r)=\{uw : u \in B^\text{red}_w \setminus \{ w \}\}$ 
and the weights on the edges $\dlug_r(w,u)=\dist_G(w,u)$ for every $u \in B^\text{red}_w \setminus \{ w \}$.
As shown in~\cite{Thorup05}, the expected size of each $B_w$, and consequently of each piece $\piemax = O(\sqrt{n})$.

To build the oracle we pick $\porset$, compute shortest path distances from 
every vertex of $\porset$ using Dijkstra's algorithm, and store the 
distances.
To obtain distances from $w$ to vertices of $B^\text{red}_w$, we also 
run a modified version of Dijkstra's 
algorithm starting from each $w \in V$ to find shortest distances 
inside $B_w$.
We break the execution of the algorithm as soon as we visit a vertex 
from $\porset$,\footnote{This modification is described in detail in~\cite{Thorup05}.} 
and store this vertex as $p_w$.

The total time of running shortest paths algorithms from all portals 
is $|\porset| \cdot O(m + n \log n) = O(\sqrt{n}(m + n\log n) )$.
Computing the distances for sets $B_w$ requires 
$n \cdot O(m / \sqrt{n} + n \log n/ \sqrt{n}) = O(\sqrt{n}(m + n\log n))$ 
expected time.
To store resulting distances and graphs we require 
$O(n \sqrt{n})$ expected space.

Let us now prove that we obtain a generic oracle.
Consider two vertices $u$ and $w$.
If $u \in B^\text{red}_w$, then, by the construction, $w \in  B^\text{red}_u$ and the shortest distance between them can be measured exactly both by inspecting $\pieces(u)$ and $\pieces(w)$.
Now, assume $u \notin B^\text{red}_w$. This means that $u \notin B_w$ or $w \notin B_u$.
Without loss of generality, assume $u \notin B_w$.
This implies $\dist(u, w) \geq \dist(w, p_w)$.
Then, $\dist(w, p_w) + \dist(p_w, u) \leq \dist(w, p_w) + (\dist(p_w, w) + \dist(w, u)) \leq 3\dist(u, w)$.
Obviously $p_w \in \portals(v) \cap \portals(w)$.
\end{proof}

\begin{table}
\begin{center}
\begin{tabular}{|l|l|l|l|l|l|l|}
\hline
parameter & $\portot$ & $\piemaxvass$ & $\piemax$ & $\pornum$ & $\pormax$ \\
\hline
value & $O(n \sqrt{n})$ & 1 & $O(\sqrt{n})$ & $O(\sqrt{n})$ & $O(\sqrt{n})$\\
\hline
\end{tabular}
\end{center}
\caption{Summary of parameters: all parameters but $\piemaxvass$ are in expectation.}\label{tab:general_oracle_param}
\end{table}

We now turn the oracle of Lemma~\ref{lem:generic_general_full} into a vertex-color distance oracle.

\begin{theorem}
\label{thm:distance-oracle-general}
Let $G = (V, E)$ be a graph, $n = |V|$. There exists a $3$-approximate vertex-color distance oracle for $G$, which:
\begin{itemize}
\item may answer $\operdistance(v, i)$ and $\opernearest(v, k)$ queries in 
   $O(\sqrt{n} \log n)$ expected time, provided that $k$ is a constant,
\item handles $\opermerge(i, j, u, v)$, $\opersplit(l, u, v)$, 
  $\operactivate(i)$ and $\operdeactivate(i)$ operations in $O(\sqrt{n} \log n)$ expected time,
\item can be constructed in $O(\sqrt{n}(m + n\log n))$ expected time,
\item uses expected $O(n\sqrt{n} \log n)$ space.
\end{itemize}
\end{theorem}

\begin{proof}
We use Lemma~\ref{lem:generic_general_full} to construct a generic 
$3$-approximate oracle
in $O(\sqrt{n}(m + n \log n))$ expected time using 
expected $O(n\sqrt{n})$ space.
The parameters of this oracle are shown in
Table~\ref{tab:general_oracle_param}.
We plug this parameters into Corollary~\ref{cor:generic_full_oracle}
and obtain what is claimed.
\end{proof}

We also observe that, using the ideas of~\cite{Thorup05}, we can improve the running time at the cost of worse --- but still constant --- approximation ratio.

\begin{theorem}
\label{thm:distance-oracle-incremental-general}
Let $G = (V, E)$ be a graph, $n = |V|$ and $l$ be a parameter. 
There exists a $(2l-1)$-approximate vertex-color distance oracle for $G$, which:
\begin{itemize}
\item answers $\operdistance(v, i)$ and $\opernearest(v, k)$ queries in 
   $O(ln^{1/l})$ expected time, provided that $k$ is a constant,
\item handles $\opermerge(i, j)$ and  
  $\operactivate(i)$ operations in $O(ln^{1/l} \log n)$ amortized expected time,
\item can be constructed in $O(lmn^{1/l}+ln^{1+1/l} \log n)$ expected time,
\item uses expected $O(ln^{1+1/l})$ space.
\end{itemize}
\end{theorem}

\begin{proof}
We again follow the ideas of~\cite{Thorup05}.
In particular the construction is the same as in Lemma~\ref{lem:emulator_construction}.
In $O(lmn^{1/l})$
we compute for each $v \in V$ a set $B(v)$ called a bunch.
The expected size of each bunch is $O(ln^{1/l})$.
From the query algorithm in~\cite{Thorup05} it follows that for any two $u, v \in V$ there 
exists $w \in B(u) \cap B(v)$
such that $\dist(u, w) + \dist(w, v) \leq (2l-1)\dist(u, v)$.
For every vertex $v$, we set $\portals(v)=B(v)$ and $\pieces(v)=\emptyset$.
We obtain the expected total number of portals $\portot = O(kn^{1+1/l})$ and the expected number
of portals per vertex $\pormax=O(ln^{1/l})$. 
The number of pieces assigned to a vertex is $\piemaxvass=0$.
By Lemma~\ref{lem:generic_incremental_oracle} (see Table~\ref{tab:oracles_summary}),
we obtain an incremental oracle using total 
$O(ln^{1+1/l})$ expected space, which is initialized in $O(lmn^{1/l}+ln^{1+1/l} \log n)$
expected time,
handles updates in $O(ln^{1/l} \log n)$ amortized expected time and 
handles queries in $O(ln^{1/l})$ expected time. 
\end{proof}

\subsection{$(1+\eps)$-approximate oracles for planar graphs}
In this section we show how to employ generic oracle constructions from 
Section~\ref{sec:generic-oracles} in order to obtain dynamic
$(1+\eps)$-approximate vertex-color distance oracles for planar graphs.
We are able to obtain a better approximation and, for the incremental variant,
better running times, due to planar separators. The separators are used to construct small 
sets of portals. 
We follow here the construction of Thorup~\cite{Thorup04}, 
extending it slightly in order to support vertex colors.

The base building block of the construction is the following:
\begin{lemma}
\label{lem:pathoracle}
Let $G = (V,E)$ be a planar graph, $Q$ be a shortest path in $G$ of length at most $\alpha$, $n = |V|$ and $\eps > 0$.
There exists a set $P \subseteq V$ of size $O(\eps^{-1})$ such that the following holds.
Let $u, v \in V$.
Assume that the shortest path between $u$ and $v$ crosses $Q$ and $\dist_G(u,v) \in [\frac{\alpha}{2}, \alpha]$.
Then, there exists $p \in P$ such that $\dist_G(u,v) \leq \dist_G(u, p) + \dist_G(p,v) \leq (1+\eps)\dist_G(u,v)$.
The set $P$ can be computed in $O(n)$ time.
\end{lemma}

\begin{proof}
The set $P$ is constructed by putting portals on the path $Q$.
We let the first portal be the first vertex of $Q$, and then put portals one by one, each at a distance $\frac{\eps}{4} \dlug(Q)$ from the previous one.
If this would require us to put a portal on an edge $e$ of $Q$, we put two portals in the two endpoints of $e$.
Clearly, $|P| = O(\eps^{-1})$.

It remains to show that the oracle correctly approximates distance between $u$ and $v$.
Let $X_{v,u}$ be the shortest path connecting $u$ to $v$.
Assume that it crosses $Q$ in vertex $q \in G$.
Let $p_q \in P$ be the portal closest to $q$, that is, $\dist(q,p_q) \leq \frac{\eps}{4} \dlug(Q)$.
Then
$$\dist(u, v) \leq \dist(v,p_q) + \dist(p_q,u) \leq \dlug(X_{v,u}) + \frac{\eps}2\dlug(Q) \leq  \dlug(X_{v,u}) + \frac{\eps}2 \alpha \leq (1+\eps)\dlug(X_{v,u}),$$
as requested.
\end{proof}

We now show how to use the above Lemma to construct an oracle that may answer queries about any two vertices in the graph.
This construction was done by Thorup~\cite{Thorup04}, but we need to extend it slightly in order to support vertex colors.

First, we fix $\alpha$ and show how to build a generic oracle that finds all distances in range $[\frac{\alpha}{2}, \alpha]$.
The construction begins by replacing $G$ with a set of graphs $G_1^{\alpha} \dots G_k^{\alpha}$ of diameter $O(\alpha)$, which carry all necessary distance information, and then building a recursive subdivision for each $G_i^{\alpha}$.
The division is obtained by using paths of length $\leq \alpha$ as separators.
Consider a recursive call for a subgraph $G'$ and one of the separator paths $Q$.
We run the construction of Lemma~\ref{lem:pathoracle} on $G'$ and $Q$ and then add all vertices of $P$ to the portal sets of all vertices of $G'$.

Finally, we repeat the construction for $\log D$ different values of $\alpha$, so that the intervals $[\frac{\alpha}2, \alpha]$ altogether cover the entire range of possible distances.

\begin{lemma}[\cite{Thorup04}]
\label{lem:layers}
Given a planar $G$ and a scale $\alpha$ we can construct a series of graphs $G_1^{\alpha}, \ldots, G_k^{\alpha}$ satisfying the following properties:
\begin{itemize}
 \item the total number of vertices and edges in all $G_i^{\alpha}$ is linear in the size of $G$,
 \item $G_i^{\alpha}$ is a minor of $G$ (in particular, every $G_i^{\alpha}$ is planar),
 \item each $G_i^{\alpha}$ has exactly one vertex obtained by contractions, which we call a \emph{supervertex},
 \item each $G_i^{\alpha}$ has a spanning tree $T_i^{\alpha}$ of depth at most $3\alpha$ rooted in the supervertex,
 \item each vertex $v$ has index $\tau (v)$ such that a vertex $w$ is at distance $d \leq \alpha$ from $v$ in $G$ if and only if $d$ is the shortest distance from $v$ to $w$ in $G_{\tau(v)}^{\alpha}$;
       this also holds after removing the supervertex,
 \item each vertex $v$ is in exactly three graphs: $G_{\tau(v)-1}^{\alpha}$, $G_{\tau(v)}^{\alpha}$ and $G_{\tau(v)+1}^{\alpha}$
\end{itemize}
 \end{lemma}

\begin{proof}[Proof sketch] For the full proof we refer to~\cite{Thorup04}.
We construct a partition of $V(G)$ into layers $L_i$.
Let $L_{<i}=\bigcup_{j = 0}^{i-1} L_j$.
We arbitrarily choose vertex $v_0$ and let $L_0 = \lbrace v_0 \rbrace$.
Layer $L_{i+1}$ contains new vertices reachable from vertices of $L_i$ within the distance of $\alpha$, i.e., all reachable vertices not in previous layers.
To create $G_i^{\alpha}$ we take a subgraph of $G$ induced by $L_{<i+2}$ and contract vertices of $L_{<i-1}$.
For a vertex $v$ in $G$ we let $\tau(v)$ be the index of its layer.
Note that any path starting in $v$ of length at most $\alpha$ is contained in $G_{\tau(v)}^{\alpha}$ and does not cross the contracted supervertex.
\end{proof}

As a consequence of the above, each distance query concerning $v$ can be answered by accessing $G_{\tau(v)}^{\alpha}$, provided that the corresponding distance is at most $\alpha$.
In our construction we also use the following Lemma for finding separators.

\begin{lemma}[Lemma 2.3 in~\cite{Thorup04}] \label{lem:treesep}
 In linear time, given an undirected planar graph $G$ with a rooted spanning tree $T$ and
 nonnegative vertex weights, we can find three vertices $u$, $v$ and $w$ such that each component
 of $G \setminus (V(T(u) \cup T(v) \cup T(w))$ has at most half the weight of $G$, where $T(x)$
 denotes the path from $x$ to the root in $T$.
\end{lemma}

We are now ready to show the construction.

\begin{lemma}
\label{lem:generic_planar_incremental}
Let $G$ be a planar graph, $n = |V|$, $\eps > 0$ and $D$ be the stretch of the metric induced by $G$.
A $(1+\eps)$-approximate generic distance oracle for $G$ can be built in $O(\eps^{-1} n \log n \log D)$ time.
For every $v \in V$, $|\portals(v)| = O(\eps^{-1} \log n \log D)$ and $\pieces(v) = \emptyset$.
\end{lemma}

\begin{proof}
Initially, we set $\portals(v) = \emptyset$ for every $v \in V$.
During the construction, the algorithm adds new portals.

We first show the construction for a fixed value of $\alpha$.
It starts with computing the sequence $G_1^{\alpha}, \ldots, G_k^{\alpha}$ by applying Lemma~\ref{lem:layers}.
Then, we recursively decompose each $G_i^{\alpha}$.

We use Lemma~\ref{lem:treesep} on $G_i^{\alpha}$ and $T_i^{\alpha}$ setting the weight of the supervertex to $0$ and the weights of other vertices to $1$.
Lemma~\ref{lem:treesep} guarantees that $G_i^{\alpha}$ has a separator $Q$ composed of three paths of $T_i^{\alpha}$.
The first step is to represent all distances that correspond to paths crossing $Q$ in a graph obtained from $G_i^{\alpha}$ by removing the supervertex.
Denote this graph by $\tilde{G}_i^{\alpha}$.
Observe that although $G_i^{\alpha}$ is a minor of $G$, $\tilde{G}_i^{\alpha}$ is a \emph{subgraph} of $G$, as we have removed the only vertex obtained by contractions.

We remove the supervertex from $Q$ and obtain a separator $\tilde{Q}$ of $\tilde{G}_i^{\alpha}$.
Note that the three paths that form $Q$ start in the supervertex, so $\tilde{Q}$ also consists of three paths.
Since these paths have length $\leq 3\alpha$, we may split each of them into three paths of length $\leq \alpha$, thus decomposing $\tilde{Q}$ into nine paths $Q_1, \ldots, Q_9$, each of length $\leq \alpha$.

To represent all distances in $\tilde{G}_i^{\alpha}$ crossing $\tilde{Q}$, we use the construction of Lemma~\ref{lem:pathoracle} for each of the nine paths $Q_j$, $j \in \{1 \dots 9 \}$.
In this way for a path $Q_j$ we obtain a set of vertices $P_j$.
We add those vertices to the set $\portals(v)$ for every vertex of $\tilde{G}_i^{\alpha}$.
For every new portal, we compute (and store) distances to every vertex of $\tilde{G}_i^{\alpha}$ using linear time algorithm for shortest paths in planar graphs~\cite{Henzinger-97}.

Next in $G_i^{\alpha}$ we contract separator $Q$ and obtain a planar graph $G'$.
Since the separator $Q$ contains the supervertex of $G_i^{\alpha}$, the contracted $Q$ in $G'$ becomes a supervertex $r'$, and tree $T_i^{\alpha}$ becomes a tree $T'$ rooted at $r'$.
We note that $T'$ is a tree spanning $G'$ rooted in $r'$ and its depth is at most $3\alpha$.
Moreover, removing $r'$ splits $G'$ into at least two components (cut out in $G_i^{\alpha}$ by the separator $Q$), each containing at most half vertices of $V(G')$.
Each of these components has a spanning tree which is a subtree of $T'$ rooted at the supervertex $r'$.
Hence, we can recurse on each of these components (each component is taken together with $r'$).

The construction allows us to approximate distances in range $[\frac{\alpha}{2}, \alpha]$, so we repeat it $\log D$ times for different values of $\alpha$, where $D$ is the stretch of the metric induced by $G$.

We now show that we obtain a $(1+\eps)$-approximate generic oracle.
Consider two vertices $u$ and $v$.
Let $\alpha$ be the scale such that $\dist_G(u,v) \in [\frac{\alpha}2, \alpha]$ .
There exists a graph $G_i^{\alpha}$ such that $\dist_{G_i^{\alpha}}(u,v) = \dist_G(u,v)$.
Consider the shortest path $X_{u,v}$ in $G_i^{\alpha}$ between $u$ and $v$.
In some recursive call of the decomposition algorithm, the path crosses the separator $Q$.
It is easy to see that in that step, we add both to $\portals(u)$ and $\portals(v)$ a vertex $p$ such that $\dist_G(u, p) + \dist_G(p, v) \leq (1+\eps)\dist_G(u,v)$.
Moreover, if we denote by $G'$ the subgraph in the recursion in exactly that step, 
then $G'$ contains $Q'$, so $\dist_{G'}(u, p) = \dist_G(u,p)$ and $\dist_{G'}(p, v) = \dist_{G}(p,v)$.

It remains to bound the running time and the number of portals of each vertex.
Consider a vertex $v$ of $G$.
For each of $O(\log D)$ values of $\alpha$, $v$ belongs to three layers $G_i^{\alpha}$.
Let $n_i=|V(G_i^{\alpha})|$.
In a recursive subdivision of $G_i^{\alpha}$, $v$ can appear once on each of the $O(\log n_i)$ levels.
On a single level, we add $O(\eps^{-1})$ vertices to $\portals(v)$.
Thus, since $\log n_i \leq \log n$, $|\portals(v)| = O(\eps^{-1} \log n \log D)$.
Observe that the distances inside each cluster are computed in linear time, so the initialization time is 
equal to the total size of the clusters.
The lemma follows.
\end{proof}

We combine the incremental oracle construction from
Lemma~\ref{lem:generic_incremental_oracle} (see Table~\ref{tab:oracles_summary})
with the above lemma to obtain the following result.

\begin{theorem}
\label{thm:oracle-planar-incremental}
Let $G = (V,E)$ be a planar graph, $n = |V|$ and $D$ be the stretch of the metric induced by $G$.
For every $\eps > 0$ there exists an $(1+\eps)$-approximate vertex-color distance oracle for $G$, which:
\begin{itemize}
\item answers $\operdistance(v, i)$ and $\opernearest(v, k)$ queries in $O(\eps^{-1} \log n \log D)$ expected time or \\ $O(\eps^{-1} \log n \log \log n \log D)$ deterministic time, 
provided that $k$ is a constant,
\item handles $\opermerge(i,j)$ and $\operactivate(i)$ operations in $O(\eps^{-1} \log^2 n \log D)$ expected amortized time or $O(\eps^{-1} \log^2 n \log \log n \log D)$ deterministic amortized time,
\item can be constructed in $O(\eps^{-1} n \log^2 n \log D)$ time,
\item uses $O(\eps^{-1} n \log n \log D)$ space.
\end{itemize}
\end{theorem}

\begin{proof}
Lemma~\ref{lem:generic_planar_incremental} gives a generic oracle with
$\pormax=\max {|\portals(v)|}=O(\eps^{-1} \log n \log D)$, and hence the total size
of all clusters $\portot=O(\eps^{-1} \log n \log D)$. In addition to that every
vertex is assigned an empty set of pieces so $\piemaxvass=0$.
We employ Lemma~\ref{lem:generic_incremental_oracle} (see Table~\ref{tab:oracles_summary})
to obtain what is claimed. 
%
%
%
\end{proof}

We now describe how to build a fully dynamic vertex-color distance oracle, which also supports $\opersplit$ and $\operdeactivate$ operations.
Its running time depends on two main factors.
By Lemma~\ref{lem:generic_full} the running time of an update operation is proportional to the number of portals, whereas the query time depends on the size of pieces.
Thus, in our construction we try to keep both these values low.

\begin{lemma}
\label{lem:generic_planar_full}
Let $G = (V,E)$ be a planar graph, $n = |V|$ and $D$ be the stretch of the metric induced by $G$.
For every $\rho$ such that $1 \leq \rho \leq n$, in $O(\eps^{-1} n \log n \log D)$ time we can build an $(1+\eps)$-approximate generic distance oracle for $G$.
The total number of portals is $O(\eps^{-1} \frac{n}{\rho} \log D)$, and the total size of all of clusters is $O(\eps^{-1} n \log n \log D)$.
For every $v \in V$, $\pieces(v)$ is a set consisting of $O(\log D)$ planar graphs of size $O(\rho)$.
There are $O(\frac{n}{\rho} \log D)$ distinct pieces in total.
The parameters are summarized in Table~\ref{tab:planar_oracle_param}.
\end{lemma}

\begin{proof}
For the purpose of this construction a graph is called \emph{big} if it contains at least $\frac{\rho}{3}$ vertices and $\emph{small}$ otherwise.
Moreover, a graph with more than $\rho$ vertices is called \emph{too big}.
We aim to alter the construction of a generic oracle from Lemma~\ref{lem:generic_planar_incremental} in such a way that the recursive division of graphs $G_i^{\alpha}$ produces subgraphs which are neither small nor too big.

First, we use Lemma~\ref{lem:layers} to obtain a sequence $G_i^{\alpha}$.
As long as the sequence contains two consecutive small layers $G_i^{\alpha}$, $G_{i+1}^{\alpha}$, we merge them into one layer.
Since the layers are minors of $G$, we first need to remove the supervertex from both of them.
After that, they become subgraphs of $G$, so we can merge them in a natural way.

Note that after this process every small layer is a neighbor of a big layer, so the number of small layers is $O(\frac{n}{\rho})$.
Additionally, each too big layer is not a result of merging smaller layers and satisfies conditions of Lemma~\ref{lem:treesep}.

We now deal with too big layers $G_i^{\alpha}$.
For each such layer, we perform a recursive subdivision similar to the one in the proof of Lemma~\ref{lem:generic_planar_incremental}.
However, we introduce two changes.
First, we assure that at each step, the graph is divided into exactly two subgraphs of roughly even size.
Recall that in the construction, we contract the separator in a currently handled graph $G'$ into a supervertex $r'$, and each
connected component of $G' \setminus r'$ is of size at most $\frac12$ of the size of $G'$.
By standard arguments, we can group the components of $G' \setminus r'$ into
exactly two groups, each of total size at most $\frac23$ of the size of $G'$.
We recurse on the union of the graphs in each group separately.

The second modification to the algorithm from Lemma~\ref{lem:generic_planar_incremental} is that we stop the recursion as soon as we reach a graph which is not too big.
As a result, leaves of the recursive subdivision correspond to big graphs, which are not too big.

Consider a layer $G_i^{\alpha}$ (after the merging step).
If $G_i^{\alpha}$ is not too big we add $G_i^{\alpha}$ to $\pieces(v)$ for every $v \in V(G_i^{\alpha})$ (except for the supervertex).
Otherwise, we decompose the graph recursively and for every separator that we use, we apply Lemma~\ref{lem:pathoracle} and add the obtained set to $\portals(v)$.
Moreover, once we reach a graph $G'$ which is not too big and the recursion stops, we add this graph to $\pieces(v)$ for every $v \in V(G')$.

The entire construction is repeated for $O(\log D)$ values of $\alpha$, so that the intervals $[\frac{\alpha}2, \alpha]$ altogether cover all possible distances.

We now show that we obtain a $(1+\eps)$-approximate generic oracle.
Consider two vertices $u$ and $v$ and the value of $\alpha$ such that $\dist_G(u,v) \in [\frac{\alpha}2, \alpha]$.
There exists a graph $G_i^{\alpha}$ such that $\dist_{G_i^{\alpha}}(u,v) = \dist_G(u,v)$.
Consider the shortest path $X_{u,v}$ in $G_i^{\alpha}$ between $u$ and $v$.
If layer $G_i^{\alpha}$ is not too big then it is added both to $\pieces(u)$ and $\pieces(v)$ and the property follows trivially.
Otherwise, there are two cases to consider.
One possibility is that in some recursive call of the decomposition algorithm, the path crosses the separator $Q$.
In this case we add both to $\portals(u)$ and $\portals(v)$ a vertex $p$ such that $\dist_G(u, p) + \dist_G(p, v) \leq (1+\eps)\dist_G(u,v)$.
Otherwise, we reach a subgraph which is not too big and contains the entire path $X_{u,v}$.
This subgraph is then added to $\pieces(u)$ and $\pieces(v)$.

It remains to bound the size of the oracle.
From the construction it follows that the size of every piece is $\Theta(\rho)$
so the maximal size of a piece is clearly $\piemax \in \Theta(\rho)$.
For every $\alpha$ during the recursive decomposition of $G_i^{\alpha}$ we create pieces that are 
nonoverlapping subgraphs of $G_i^{\alpha}$.
Thus, for a fixed $\alpha$, the sum of sizes of all pieces that we create is $O(n)$.
We use $O(\log D)$ values of $\alpha$, so the total size of all pieces is $\pietot=O(n \log D)$.

Moreover, for a fixed $\alpha$, because all but $O(\frac{n}{\rho})$ pieces have size $\Theta(\rho)$, 
there are $O(\frac{n}{\rho})$ pieces in total. So the total number of pieces
is $\pienum = O(\frac{n}{\rho} \log D)$ as there are $\log D$ values of $\alpha$.
For every layer $G_i^{\alpha}$ and a fixed vertex $v$, at most one piece may be added to $\pieces(v)$.
Each vertex belongs to $O(\log D)$ layers $|\pieces(v)| = O(\log D)$ for every $v \in G$
and hence $\piemaxvass = O(\log D)$.

Let us now bound the total number of portals.
Consider the recursive subdivision algorithm.
We add $O(\eps^{-1})$ new portals every time when we divide a graph into two subgraphs with a separator.
For a fixed $\alpha$, in the end all layers are divided into $O(\frac{n}{\rho})$ subgraphs (which become pieces).
Thus, the number of division steps is also $O(\frac{n}{\rho})$, and they create $O(\eps^{-1}\frac{n}{\rho})$ new portals.
This gives $O(\eps^{-1}\frac{n}{\rho} \log D)$ portals for all values of $\alpha$, in other
words $\pornum = O(\eps^{-1}\frac{n}{\rho} \log D)$.

The total size of clusters is bounded by the total size of all clusters in the construction of 
Lemma~\ref{lem:generic_planar_incremental}, which is $\portot = O( \eps^{-1} n \log n \log D)$.
Also after Lemma~\ref{lem:generic_planar_incremental},
the maximum number of portals assigned to a vertex is $\pormax = O(\eps^{-1} \log n \log D)$.
The lemma follows.

\begin{table}
\begin{center}
\renewcommand{\arraystretch}{1.5}%
\begin{tabular}{|>{\hspace{-3pt}}c<{\hspace{-3pt}}|>{\hspace{-3pt}}c<{\hspace{-3pt}}|>{\hspace{-3pt}}c<{\hspace{-3pt}}|>{\hspace{-3pt}}c<{\hspace{-3pt}}|>{\hspace{-3pt}}c<{\hspace{-3pt}}|>{\hspace{-3pt}}c<{\hspace{-3pt}}|>{\hspace{-3pt}}c<{\hspace{-3pt}}|>{\hspace{-3pt}}c<{\hspace{-3pt}}}
\hline
$\portot$ & $\pormax$ & $\pornum$ & $\pienum$ & $\piemaxvass$ & $\piemax$ & $\pietot$ \\
\hline
$O( \eps^{-1} n \log n \log D)$ & $O(\eps^{-1} \log n \log D)$ & $O(\eps^{-1}\frac{n}{\rho} \log D)$ & $\frac{n}{\rho} \log D$ & $O(\log D)$ & $\Theta(\rho)$ & $O(n \log D)$\\
\hline
\end{tabular}
\end{center}
\caption{Summary of parameters.}\label{tab:planar_oracle_param}
\end{table}
\end{proof}

\begin{theorem}
\label{thm:planar-full-oracle}
Let $G = (V,E)$ be a planar graph, $n = |V|$ and let $D$ be the stretch of the metric induced by $G$.
For every $\rho,\eps > 0$, there exists an $(1+\eps)$-approximate vertex-color distance oracle for $G$, which:
\begin{itemize}
\item may answer $\operdistance(v, i)$ and $\opernearest(v, k)$ queries in $O((\eps^{-1} \log n + \rho) \log D)$ time, provided that $k$ is a constant,
\item handles $\opermerge(i, j, u, v)$, $\opersplit(l, u, v)$, $\operactivate(i)$ and $\operdeactivate(i)$ operations in \\ $O(\eps^{-1} \frac{n}{\rho} \log n \log D)$ time,
\item can be constructed in $O(\eps^{-1} n \log^2 n \log D)$ time,
\item uses $O(\eps^{-1} n \log^2 n \log D)$ space.
\end{itemize}
In addition to that, the oracles admits parameters as in Table~\ref{tab:planar_oracle_param}.
\end{theorem}

\begin{proof}
The theorem follows from combining the oracle of Lemma~\ref{lem:generic_planar_full}
(the parameters are given in Table \ref{tab:planar_oracle_param}) with
Corollary~\ref{cor:generic_full_oracle} (see Fully Dynamic 1 oracle in Table \ref{tab:oracles_summary}).
\end{proof}

\subsection{Additional operations and properties}
\label{sec:additional-properties}
\paragraph{Near-metric spaces}
The generic oracles we build measure distance between $u$ and $v$ by taking the minimum of portal and piece 
distance (see Definition \ref{def:oracle}) between $u$ and $v$.
Let us call this distance an \emph{oracle distance}.
From the construction of our vertex-color distance oracles it is easy to observe that 
they answer $\operdistance$ queries by choosing the minimum oracle distance to the vertex of respective color.
A similar property holds for $\opernearest$ queries.
Roughly speaking, the oracles behave as if they were exact and the real distance between two vertices was 
the oracle distance (though the oracle distance does not satisfy triangle inequality).
This property, which is used in our algorithms, is summarized in the following.

\begin{corollary}
\label{cor:consistency}
Consider one of nearest neighbor oracles from Section~\ref{sec:approx_distance_oracles}.
Assume it provides $\GDapx$-approximate distances between vertices of a graph $G=(V,E, \dlug_G)$.
Then, the oracle yields a $\GDapx$-near metric space $\GD$ that approximates $G$.
In particular:
\begin{itemize}
\item The answer to a $\operdistance(v, i)$ query is equal 
   to $\min \{\dlug_\GD(v,w): w \textrm{ is active and has color }i\}$.
\item The answer to a $\opernearest(v, k)$ query is the $k$-th smallest value of $\operdistance(v, i)$ over 
   all active colors $i$.
\end{itemize}
\end{corollary}

\begin{proof}
Consider a clique $\GD=(V,\binom{V}{2},\dlug_\GD)$, where $\dlug_\GD(u,v)$ is defined
as the minimum over the piece distance and the portal distance between $u$ and $v$
(i.e., the oracle distance).
By Definition \ref{def:oracle} it is a well defined simple graph with finite weights.
This graph certifies the claim stated in the corollary.
\end{proof}

\paragraph{Reconnecting edges.}

In one of our algorithms, at some point we encounter the following problem.
We are given a set of active colors $c_1, \ldots, c_k$ that we need to connect together in a minimum spanning tree fashion.
That is, we would like to compute the MST of a complete graph $G_c$ over colors $c_1, \ldots, c_k$, where the length of an edge connecting two colors is the distance between the nearest vertices of these colors.
Solving this problem efficiently seems hard, even for the case when $k=2$.
However, we are able to show a solution for its simpler variant.
Namely, we define the length of the edge in $G_c$ to be the length of the shortest path that connects two vertices of respective colors and goes through a portal.

\begin{lemma}
\label{lem:portal-mst}
Assume we are given a vertex-color distance oracle with $\pornum$ portals and a set of active colors $c_1, \ldots, c_k$.
Consider a complete graph $G_c$ on the set of colors, such that the length of an edge connecting two colors $c_i$ and $c_j$ is the length of a shortest path that connects vertices $v_i$ and $v_j$ of colors $c_i$ and $c_j$ and goes through a vertex in $portals(v_i) \cap portals(v_j)$.
Then, $\mst(G_c)$ can be computed in $O(k\pornum + k \log k)$ time.
\end{lemma}

\begin{proof}
We first build a graph $G'_c$, and then compute its minimum spanning tree.
$G'_c$ is a subgraph of the graph $G_c$ defined above.
For efficiency reasons it may not contain some edges of $G_c$ that surely do not belong to $\mst(G_c)$.

We iterate through all portals of the oracle and add edges to $G'_c$.
Consider a portal $p$ and assume that vertices of colors $c_1, \ldots, c_{k'}$ are at distances $d_1 \leq \ldots \leq d_s$ from $p$.
In particular, $c_1$ is the color that is nearest from $p$.
The shortest path connecting colors $i$ and $j$ going through $p$ has length $d_i + d_j$.
Observe that we may add to $G'_c$ only edges corresponding to pairs of colors $c_1$, $c_i$ (for $i \geq 2$).
Indeed, if $i, j \geq 2$, then $d_i + d_j \geq d_1 + d_j$ and $d_i + d_j \geq d_i + d_1$.
Thus, any other edge that we could add would be the heaviest on the cycle formed by colors $c_1$, $c_i$ and $c_j$ and as such it is not necessary for computing the MST.

Thus, we end up with a graph with $k$ vertices and $O(k\pornum)$ edges.
We compute its MST with Prim's algorithm in $O(k\pornum + k \log k)$ time.
\maybeqed
\end{proof}

\paragraph{Efficient Steiner tree computation.} We also observe that a vertex-color distance oracle gives us
an efficient way of computing an approximate Steiner tree for a query consisting of an entire set $S$ of terminals.

\begin{lemma}\label{lem:oracle-GD-mst}
Assume we are given a graph $G = (V, E, \dlug_G)$, $n = |V|$ and an instance $\DO$ 
of a vertex-color distance oracle for $G$, which yields a near-metric space $\GD$.
Then, given a set $S \subseteq V(G)$, 
one can compute $\mst(\indu{\GD}{S})$ in $O(|S| (\pormax + \piemaxvass \piemax) \log n)$ time,
where $\pormax$ is the maximum number of portals per vertex, $\piemaxvass$ is
maximum number of pieces assigned to a vertex and $\piemax$ bounds the size of a piece.
\end{lemma}

\begin{proof}
We compute a minimum spanning tree of $S$ in $\GD$ using Prim's algorithm.
We start with $\drzewo$ being an isolated vertex (an arbitrarily chosen element of $S$) and, 
at each step, we choose an edge connecting $S \setminus V(\drzewo)$ and $V(\drzewo)$ in $\GD$ of 
minimum possible weight and add it to $\drzewo$.
To choose such an edge, it suffices to choose the minimum among the portal and piece distance between 
$S \setminus V(\drzewo)$ and $V(\drzewo)$.

First, consider piece distances.
For each $v \in S \setminus V(\drzewo)$ we maintain the piece distances between $v$ and the nearest 
vertex (w.r.t piece distance) of $V(\drzewo)$.
All these piece distances are kept in a heap $H_v$.
In addition to that, we maintain a global heap $H$ with the minimas
of the heaps $H_v$ for all vertices $v \in S \setminus V(\drzewo)$.
The minimum piece distance between $V(\drzewo)$ and $S \setminus V(\drzewo)$
is the minimum value on $H$.
If a vertex $u$ is added to $V(\drzewo)$, to update this information it suffices to run Dijkstra's 
shortest-path algorithm from $u$ in each of $\pieces(u)$, update information
in heaps $H_v$ for all vertices of $S \setminus V(\drzewo)$ that were found,
each time updating the value on $H$ as well.
For all vertices of $S$ this requires $O(|S| \piemaxvass \piemax \log n)$ time.

To find the minimum portal distance, for each portal $p$ of a vertex in $S$ we consider $\cluster(p)$ and 
store the distance $d^p_\drzewo$ to the nearest vertex of $V(\drzewo)$ and a heap of distances to vertices 
of $S \setminus V(\drzewo)$ (both in $\cluster(p)$).
The portal distance is the minimum over portals of the sum of $d^p_\drzewo$ and the minimum on
this portals heap.
We again maintain a global heap which has access to minimum portal distance.
Thus, for every vertex $v \in S$, we need to add distances to all $\portals(v)$ to a heap 
in $O(|S|\cdot|\portals(v)| \log n)$ time, each time updating the global heap as well.
Whenever we connect a vertex $u$ to the tree $\drzewo$, we update the heaps
of all its portals and the global heap. This again takes a total $O(|S|\cdot|\portals(v)| \log n)$
time. 
The time for removing the elements is at most as big as the time for inserting them.
The lemma follows.
\end{proof}

By Lemma~\ref{lem:oracle-GD-mst}, we obtain the following two corollaries, depending on which
oracle we use.

\begin{corollary}\label{cor:st-query}
Assume we are given a graph $G = (V, E, \dlug_G)$, $n = |V|$.
After $O(\sqrt{n}(m + n\log n))$ preprocessing, using $O(n\sqrt{n} \log n)$ space,
we can repeatedly compute a $3$-approximation of $\mst(\indu{\mclo{G}}{S})$, and consequently a $6$-approximation of the minimum Steiner tree spanning $S$ in $G$,  in $O(|S| \sqrt{n} \log n)$ time
for any subset $S \subseteq V$.
\end{corollary}
\begin{proof}
We employ as $\DO$ in Lemma~\ref{lem:oracle-GD-mst} the fully dynamic oracle of 
Theorem~\ref{thm:distance-oracle-general}.
By Theorem~\ref{thm:distance-oracle-general},
$\DO$ can be constructed in $O(\sqrt{n}(m + n\log n))$ expected time
and uses expected $O(n\sqrt{n} \log n)$ space.
By Lemma~\ref{lem:oracle-GD-mst},
we can query for $\mst(\indu{\GD}{S})$ in $O(|S| \sqrt{n} \log n)$ time.
By Lemma~\ref{lem:two-apx-st}, it is a $6$-approximation of the Steiner tree.
\end{proof}

\begin{corollary}\label{cor:st-query-planar}
Assume we are given a planar graph $G = (V, E, \dlug_G)$, $n = |V|$.
After $O(\eps^{-1} n \log^2 n \log D)$ preprocessing, using $O(\eps^{-1} n \log n \log D)$ space,
we can repeatedly compute a $(1+\eps)$-approximation 
of $\mst(\indu{\mclo{G}}{S})$, and consequently a $2+\eps$-approximation of the minimum Steiner tree spanning $S$ in $G$, in $O(|S|\eps^{-1} \log^2 n \log D)$ time
for any subset $S \subseteq V$.
\end{corollary}
\begin{proof}
We employ as $\DO$ in Lemma~\ref{lem:oracle-GD-mst} the fully dynamic oracle of 
Theorem~\ref{thm:oracle-planar-incremental}.
By Theorem~\ref{thm:oracle-planar-incremental},
$\DO$ can be constructed in $O(\eps^{-1} n \log^2 n \log D)$ time
and uses $O(\eps^{-1} n \log n \log D)$ space.
By Theorem~\ref{thm:oracle-planar-incremental},
$\DO$ assigns at maximum
$\pormax=\max {|\portals(v)|}=O(\eps^{-1} \log n \log D)$
portals to a vertex and the number of pieces assigned to a vertex is $\piemaxvass=0$.
By Lemma~\ref{lem:oracle-GD-mst},
we can query for $\mst(\indu{\GD}{S})$ in $O(|S| \eps^{-1} \log^2 n \log D)$ time.
By Lemma~\ref{lem:two-apx-st}, it is $2(1+\eps)$-approximation of the Steiner tree, which can be turned into $(2+\eps)$-approximation by using $\eps' = \eps/2$ instead of $\eps$.
\end{proof}

\paragraph{Faster implementation of the Imase-Waxman incremental algorithm.}
Imase and Waxman~\cite{ImaseW91} proved that the natural online
algorithm for Steiner tree --- the one that attaches the new terminal
to the closest previous one --- achieves $O(\log n)$ approximation guarantee.
We remark here that our oracles give a different way to obtain this guarantee;
our implementation is faster if the number of terminals is large in terms of the size of host 
graph.

\begin{theorem}\label{thm:imase-waxman-implementation}
Assume we are given a graph $G = (V,E,\dlug_G)$, $n = |V|$.
Then there exists an algorithm that maintains an online Steiner tree upon a sequence
of terminal addition (i.e., the algorithm needs to maintain a Steiner tree upon each addition,
and is not allowed to alter the edges that are already in the tree)
with $O(\log n)$ approximation guarantee, using expected time $O(\sqrt{n} \log n)$ per terminal
addition.
The algorithm can be initialized in $O(\sqrt{n}(m + n \log n))$ time and uses $O(n \sqrt{n} \log n)$
space.
\end{theorem}
\begin{proof}
We use the vertex-color distance oracle of Theorem~\ref{thm:distance-oracle-general}.
In the oracle the current set of terminals is a single color, and only this color is active.
To add a new terminal, it suffices to issue one $\operdistance$ and one $\opermerge$ operation.
The new terminal is attached to not necessarily the closest previous terminal, but approximately
closest.
\end{proof}

\newcommand{\laycnt}{h}
\section{Sublinear time dynamic algorithms}\label{sec:algorithms}
We propose here dynamic algorithms for each of the dynamic scenarios we considered. The algorithms are obtained by implementing the online algorithms of Section~\ref{sec:edge_replacements} using vertex-color distance oracles from Section~\ref{sec:approx_distance_oracles}.

\newcommand{\fulloracle}[1]{\textsc{Full}\ensuremath{(#1)}}

\subsection{Decremental algorithm}
\label{sec:decremental}
We now consider the decremental scenario, in which the goal is to maintain an approximate Steiner tree subject to terminal deletions.
The graph we work with is accessed with a fully dynamic vertex-color distance oracle $\DO$ (see Section~\ref{sec:approx_distance_oracles}), which yields a $\GDapx$-near metric space $\GD$ that approximates $G$ (see Corollary~\ref{cor:consistency}).
However, our algorithm not only uses the operations provided by the oracle, but also the sets of portals and pieces of each vertex.
In particular, the algorithm uses the concepts of portal and piece distances, which are defined in Section~\ref{sec:generic-oracles}.

We use the decremental scheme of Section \ref{ss:dec-scheme_full}, so for a fixed $\eps > 0$ we plan to maintain a $(1+\eps)$-approximation of $\mst(\indu{\GD}{S})$.
Following the decremental scheme, we set $\degth = 1 + \lceil \eps^{-1} \rceil$.

Due to decremental scheme, when a vertex is removed from the terminal set, it is not necessary to remove it from the maintained tree $\drzewo$ unless its degree in $\drzewo$ is at most $\degth$.
Thus an update may cause $\drzewo$ to split into at most $\degth$ trees.
The main challenge of this section is to find in an efficient way the set of edges of minimum total cost that reconnects them.

In order to maintain the desired tree $\drzewo$ we use the dynamic MSF algorithm (see Theorem~\ref{thm:Thorup_full}).
In what follows we construct a graph $H$ on the vertex set $V$ (the vertex set of both input graph $G$ and $\GD$) with the property that there is a connected component $H_0$ of $H$ containing all current terminals.
We declare $\drzewo$ to be $\mst(H_0)$ (maintained by the dynamic MSF algorithm).
All other connected components of $H$ will be isolated vertices.
Our goal now is to maintain $H$ by adding and removing edges in such a way that maintaining the MSF of $H$ with the dynamic MSF algorithm does in fact implement the decremental scheme of Section~\ref{ss:dec-scheme_full} on $\GD$.

The simplest approach would be to add to $H$, for every two $u, w \in S$, an edge $uw$ of length $\dlug_\GD(u,w)$.
This approach, however, would obviously be inefficient.
Instead of that, we use the structure of the vertex-color distance oracle.
For every two $u, w \in V(T)$ such that $u$ is piece-visible from $w$, we add to $H$ an edge $uw$ of length being equal to the piece distance.
Moreover, we will be adding some edges corresponding to portal distances, but they will be chosen in a careful way to ensure that the number of such edges is low.

Our algorithm uses a single instance $\DO$ of the fully dynamic vertex-color distance oracle.
Initially all vertices are inactive and have distinct colors.
Denote the terminal set by $S$.
Using Lemma~\ref{lem:oracle-GD-mst} we compute the initial minimum spanning tree $\drzewo=\mst(\indu{\GD}{S})$ and modify $\DO$ so that $\DO$ contains one active color, whose corresponding tree is equal to $\drzewo$.
Other vertices remain inactive.

We initialize the edge set of $H$ to be the set of all edges of the initial $\drzewo$ and all edges between terminals corresponding to mutually piece-visible vertices.
Note that we may break ties of weights in $H$ in such a way that $H$ has a unique minimum spanning forest.
Consequently, the forest computed by the dynamic MSF algorithm contains the tree $T$.

We update the vertex-color distance oracle $\DO$ to maintain the following invariant: the colors of $\DO$ correspond to the connected components of the graph $H$, and the only active color is associated with a tree equal to $T$.
Hence, any change to the spanning forest of $H$ results in a constant number of modifications to $\DO$.

According to the decremental scheme of Section~\ref{ss:dec-scheme_full}, when a vertex of degree greater than $\degth$ is deleted, we simply mark it as a non-terminal.
Otherwise, if the procedure $\remove(v)$ requires us to delete a vertex of degree at most $\degth$ from the currently maintained tree, we update $H$ and the distance oracle $\DO$, as follows.
\begin{enumerate}
\item Split the active color $V(\drzewo)$ in $\DO$ into $s+1$ pieces being $\{v\}$ and the $s$ connected components of $\drzewo \setminus \{v\}$, and then deactivate $\{v\}$.\label{step:dec:1}
\item Compute the set of reconnecting edges that correspond to portal distances (using Lemma~\ref{lem:portal-mst}) and add these edges to $H$.\label{step:dec:2}
\item Remove all edges incident to $v$ from $H$.\label{step:dec:3}
\item At this point, the dynamic MSF algorithm may update the maintained tree.
    Update the vertex-color distance oracle $\DO$, so that it reflects these changes of the MSF of $H$.
\end{enumerate}

We now show that our algorithm indeed implements the decremental scheme of Section~\ref{ss:dec-scheme_full}.

\begin{lemma}\label{lem:dec:correct_full}
After each removal of a vertex from $\drzewo$, $\drzewo$ is a minimum spanning tree of $\indu{\GD}{V(\drzewo)}$ and $\drzewo$ implements the decremental scheme of Section \ref{ss:dec-scheme_full}.
\end{lemma}
\begin{proof}
Clearly, at the initialization step for a given terminal set $S$, in $H$ there is a component $H_0$ containing $\drzewo = \mst(\indu{\GD}{S})$ and all vertices of $V \setminus S$ are isolated.
Moreover, the minimum spanning tree of $H_0$ computed by the dynamic MSF algorithm equals $T$.

After initialization, for each pair of mutually piece-visible vertices of $V(T)$, the graph $H$ contains the corresponding edge, as we added this edge explicitly to $H$.
Clearly, this invariant is satisfied after every terminal deletion.

Let us now analyze the case when $\remove$ removes a vertex $v$ of degree $s \leq \degth$.
Denote by $T_1, \ldots, T_s$ the trees which $T$ decomposes into.
We use Lemma~\ref{lem:portal-mst} to compute a set of edges corresponding to portal distances that have minimum weight and connect trees $T_1, \ldots, T_s$.
These edges are added to $H$ in Step~\ref{step:dec:2}.

Assume that before the operation of removing a vertex $v$ from $\drzewo$, the tree $\drzewo$ is a minimum spanning tree of $\indu{\GD}{V(\drzewo)}$.
Then, the addition of the reconnecting edges in Step~\ref{step:dec:2} does not modify $\drzewo$, as we use a tie breaking strategy for edges' weights and assume that out of two edges of the same weight, the lighter is the one that was added earlier (see Section \ref{sec:preliminaries}).
These edges may be used when the edges incident to $v$ are deleted in Step~\ref{step:dec:3}.

We now prove that $H$ contains a set of edges $F^v$ of minimum total cost that reconnects trees $T_1, \ldots, T_s$ into a tree $T'$.
If $s = 1$, the claim follows trivially (the set $F_v$ is empty).
Otherwise, consider an edge $e$ of $F^v$.
Removing $e$ from $T'$ yields two trees $T'_1$ and $T'_2$.
We show that $H$ contains an edge connecting $T'_1$ with $T'_2$ of cost equal to $e$.
If $\dlug_\GD(e)$ is equal to the piece distance between the endpoints of $e$, then the edge $e$ is contained in $H$.
Otherwise, the length of $e$ corresponds to portal distance.
But, since the set computed by Lemma~\ref{lem:portal-mst} is a MST of a graph defined over portal distances, we have added to $H$ the edge connecting $T'_1$ with $T'_2$ of minimum portal distance.
The length of this edge is at most the length of $e$ (in fact the lengths are equal).

The dynamic MSF algorithm guarantees that the new tree $\drzewo'$ after the operation consists of the connected components of $\drzewo \setminus v$ connected with a set of edges of minimum possible cost with respect to the weights in $\GD$.
Consequently, $\drzewo'$ is a MST of $\indu{\GD}{V(\drzewo')}$.
\maybeqed\end{proof}


\begin{lemma}\label{lem:dec_full}
Let $\eps > 0$.
Assume we are given a fully dynamic $\GDapx$-approximate vertex-color distance oracle with $\pornum$ portals.
For every $v \in V$, $|\portals(v)| \leq \pormax$, $|\pieces(v)| \leq \piemaxvass$ and an (expected) size of each piece is bounded by $\piemax$.
Then, there exists a dynamic algorithm that maintains a $2(1+\eps)\GDapx$ approximation of a minimum cost Steiner tree under a sequence of terminal deletions.
For every terminal deletion, in amortized sense, the algorithm issues $O(\eps^{-1})$ merge and split operations and uses $O(\eps^{-1} \log^4 n + \eps^{-1} \pornum + \eps^{-1} \log \eps^{-1})$ additional (expected) time.
Preprocessing requires $O(|S|)$ merge and split operations, as well as $O(|S| (\piemax \piemaxvass \log^4 n + \pornum \log n))$ (expected) time.
The algorithm uses $O(|S| (\eps^{-1} + \piemax \piemaxvass))$ (expected) space.
\end{lemma}

\begin{proof}
We use the decremental algorithm described in this section which, by Lemma~\ref{lem:dec:correct_full}, follows the decremental scheme.
The approximation guarantee follows from Lemma~\ref{lem:dec-eff_full}.

Let us now analyze the efficiency of the algorithm.
At the initialization step, we compute the minimum spanning tree of the terminal set.
By Lemma~\ref{lem:oracle-GD-mst}, this requires $O(|S|(\pornum + \piemax \piemaxvass) \log n)$ (expected) time.
Moreover, we add edges between mutually piece-visible terminals.
For each vertex $v$ we compute shortest paths in each element of $\pieces(v)$ and add respective edges to $H$.
The computation of shortest paths in each piece takes $O(|S| \piemax \piemaxvass \log n)$ (expected) time.
Thus, during initialization, we add a total of $O(|S| \piemax \piemaxvass)$ edges to $H$ (in expectance), which, by Theorem~\ref{thm:Thorup_full}, requires $O(|S| \piemax \piemaxvass \log^4 n)$ (expected) time.

When a vertex $v$ is removed from the maintained spanning tree we need to compute the set of reconnecting edges, using Lemma~\ref{lem:portal-mst} in $O(\degth \pornum + \degth \log \degth) = O(\eps^{-1} \pornum + \eps^{-1} \log \eps^{-1})$ (expected) time.
We obtain a set of $O(\degth) = O(\eps^{-1})$ edges that we add to $H$ in $O(\eps^{-1} \log^4 n)$ time.
Moreover, with respect to the operations on the oracle $\DO$, we merge and activate vertices of $S$ $O(\degth)$ times per removal.

For the space bound, we need $O(|S|(\eps^{-1} + \piemax \piemaxvass))$ (expected) space to store the graph $H$ and maintain the decremental MSF algorithm.
In addition to that, space is used by the vertex-color distance oracle.
\maybeqed\end{proof}

By using the two oracles developed in Section~\ref{sec:approx_distance_oracles}, we obtain the following two dynamic algorithms.

\begin{theorem}
Let $G=(V, E, \dlug_G)$ be a graph, $n = |V|$, $m = |E|$ and $\eps > 0$.
Let $S \subseteq V$ be a dynamic set, subject to vertex removals.
Then, after preprocessing in $O(\sqrt{n}(m + n \log n) + |S| \sqrt{n} \log^4 n)$ expected time, we may maintain a $(6+\eps)$-approximate Steiner tree that spans $S$, handling each removal from $S$ in $O(\eps^{-1}\sqrt{n} \log n)$ expected amortized time.
The algorithm uses $O(n\sqrt{n}\log n)$ expected space.
\end{theorem}

\begin{proof}
We combine Lemma~\ref{lem:dec_full} with $3$-approximate vertex-color distance oracles of Theorem~\ref{thm:distance-oracle-general} (based on Lemma~\ref{lem:generic_general_full}).
There are $\pienum = n$ pieces, each of expected size $\piemax = O(\sqrt{n})$.
Each vertex is assigned a single piece ($\piemaxvass = 1$).
For every vertex the number of portals is $\pormax = O(\sqrt{n})$ and so is the total number of portals ($\pornum = O(\sqrt{n})$).
Every $\opermerge$ and $\opersplit$ operation requires $T_{\opermerge} = O(\sqrt{n} \log n)$ time.
By plugging this parameters to Lemma~\ref{lem:dec_full} we obtain an algorithm which processes each removal in time
\begin{align*}
&O(\eps^{-1} \log^4 n + \eps^{-1} \pornum + \eps^{-1} \log \eps^{-1} + \eps^{-1} T_{\opermerge}) \\
  &\quad= O(\eps^{-1} \log^4 n + \eps^{-1} \sqrt{n} + \eps^{-1} \log \eps^{-1} + \eps^{-1} \sqrt{n} \log n) \\
  &\quad= O(\eps^{-1} \sqrt{n} \log n).
\end{align*}
The last equality holds under a natural assumption that $\log \eps^{-1} = O(\sqrt{n})$; otherwise
a brute-force approach works in $O(\eps^{-1})$ time.

It takes $O(\sqrt{n}(m + n \log n))$ expected time to initialize the oracle and 
\begin{align*}
&O(|S|(T_{\opermerge} + \piemax \piemaxvass \log^4 n + \pornum \log n)) \\
  &\quad = O(|S|(\sqrt{n} \log n + \sqrt{n} \log^4 n + \sqrt{n} \log n) \\
  &\quad = O(|S| \sqrt{n} \log^4 n)
 \end{align*}
expected time to initialize the algorithm.
Altogether, the expected initialization time is $O(\sqrt{n}(m + n \log n) + |S| \sqrt{n} \log^4 n)$.

The space usage is $O(|S| (\eps^{-1} + \piemax \piemaxvass) + n\sqrt{n}\log n ) = O(n\sqrt{n}\log n)$ expected 
space (again, we assume $\log \eps^{-1} = O(\sqrt{n})$).

We obtain approximation factor of $2(1+\eps)3 = 6(1+\eps)$, which can be adjusted to $6+\eps$ by using $\eps' = \eps/6$ instead of $\eps$.
\end{proof}

\begin{theorem}
Let $G=(V, E, \dlug_G)$ be a planar graph, $n = |V|$ and $\eps > 0$.
Denote by $D$ be the stretch of the metric induced by $G$.
Let $S \subseteq V$ be a dynamic set, subject to vertex removals.
Then, after preprocessing in $O(\eps^{-1} n \log^2 n \log D + |S| \eps^{-0.5} \sqrt{n} \log^{2.5} n \log D)$ time we may maintain $(2+\eps)$-approximate Steiner tree that spans $S$, handling each removal from $S$ in $\tilde{O}(\eps^{-1.5}\sqrt{n}\log D)$ amortized time.
\end{theorem}

\begin{proof}
We combine Lemma~\ref{lem:dec_full} with $(1+\eps)$-approximate vertex-color distance oracles of Theorem~\ref{thm:planar-full-oracle}.
There are $\pienum = O(n / \rho \log D)$ pieces, each of size $\piemax = O(\rho)$.
For every vertex the number of pieces is $\piemaxvass = O(\log D)$.
The total number of portals is $\pornum = O(\eps^{-1} n / \rho \log D)$, whereas the number of portals assigned to a vertex is at most $\pormax = O(\eps^{-1} \log n \log D)$.
Every $\opermerge$ and $\opersplit$ operation requires $T_{\opermerge} = O(\eps^{-1} n / \rho \log n \log D)$ time.

We plug this parameters to Lemma~\ref{lem:dec_full} and obtain an algorithm which processes each removal in time
\begin{align*}
&O(\eps^{-1} \log^4 n + \eps^{-1} \pornum + \eps^{-1} \log \eps^{-1} + \eps^{-1} T_{\opermerge}) \\
&\quad= O(\eps^{-1} \log^4 n + \eps^{-2} n / \rho \log D + \eps^{-1} \log \eps^{-1} + \eps^{-2} n / \rho \log n \log D) \\
&\quad= O(\eps^{-2} n / \rho \log n \log D).
\end{align*}

It takes $O(\eps^{-1} n \log^2 n \log D)$ time to initialize the oracle and 
\begin{align*}
&O(|S|(T_{\opermerge} + \piemax \piemaxvass \log^4 n + \pornum \log n)) \\
&\quad= O(|S|(\eps^{-1} n / \rho \log n \log D + \rho \log D \log^4 n + \eps^{-1} n / \rho \log n \log D)) \\
&\quad= O(|S|(\eps^{-1} n / \rho \log n \log D + \rho \log D \log^4 n))
\end{align*}
time to initialize the algorithm.
We balance the two terms from initialization time and get $\rho = \sqrt{n} \eps^{-0.5} / \log^{1.5}n$.
Thus, the initialization time becomes $O(\eps^{-1} n \log^2 n \log D + |S| \eps^{-0.5} \sqrt{n} \log^{2.5} n \log D)$.

The amortized time for removing one terminal is $O(\eps^{-2} n / \rho \log n \log D) = O(\eps^{-1.5} \sqrt{n} \log^{2.5} n \log D)$.

We obtain approximation factor of $2(1+\eps)^2$, which can be adjusted to $2+\eps$ by using a value of $\eps$ which is a constant factor smaller.
\end{proof}

\subsection{Incremental algorithm}
\label{sec:incremental_algorithm}

The goal is to maintain an approximate Steiner tree on a set of terminals, subject to terminal additions.
Fix a constant $\effeps > 0$ and set $\stepeps = \effeps / 2$.
We assume that the input graph $G$ may \emph{only} be accessed by incremental vertex-color distance oracles (see Section~\ref{sec:approx_distance_oracles}).
In particular, improving the oracles would immediately yield a better algorithm for our problem.
We denote the approximation ratio of the oracles by $\GDapx$ and assume $\GDapx = O(1)$.
We also assume that all the oracles yield the same $\GDapx$-near metric space $\GD = (V,\binom{V}{2},\dlug_\GD)$ that approximates the metric closure of $G$ (see Corollary~\ref{cor:consistency}).
In particular, the tree that we maintain has edge lengths corresponding to the lengths in $\GD$.
Moreover, we assume that the distances returned by the oracles are discrete, meaning that every distance is a power of $1+\stepeps$.
If this is not the case, we may round up every distance returned by the oracle.
As a result, the distances will correspond to a $\GDapx(1+\stepeps)$-near metric space.
However, for simplicity of presentation, let us assume that we work with a $\GDapx$-near metric space.
We will include the additional $1+\stepeps$ factor in the final analysis.

Our first goal is to maintain a good approximation of the minimum spanning tree of $S_i$ in graph $\GD$.
In order to achieve that, we use the incremental online algorithm from Section~\ref{ss:inc-scheme_full}.
The parameters $\effeps$ and $\stepeps$ that we have chosen correspond to the parameters used in the description of the  scheme.
We denote the terminals that are added by $t_1, t_2, t_3, \ldots$.
Let $S_i = \{t_1, \ldots, t_i\}$ and let $T_i$ denote the tree maintained by the algorithm after adding $i$ terminals.

Assume that all edge lengths of $\GD$ belong to the interval $[1, D]$.
Note that $D$ is at most $\GDapx$ times greater than the stretch of the metric induced by $G$.
We define the \emph{level} of an edge $uv$, as $\poziom(uv) := \log_{1+\stepeps} \dlug_\GD(uv)$.
Similarly, the level of a tree $T$ is $\poziom(T) := \lfloor \log_{1+\stepeps} \dlug_\GD(T) \rfloor$.
Note that since the distances are powers of $1+\stepeps$, the levels of edges are integers.
The algorithm represents the current tree $T_i$ by maintaining $h = \lfloor \log_{1+\stepeps} D\rfloor$ \emph{layers} $L_1, \ldots, L_{\laycnt}$.
Layer $L_j$ contains edges of the current tree $T_i$ whose level is at most $j$.
In particular, layer $L_h$ contains all edges of $T_i$.
We also define $\poziom_\bot(T) = \lfloor \log_{1+\effeps} \frac{\effeps \dlug_\GD(T)}{8\GDapx^2(1+\effeps)^3 n} \rfloor$.
Intuitively, edges of level at most $\poziom_\bot(T)$ are so short that they may sum up to at most $\effeps/2$ fraction of the weight of $T$.

For each layer $L_j$, we maintain an incremental vertex-color distance oracle $\DO_j$ on the graph $G$.
The partition of $V$ into colors corresponds to the connected components of a graph $(V,L_j)$.
A color is active if and only if its elements belong to $S_i$.

\begin{algorithm}
\caption{Adding a new terminal $t_i$}
\begin{algorithmic}[1]

\Procedure{\findrep}{$j$}
 \While{\textbf{true}}
  \State $t' := $ active vertex in $\DO_j$ which is second nearest to $t_i$ \label{closest_vertex}
  \If{$\poziom(t_it') >  j$}
   \State \textbf{break}
   \EndIf
    \State $e' := $ the heaviest edge in $T_i$ on the path from $t_i$ to $t'$ \label{heaviest_edge}
   \State replace $e'$ with $t_it'$ in $T_i$
    \For{$k = 1, \ldots, \laycnt$} \label{one_edge_replacement_begin}
      \If{$\poziom(e') \leq k$}
        \State replace $e'$ with $t_it'$ in $L_k$
      \ElsIf{$\poziom(t_it') \leq k$}
        \State add $t_it'$ to layer $L_k$
        \State merge colors of $t_i$ and $t'$ in $\DO_k$
      \EndIf
    \EndFor \label{one_edge_replacement_end}
 \EndWhile
\EndProcedure

\Procedure{AddTerminal}{$t_i$}
\State $t' := $ the active vertex in $\DO_h$ which is nearest to $t_i$ \label{first_edge_begin}
\State $T_{i} := T_{i-1} \cup \{t_it'\}$
\For{$j=1, \ldots, \laycnt$}
 \State{activate color of $t_i$ in $\DO_j$}
 \If{$\poziom(t_it') \leq j$}
   \State add $t_it'$ to $L_j$
   \State merge colors of $t_i$ and $t'$ in $\DO_j$
 \EndIf
\EndFor \label{first_edge_end}
\State $T_M := \textrm{maximum weight tree among } T_0, \ldots, T_{i-1}$
\For{$j=\poziom_\bot(T_M), \ldots, \poziom(T_M)$} \label{replacement_begin}
\State \Call{\findrep}{$j$}
\EndFor \label{replacement_end}
\EndProcedure
\end{algorithmic}
\label{algorithm:adding_a_terminal}
\end{algorithm}

We now describe the process of adding a terminal $t_i$.
The pseudocode is given in Algorithm~\ref{algorithm:adding_a_terminal}.
It consists of two stages.
First, we find the shortest edge in $\GD$ connecting $t_i$ to any of $t_1, \ldots, t_{i-1}$ and add this edge to $T_{i-1}$ to obtain tree $T_i$.
Then, we apply a sequence of $\stepeps$-efficient replacements $(e, e')$ to $T_i$ in order to decrease its weight.
According to the incremental scheme of Section~\ref{ss:inc-scheme_full} we only need to consider replacements in which $e$ is incident to the newly added terminal $t_i$.

In order to find the replacements, we use the vertex-color distance oracles.
Fix a layer number $j$ and assume that the colors of the oracle $\DO_j$ reflect the layer structure, including the newly added terminal $t_i$.
We want to find a replacement pair $(e, e')$ such that the $\poziom(e') > j$ and $\poziom(e) \leq j$.
Denote by $C$ the connected component of the graph $(V,L_j)$ that contains $t_i$ (observe that the vertex set of $C$
   is exactly the set of vertices with the same color as $t_i$ in $\DO_j$).
By definition of $C$, it consists of edges of level at most $j$ and the path in $T_i$ from $t_i$ to every $t \not\in C$ contains an edge of level $> j$.
We find the vertex $t \not\in C$ which is the nearest to $t_i$ by querying $\DO_j$.
Since $t_i$ has color $C$, we issue a $\opernearest(t_i, 2)$ query to find the second nearest color from $t_i$.
Assume that we find a vertex $t'$.
If $\poziom(t_it') \leq j$, we have found a replacement pair.
Let $e'$ be the heaviest edge on the path from $t_i$ to $t'$ in $T_i$ and $e = t_it'$.
Clearly, $(e, e')$ is a replacement pair, and since the weight of $e$ is lower than the weight of $e'$ and the weights can differ at least by a factor of $1+\stepeps$, this replacement is $\stepeps$-efficient.
Note that we choose $e'$ to be the heaviest edge of the path from $t_i$ to $t'$ in $T_i$, as required by the incremental scheme.

Our algorithm iterates through layers in increasing order, and for each layer it repeatedly looks for replacements using the procedure described above.
As we later show, in this way we will detect all replacements that we need.
In particular, applying a replacement on a higher level may not create new $(\effeps/2, 2\GDapx(1+\effeps)^2)$-good replacements on lower levels.

Let $T_M$ be the tree with maximum weight among $T_1, \ldots, T_{i-1}$.
For efficiency, we only iterate through layers $j = \poziom_\bot(T_M), \ldots, \poziom(T_M)$ to perform replacements.\footnote{A simple solution would be to iterate through all layers, or through a range of layers defined in terms of $T_{i-1}$, but we use a slightly more sophisticated approach in preparation for more efficient algorithm that we describe later. The key property here is that the weight of $T_M$ may be only slightly larger than the weight of $T_i$, but at the same time the weight of $T_M$ may only increase after adding every terminal.}
After this procedure, some $\stepeps$-efficient replacements may still be possible in $T_i$, but they would not be $(\effeps/2)/(2\GDapx(1+\effeps)^2)$-heavy.

Every time we change the current tree, we update all layers accordingly.
In case we add a new edge, this boils down to a single merge operation (in each layer that is supposed to contain the edge).
A more difficult case is when we perform a replacement, as then we \emph{remove} an edge that is replaced.
Assume that we perform a replacement $(e, e')$ and we need to update a layer $L_j$ that contains the edge $e'$.
From the algorithm we know that $e'$ is the heaviest edge on the path connecting the endpoints of $e$, so all edges of this path are contained in $L_j$ and the endpoints of $e$ have the same color.
Right after we remove $e'$ from $T_i$, we add $e$.
Since the level of $e$ is lower than the level of $e'$, the edge $e$ is added to layer $L_j$, which means that in the end the colors in $L_j$ do not change at all.

We first show that one call to \findreptxt{} correctly finds the replacements.

\begin{lemma}\label{lem:onefindrep}
After a call to \findreptxt{$(j)$}, there are no possible $\stepeps$-efficient replacements $(e, e')$, such that $\poziom(e') > j$, $\poziom(e) \leq j$ and $e$ is incident to $t_i$.
Every replacement made by \findreptxt{} is $\stepeps$-efficient.
Every time a replacement $(e, e')$ is made, the edge $e'$ is the heaviest edge on the path between the endpoints of $e$.
\end{lemma}
\begin{proof}
The second and third claims of the lemma follow directly from the discussion above.
It remains to prove the first one.
Note that once we identify a replacement pair, the loop in lines~\ref{one_edge_replacement_begin}-\ref{one_edge_replacement_end} correctly updates the data structures.

Assume that a desired replacement exists and $e = t_it''$.
The path from $t_i$ to $t''$ in the current spanning tree goes through $e'$.
Consider the vertex-color distance oracle $\DO_j$.
Since its colors correspond to connected components of the subgraph consisting of edges of level at most $j$ and $\poziom(e') > j$, the color of $t''$ is different from the color of $t_i$.
Thus, the vertex $t'$ found in~\ref{closest_vertex}rd line of the algorithm, is at distance at most $d(t_it'') \leq d(t_it')$ from $t_i$ and $\poziom(t_it') \leq j$, so the \textbf{while} loop does not terminate.
\maybeqed\end{proof}

We now show, that considering layers in the increasing order is correct.
In particular, performing a replacement using one layer may not create replacements at lower levels.

\begin{lemma}\label{lem:findrep}
After \findreptxt{$(j)$} is called with $j = a, \ldots, b$, there does not exist a $\stepeps$-efficient replacement $(e, e')$, such that $\poziom(e) \leq b$, $\poziom(e') > a$ and $e$ is incident to $t_i$.
Moreover, every time the algorithm identifies a replacement pair $(e, e')$ in the iteration $j$, we either have $\poziom(e) = j$, or
$j=a$ and $\poziom(e) \leq a$.
\end{lemma}
\begin{proof}
The proof is by induction on $b-a$. The base case $b=a$ is provided by Lemma~\ref{lem:onefindrep}; observe that also any replacement
 pair $(e,e')$ identified for $j=a$ satisfies $\poziom(e) \leq a$.
  Consider now an inductive step for $b > a$.

Throughout the proof we only consider replacements in which the replacement edge is incident to $e$.
By induction hypothesis, we know that in the beginning of last iteration for $j=b$
there are no $\stepeps$-efficient replacement $(e, e')$, such that $\poziom(e) \leq b-1$ and $\poziom(e') > a$.
Moreover, by Lemma~\ref{lem:onefindrep}, when the iteration is complete, there are no $\stepeps$-efficient replacement $(e, e')$, such that $\poziom(e) \leq b$ and $\poziom(e') > b$.
We combine this two facts to show the induction step, but we need to prove that performing replacements may not introduce any new replacements $(f, f')$ satisfying $\poziom(f) < b$ and $\poziom(f') > a$.

Consider the first replacement $(e,e')$ that we perform in the iteration $j=b$.
Then, from the induction hypothesis, it follows immediately that $\poziom(e) = b$.
We show that performing the replacement may not introduce any new replacements $(f, f')$, such that $\poziom(f) < b$ and $\poziom(f') > a$.
Indeed, observe that such a replacement can be identified by querying $\DO_{\poziom(f')}$.
But then the vertex-color distance oracles $\DO_{b'}$, for $b' < b$, are not modified, so no replacement pair $(f, f')$ such that $\poziom(f) < b$
and $\poziom(f') > a$ can be created.

Consequently, the next replacement $(e,e')$ performed in the iteration $j=b$ also satisfies $\poziom(e) = b$, and the argument
from the previous paragraph applies again. Consequently,
as we perform replacements, we do not introduce any replacement pair $(f, f')$ such that $\poziom(f) < b$ and $\poziom(f') > a$.
This concludes the proof of the lemma.
\maybeqed\end{proof}

\begin{lemma}
\label{lem:incalg_correct}
The maintained tree $T_k$ is a $2(1+\effeps)\GDapx$-approximation of a minimum Steiner tree in $G$ that spans $S_k$.
\end{lemma}

\begin{proof}
The proof proceeds by induction on the number of terminals $i$.
Our goal is to use Lemma~\ref{lem:inc-proof_full}, so we have to show that the algorithm presented here implements the incremental scheme from Section~\ref{ss:inc-scheme_full}.

In order to show that we follow the incremental scheme from Section~\ref{ss:inc-scheme_full} we need to prove that there are no $(\effeps/2, 2\GDapx(1+\effeps)^2)$-good replacements $(e, e')$ with $e$ incident to $t_i$.
By Lemma~\ref{lem:findrep}, after all calls to \findreptxt{}, for every such replacement, we either have $\poziom(e) > \poziom(T_M)$ of $\poziom(e') \leq \poziom_\bot(T_M)$.
Since the edges that we replace either belong to $T_{i-1}$ or have been added as a replacement for an edge in $T_{i-1}$, and $\dlug_\GD(T_M) \geq \dlug_\GD(T_i)$ the first case may not hold.
Let us then consider the second case.
We show that the condition implies that $(e, e')$ is not $(\effeps/2)/(2\GDapx(1+\effeps)^2)$-heavy.
We have
\begin{align*}
\dlug(e') & \leq \frac{\effeps\dlug(T_M)}{8(1+\effeps)^3 \GDapx^2 n} & \textrm{ from } \poziom(e') \leq \poziom_\bot(T_M)\\
    & \leq \frac{\effeps2(1+\effeps)\GDapx\dlug(T_i)}{8(1+\effeps)^3 \GDapx^2 n} & \textrm{by Lemma~\ref{lem:inc-apx_full}, } \dlug(T_M) \leq 2(1+\effeps)\GDapx\dlug(T_i)\\
    & \leq \frac{\effeps\dlug(T_i)}{4(1+\effeps)^2 \GDapx n}.
\end{align*}

This shows that the update procedure correctly implements the incremental scheme from Section~\ref{ss:inc-scheme_full}, so, by Lemma~\ref{lem:inc-proof_full}, it computes a $2(1+\effeps)\GDapx$-approximation of a minimum Steiner tree in $G$ that spans $S_k$.
\maybeqed\end{proof}

The above algorithm can be speeded up by reducing the number of layers, so that it does not depend on the stretch of the metric induced by the input graph.
Observe that although we maintain $O(\log_{1+\stepeps} D)$ layers, we only use layers $\poziom_\bot(T_M), \ldots, \poziom(T_M)$ to find the replacements.
Thus, we may modify our algorithm so that it only maintains the layers that are necessary.
We will maintain $h$ layers $L_1, \ldots, L_h$ that correspond to levels $\poziom_\bot(T_M), \ldots, \poziom_\bot(T_M)+h-1$.
The value $h$ is chosen such that $\poziom_\bot(T_M)+h-1$ is at least $\poziom(T_M)$, regardless of the weight of $T_M$.
Thus, for some tree $T$,
\[
 h = \left\lfloor \log_{1+\stepeps} \dlug(T) \right\rfloor - \left\lfloor \log_{1+\stepeps} \frac{\effeps \dlug(T)}{8\GDapx^2(1+\effeps)^3 n} \right\rfloor = O(\log_{1+\stepeps} (\GDapx^2 (1+\effeps)^3 n / \effeps)) = O(\effeps^{-1} \log (n/\effeps))
\]
Since the weight of $T_M$ may increase over time, the layers will be \emph{dynamic}, that is, the level corresponding to each layer will also increase.

The algorithm may be easily adapted to handle the dynamic layers.
Once a new tree $T_i$ has been computed and the weight of $T_M$ increases, the layers may correspond to higher levels and may need to be updated by adding some edges to them.
However, there are at most $n-1$ of these additions for every layer.

Let us now analyze the efficiency of the algorithm.
Note that while the layers are used for the description of the algorithm, they may be omitted in the implementation, as they are never read.
We first bound the running time of one call to \findreptxt{}.

\begin{lemma}
\label{lem:findrep_eff}
Assume that a call to \findreptxt{} performs $r$ replacements.
Then, it issues $r+1$ $\opernearest$ queries and at most $r$ $\opermerge$ operations on every level.
Moreover, it uses $O(r \log n)$ additional time.
\end{lemma}

\begin{proof}
We issue a $\opernearest$ query to find replacements.
$r$ of these queries identify new replacements and the last one detects that there are no more replacements to be made.
After a replacement is detected, we may execute a $\opermerge$ operation on every level.
Moreover, after a replacement is detected, we find the heaviest edge on some path in the tree (see line~\ref{heaviest_edge} of Algorithm~\ref{algorithm:adding_a_terminal}).
To do that efficiently, we maintain the the trees $T_i$ as top trees~\cite{AlstrupHLT05}.
This allows us to perform this operation in $O(\log n)$ time.
\end{proof}

\begin{lemma}
\label{lem:generic-incremental-steiner}
Let $0 < \effeps = O(1)$.
Given $h = O(\effeps^{-1} \log (n/\effeps))$ instances of an incremental $\GDapx$-approximate vertex-color distance oracle ($\GDapx = O(1)$), there exists a dynamic algorithm that maintains a $2\GDapx(1+\effeps)$-approximation of a minimum-cost Steiner tree under a sequence of terminal insertions.
For every terminal addition, in amortized sense, the algorithm executes $O(\effeps^{-1} \log (n/\effeps))$ $\operactivate$, $\opermerge$ and $\opernearest$ operations.
The space usage is dominated by the space usage of the oracles. 
\end{lemma}

\begin{proof}
By Lemma~\ref{lem:incalg_correct}, the incremental algorithm implements the incremental scheme.
As a result, we may use Lemma~\ref{lem:inc-eff_full} to bound the running time.

Assume that we have added $k$ terminals.
First, we need to bound the number of replacements.
Since we round the distances up to a power of $(1+\stepeps)$, we operate on a $\GDapx(1+\stepeps)$-near metric space.
Our algorithm implements the incremental scheme with parameters $\stepeps = \effeps/2$.
Thus, by Lemma~\ref{lem:inc-eff_full}, adding $k$ terminals may trigger $O(k \effeps^{-1} \log ((1+\effeps)\GDapx)) = O(k \effeps^{-1})$ replacements.
In particular, this shows that the update procedure terminates.

Let us now calculate the number of operations performed on the vertex-color distance oracles.
We use one copy of the oracle for each layer.
In the course of adding $k$ terminals, in each of the $h$ layers at most $k$ vertices may be activated and colors may be merged $k-1$ times.
Note that this includes the merging of colors caused by maintaining dynamic layers.
For every terminal insertion we also issue one nearest query to find the first edge connecting the new terminal with the current tree.
Since a single terminal insertion triggers $O(\effeps^{-1})$ replacements, by Lemma~\ref{lem:findrep_eff}, we we issue $O(h + \effeps^{-1}) = O(\effeps^{-1}\log (n/\effeps))$ $\opernearest$ queries, in amortized sense.
We also spend $O(\effeps^{-1} \log n)$ additional time, but this time is dominated by the oracle operations.

Apart from the oracles, the algorithm uses $O(n)$ space to represent the trees.
Thus, the space usage is dominated by the oracles.
\maybeqed\end{proof}

By using the two oracles developed in Section~\ref{sec:approx_distance_oracles}, we obtain the following two dynamic algorithms.

\begin{theorem}
Let $G=(V, E, \dlug_G)$ be a graph, $n = |V|$ and $\eps > 0$.
Let $S \subseteq V$ be a dynamic set, subject to vertex insertions (initially $S=\emptyset$).
Then, after preprocessing in $O(\eps^{-1} \sqrt{n}(m + n \log (n / \eps)))$ expected time we may maintain $(6+\eps)$-approximate Steiner tree that spans $S$, handling each insertion to $S$ in $\tilde{O}(\eps^{-1}\sqrt{n})$ amortized time.
The algorithm uses $O(\eps^{-1} n \sqrt{n} \log n \log (n/\eps))$ space.
\end{theorem}

\begin{proof}
We combine the $3$-approximate vertex-color distance oracle from Theorem~\ref{thm:distance-oracle-general} (based on Lemma~\ref{lem:generic_general_full}) with Lemma~\ref{lem:generic-incremental-steiner}.
Thus $\GDapx = 3 = O(1)$.
We set $\effeps = \eps/3$.

We first build $O(\effeps^{-1} \log (n/\effeps))$ copies of the oracle, by first building a single one and then copying it.
Building the first oracle requires $O(\sqrt{n}(m + n \log n))$ expected time.
Since the space usage is $O(n\sqrt{n} \log n)$, the copying requires $O(\effeps^{-1} n \sqrt{n} \log n \log (n/\effeps))$ time, and so is the space usage.

The oracle handles $\operactivate$ and $\opermerge$ operations in $O(\sqrt{n} \log n)$ time, whereas $\opernearest$ queries are answered in $O(\sqrt{n})$ time.
Since we perform $O(\effeps^{-1} \log (n/\effeps))$ of each of these operations, they require $O(\effeps^{-1}\sqrt{n}\log n \log (n/\effeps)) = \tilde{O}(\eps^{-1}\sqrt{n})$ amortized time per terminal insertion.

We obtain an approximation factor of $2(1+\eps)3 = 6(1+\eps)$, which may be reduced to $6+\eps$ by manipulating the value of $\eps$.
\end{proof}

\begin{theorem}
Let $G=(V, E, \dlug_G)$ be planar a graph, $n = |V|$ and $D$ be the stretch of the metric induced by $G$.
Fix $\eps > 0$.
Let $S \subseteq V$ be a dynamic set, subject to vertex insertions (initially $S=\emptyset$).
Then, after preprocessing in $O(\eps^{-2} n \log n \log (n/\eps) \log D))$ time we may maintain $(2+\eps)$-approximate Steiner tree that spans $S$, handling every insertion to $S$ in $O(\eps^{-2} \log^2 n \log (n/\eps) \log D)$ expected amortized time or
$O(\eps^{-2} \log^2 n \log \log n \log (n/\eps) \log D)$ deterministic amortized time.
The algorithm uses $O(\eps^{-2} n \log n \log (n/\eps) \log D))$ space.
\end{theorem}

\begin{proof}
We combine the $(1+\eps)$-approximate vertex-color distance oracle of Theorem~\ref{thm:oracle-planar-incremental} with Lemma~\ref{lem:generic-incremental-steiner}.
We set $\effeps = \eps$.

We first build $O(\effeps^{-1} \log (n/\effeps))$ copies of the oracle, by first building a single one and then copying it.
Building the first oracle requires $O(\effeps^{-1} n \log^2 n \log D)$ time.
Since the space usage is $O(\effeps^{-1} n \log n \log D)$, the copying requires $O(\effeps^{-2} n \log n \log (n/\effeps) \log D)) = O(\eps^{-2} n \log n \log (n/\eps) \log D))$ time, and so is the space usage.
The copying also dominates the preprocessing time.

The oracle handles $\operactivate$ and $\opermerge$ operations in $O(\eps^{-1} \log^2 n \log D)$ expected amortized time, whereas $\opernearest$ queries are answered in $O(\eps^{-1} \log n \log D)$ time (and additional $\log \log n$ factor is needed for deterministic algorithms).
Since we perform $O(\effeps^{-1} \log (n/\effeps))$ of each of these operations, they require $O(\eps^{-2} \log^2 n \log (n/\effeps) \log D)$ expected amortized time (with an additional $\log \log n$ factor for a deterministic algorithm).

We obtain approximation ratio of $2(1+\eps)^2$.
It is easy to see that by manipulating the value of $\eps$, we may obtain $(2+\eps)$-approximation.
\end{proof}

\subsection{Fully dynamic algorithm}
In this section we merge the ideas of the decremental and incremental algorithms to obtain a fully dynamic algorithm, which supports both terminal additions and deletions.
The algorithm simply maintains the invariants of both the incremental and decremental algorithms.
Let $\eps > 0$ and $\stepeps = \effeps > 0$.
We set $\degth = 1 + \lceil \eps^{-1} \rceil$.
We implement the fully dynamic scheme of Section \ref{ss:fully_full} with the parameters $\eps$, $\stepeps$ and $\effeps$.
Recall that $\eps$ (and $\degth$) controls the degree threshold for nonterminals that are deleted, and $\stepeps$ defines the efficiency of the replacements that we make.
We maintain the invariant that the tree does not contain any $\effeps$-efficient replacements.
Since, $\stepeps = \effeps$, in the following we only use $\effeps$.

Similarly to the decremental algorithm, we work with a graph $G=(V,E,\dlug_G)$ that is accessed with approximate vertex-color distance oracles, that yield a $\GDapx$-near metric space $\GD$, which approximates $G$ (see Corollary~\ref{cor:consistency}).
We assume that $\GDapx = O(1)$ and that every edge length is $\GD$ has length being a power of $(1+\effeps)$.
This can be achieved by rounding up the distances returned by the oracle.
As a result, the graph $\GD$ yielded by the oracle becomes a $\GDapx(1+\effeps)$-near metric approximating $G$, but in the description we assume the approximation factor of $\GDapx$.
The additional $1+\effeps$ factor will be included in the final analysis.

As in the decremental algorithm, we do not maintain the tree explicitly, but instead we use the dynamic MSF algorithm (see Theorem~\ref{thm:Thorup_full}) on top of a carefully chosen graph $H$ with vertex set $V$.
Again, each connected component of $H$ is an isolated vertex, with the exception of one component, denoted $H_0$, that contains all terminals.
We declare that the currently maintained Steiner tree, $\drzewo$, is the tree maintained by the dynamic MSF algorithm in the component $H_0$.

We define the level of an edge $uv$ to be $\poziom(uv) := \log_{1+\effeps} \dlug_\GD(uv)$.
Our algorithm maintains a tree $T$ which spans the set of terminals and nonterminals of degree more than $\degth$.
The tree $T$ will not admit $\effeps$-efficient replacements, but since the distances are powers of $1+\effeps$, this means that $T$ does not admit any efficient replacements.
Thus, by Lemma~\ref{lem:no-good-replacement_full}, it is the MST of its vertex set.

Denote by $D$ the stretch of the metric induced by $G$.
Thus, the longest edge in $\GD$ is at most $\GDapx \cdot D = O(D)$ times longer than the shortest one.
In the algorithm we maintain $h = \lfloor \log_{1+\effeps} D \rfloor = O(\effeps^{-1} \log D)$ fully dynamic vertex-color distance oracles $\DO_1, \ldots, \DO_h$.
$\DO_h$ is the counterpart of the oracle from the decremental algorithm.
The colors of $\DO_h$ correspond to connected components of $H$.
The only active color is the one corresponding to component $H_0$.
The tree associated with this color is equal to the tree $T$.
The oracle $\DO_i$ is obtained from $\DO_h$ by splitting along edges of level more than $i$.
In other words, we consider the forest consisting of edges of $T$ of level at most $i$ and for each tree $F_j$ in this forest, $\DO_i$ contains an active color associated with the tree $F_j$.
We update the oracles in order to maintain those properties.

Note that we may assume that all oracles have the same structure, i.e., the sets $\portals(v)$ and $\pieces(v)$ are the same in each oracle.
This can be achieved by constructing one oracle and then copying it to create all the others.
Thus, referring to piece or portal distance between vertices in unambiguous.

As in the decremental algorithm of Section~\ref{sec:decremental}, we maintain an invariant that for every two vertices $u,w \in V(T)$, if $u$ is piece-visible from $w$, the graph $H$ contains an edge connecting them, whose length is equal to the piece distance.

With respect to the vertex-color distance oracles, note that each modification (edge addition, removal or replacement) in the tree $\drzewo$ requires us to perform $O(1)$ operations on each oracle $\DO_i$.

\paragraph{Vertex removal.} To use the fully dynamic scheme of Section~\ref{ss:fully_full}, we first need to describe how the algorithm behaves if the procedure $\remove$ requires us to remove a vertex $v$ of degree at most $\degth$ from the tree $\drzewo$.
We resolve this situation in exactly the same manner as in the decremental algorithm of Section~\ref{sec:decremental}.
For sake of completeness, let us recall the steps.

\begin{enumerate}
\item Split the active color $V(\drzewo)$ in $\DO_h$ into $s+1$ pieces being $\{v\}$ and the $s$ connected components of $\drzewo \setminus \{v\}$, and then deactivate $\{v\}$. Update all oracles $\DO_j$ accordingly.
\item Compute the set of reconnecting edges that correspond to portal distances (using Lemma~\ref{lem:portal-mst}) and add these edges to $H$.
\item Remove all edges incident to $v$ from $H$.
\item At this point, the dynamic MSF algorithm may update the maintained tree.
        Update the vertex-color distance oracles $\DO_i$, so that they reflect these changes of the MSF of $H$.
\end{enumerate}

\paragraph{Vertex addition.}
A second ingredient needed to implement the fully dynamic scheme is the behavior upon addition of a new terminal $v$.
If $v \in V(H_0)$, we simply mark it as a terminal and finish.

Let us now assume that we are adding a new vertex $v$ to the tree.
We use a similar method as in the incremental scheme, but here we need to assure that after adding $v$, the tree does not admit any $\effeps$-efficient replacements.
We use the following algorithm.
Note after every change to $T$ done by the dynamic MSF algorithm we update the nearest oracles accordingly.

\begin{enumerate}
\item Connect $v$ to the tree with the shortest possible edge in $\GD$. Add this edge to $H$.
\item Iterate through levels in increasing order, and for each level find and apply all occurring $\effeps$-efficient replacements.
This is done with the procedure \findreptxt{} from the (see Algorithm~\ref{algorithm:adding_a_terminal}, Section~\ref{sec:incremental_algorithm}).
Each time a replacement $(e, e')$ is found, simply add edge $e$ to $H$. The edge $e'$ may safely remain in $H$.
\item For every vertex $w$, whose degree became at most $\degth$ due to replacements, call $\remove(w)$.
\item Add edges corresponding to piece distances between $v$ and vertices of $T$.
\end{enumerate}

We claim that the dynamic MSF algorithm modifies the tree $\drzewo$ exactly the way the fully dynamic scheme requires it to do.
\begin{lemma}\label{lem:fully:correct}
The tree $\drzewo$ implements the fully dynamic scheme of Section~\ref{ss:fully_full}.
\end{lemma}
\begin{proof}
The correctness of the implementation of the procedure $\remove$ follows from the same arguments as in the decremental algorithm.
Indeed, our algorithm maintains the same invariants about graph $H$ and before the removal the tree $T$ is the MST of the terminals.
Since we use the same procedure for handling deletions, the proof is analogous to the proof of Lemma~\ref{lem:dec:correct_full}.

Let us now focus on a step when a vertex $v \notin V(H_0)$ is added to the terminal set.
By Lemma~\ref{lem:onefindrep}, every replacement $(e, e')$ is $\effeps$-efficient and $e'$ is the friend of $e$ with maximum possible cost.
Since we iterate through all levels, by Lemma~\ref{lem:findrep}, we have that after handling the insertion there are no $\effeps$-efficient replacements $(e, e')$ with $e$ incident to $v$.
Now, by Lemma~\ref{lem:fully:apx_full}, if there are no $\effeps$-efficient replacements $(e, e')$ with $e$ incident to $v$, there are no $\effeps$-efficient replacements at all.
In particular, the addition of edges incident to $v$ corresponding to pieces distances does not cause any changes to $\drzewo$.

Finally, note that as the weights of the edges in the dynamic MSF algorithm equal their levels in $\GD$, and we break ties by timestamps, if an addition of a new edge results in a replacement in $\drzewo$, such a replacement is $\effeps$-efficient.
This concludes the proof of the lemma.
\maybeqed\end{proof}

We are now ready to state the main result of this section.

\begin{lemma}\label{lem:fully_generic}
Let $G=(V,E,\dlug_G)$ be a graph, $\eps > 0$.
Denote by $D$ the stretch of the metric induced by $G$.
Assume there exists a fully dynamic $\GDapx$-approximate vertex-color distance oracles, where $\GDapx = O(1)$.
Let $\pornum$ be the (expected) total number of portals in the oracle.
Moreover, assume that every vertex is assigned at most $\piemaxvass$ pieces (in expectation), the (expected) size of each piece is bounded by $\piemax$, and the (expected) total number of pieces is $\pietot$.

Then, there exists a dynamic algorithm that maintains a $2(1+\eps)^2\GDapx$ approximation of a minimum cost Steiner tree under a sequence of terminal additions and deletions (initially, the terminal set is empty).
For every terminal addition, in amortized sense, the algorithm executes $O(\eps^{-2} \log^2 D)$ $\opersplit$ and $\opermerge$ and issues $O(\eps^{-1} \log D)$ $\opernearest$ queries.
Moreover, it uses $O(\eps^{-1}\pornum + \eps^{-2} \log n \log D + (\piemaxvass \piemax + \eps^{-1})\log^4 n)$ additional (expected) time.

The initialization time is dominated by the time needed to initialize $\Theta(\eps^{-1} \log D)$ fully dynamic vertex-color distance oracles.
Similarly, the space usage is dominated by the oracles.
\end{lemma}

\begin{proof}
Set $\effeps = \eps$.
By Lemma~\ref{lem:fully:correct}, our algorithm implements the fully dynamic scheme.
Thus, by Lemma~\ref{lem:fully:apx_full}, it maintains a $2(1+\eps)^2\GDapx$ approximation of a minimum cost Steiner tree.

The algorithm starts by constructing $\Theta(\effeps^{-1} \log D) = O(\eps^{-1} \log D)$ oracles $\DO_i$.
Initially, the graph $H$ is empty.

After each removal of a vertex in $\drzewo$ we apply Lemma~\ref{lem:portal-mst}, which requires $O(\eps^{-1}\pornum + \eps^{-1} \log \eps^{-1})$ time.
Moreover, we issue $O(\effeps^{-1})$ $\opermerge$ and $\opersplit$ operations in each of the $O(\effeps^{-1} \log D)$ oracles, that is, a total of $O(\effeps^{-2} \log D) = O(\eps^{-2} \log D)$ operations.

Now consider a vertex addition.
By Lemma~\ref{lem:full-eff_full}, an addition may trigger $O(\effeps^{-1} \log D) = O(\eps^{-1} \log D)$ replacements.
When a new vertex is added we call \findreptxt{} $O(\effeps^{-1} \log D)$ times, and each call issues one $\opernearest$ query, regardless of the number of replacements.
Moreover, for every replacement, all calls to \findreptxt{} issue an additional $\opernearest$ query and $O(\effeps^{-1} \log D)$ $\opermerge$ operations.
It also uses $O(\effeps^{-1} \log n)$ additional time.
This gives a total of $O(\effeps^{-2} \log^2 D) = O(\eps^{-2} \log^2 D)$ $\opermerge$ and $\opersplit$ operations per insertion, $O(\effeps^{-1} \log D) = O(\eps^{-1} \log D)$ $\opernearest$ queries and $O(\effeps^{-2} \log n \log D) = O(\eps^{-2} \log n \log D)$ additional time.

Whenever we add or remove a vertex from $\drzewo$, we add or remove all edges corresponding to piece distances.
There are at most $O(\piemaxvass \piemax)$ such edges for every vertex (in expectance).
All these edges can be computed in $O(\piemaxvass\piemax \log n)$ using Dijkstra's algorithm.

We are left with bounding the amortized cost of maintaining the dynamic MSF algorithm over the graph $H$.
Each call to $\remove$ that results in a modification of $\drzewo$ as well as each addition of a new terminal adds to $H$ $O(\piemaxvass \piemax)$ edges corresponding to piece distances.
By Lemma~\ref{lem:full-eff_full} and Theorem~\ref{thm:Thorup_full} this amortizes to $O(\piemaxvass \piemax \log^4 n)$ per operation.
Moreover, each removal adds $O(\eps^{-1})$ new edges to $H$, which accounts for $O(\eps^{-1} \log^4 n)$ amortized time.

Altogether, the amortized cost of one operation is $O(\eps^{-2} \log n \log D + \eps^{-1}\pornum + \eps^{-1} \log \eps^{-1} + \piemaxvass \piemax \log^4 n + \eps^{-1} \log^4 n) = O(\eps^{-1}\pornum + \eps^{-2} \log n \log D + (\piemaxvass \piemax + \eps^{-1})\log^4 n)$.

Apart from the vertex-color distance oracles, the space usage is $O(n)$.
Thus, the oracles dominate the space usage.
\maybeqed\end{proof}

By using the two oracles developed in Section~\ref{sec:approx_distance_oracles}, we obtain the following dynamic algorithms.

\begin{theorem}
Let $G=(V, E, \dlug_G)$ be a graph, $n = |V|$, $m = |E|$ and $\eps > 0$.
Denote by $D$ the stretch of the metric induced by $G$.
Let $S \subseteq V$ be a dynamic set, subject to vertex additions and removals (initially $S = \emptyset$).
Then, after preprocessing in $O(\sqrt{n}(m + \eps^{-1} n \log n \log D))$ expected time, we may maintain a $(6+\eps)$-approximate Steiner tree that spans $S$, handling each change to $S$ in $\tilde{O}(\eps^{-2} \sqrt{n} \log^2 D)$ expected amortized time.
The algorithm uses $O(n \sqrt{n} \log n \eps^{-1} \log D)$ expected space.
\end{theorem}

\begin{proof}
We combine Lemma~\ref{lem:fully_generic} with vertex-color distance oracles for general graphs of Theorem~\ref{thm:distance-oracle-general} (based on Lemma~\ref{lem:generic_general_full}).
We have that $\GDapx = 3 = O(1)$.
The expected total number of portals is $\pornum = \sqrt{n}$.
Every vertex is assigned a single piece $\piemaxvass = 1$ of expected size at most $\piemax = O(\sqrt{n})$.
The total number of pieces is $\pietot = O(n)$.
Each $\opermerge$ and $\opersplit$ operation requires $T_{\opermerge} = O(\sqrt{n} \log n)$ expected time, whereas a $\opernearest$ query can be handled in $T_{\operquery} = O(\sqrt{n})$ expected time.

It takes $O(\sqrt{n}(m + n\log n))$ expected time to initialize one oracle, and, since the expected space usage is $O(n \sqrt{n} \log n)$, copying the oracle  $\Theta(\eps^{-1} \log D)$ times requires $O(n \sqrt{n} \log n \eps^{-1} \log D)$ time.
Thus, the expected preprocessing time is $O(\sqrt{n}(m + n\log n) + n \sqrt{n} \log n \eps^{-1} \log D) = O(\sqrt{n}(m + \eps^{-1} n \log n \log D))$.
The space usage is dominated by the oracles, and amounts to $O(n \sqrt{n} \log n \eps^{-1} \log D)$.

Each terminal addition or deletion requires time
\begin{align*}
&O(\eps^{-2} \log^2 D T_{\opermerge} + \eps^{-1} \log D T_{\operquery} + \eps^{-1}\pornum + \eps^{-2} \log n \log D + (\piemaxvass \piemax + \eps^{-1})\log^4 n) \\
  &\quad = O(\eps^{-2} \log^2 D \sqrt{n} \log n + \eps^{-1} \log D \sqrt{n} + \eps^{-2} \log n \log D + (\sqrt{n} + \eps^{-1}) \log^4 n) \\
  &\quad = O(\eps^{-2} \log^2 D \sqrt{n} \log n + \sqrt{n} \log^4 n) \\
  &\quad = \tilde{O}(\eps^{-2} \sqrt{n} \log^2 D).
\end{align*}

We obtain approximation factor of $2(1+\eps)^2 3 = 6(1+\eps)^2$, which can be reduced to $6+\eps$ by manipulating the value of $\eps$.
\end{proof}

Using Theorem~\ref{thm:strecz-decrease_full}, we may decrease the dependency on the stretch of the metric at the cost of increasing the polylogarithmic factors.

\begin{theorem}
Let $G=(V, E, \dlug_G)$ be a graph, $n = |V|$, $m = |E|$ and $\eps > 0$.
Denote by $D$ the stretch of the metric induced by $G$.
Let $S \subseteq V$ be a dynamic set, subject to vertex additions and removals (initially $S = \emptyset$).
Then, after preprocessing in $O(\eps^{-1}\sqrt{n}(m + \eps^{-2} n \log n \log(\eps^{-1}n) \log D))$ expected time, we may maintain a $(6+\eps)$-approximate Steiner tree that spans $S$, handling each change to $S$ in $\tilde{O}(\eps^{-5} \sqrt{n} \log D)$ expected amortized time.
The algorithm uses $O(\eps^{-2}n \sqrt{n} \log n \log(\eps^{-1}n) \log D)$ expected space.
\end{theorem}

Finally, we show a fully dynamic algorithm for planar graphs.
\begin{theorem}
Let $G=(V, E, \dlug_G)$ be a planar graph, $n = |V|$ and $\eps > 0$.
Denote by $D$ the stretch of the metric induced by $G$.
Let $S \subseteq V$ be a dynamic set, subject to vertex additions and removals (initially $S = \emptyset$).
Then, after preprocessing in $O(\eps^{-2} n \log^2 n \log^2 D)$ time, we may maintain a $(2+\eps)$-approximate Steiner tree that spans $S$, handling each change to $S$ in $\tilde{O}(\eps^{-2}\log^{2.5} D \sqrt{n})$ amortized time.
The algorithm uses $O(\eps^{-2} n \log^2 n \log^2 D)$ space.
\end{theorem}

\begin{proof}
We combine Lemma~\ref{lem:fully_generic} with $(1+\eps)$-approximate vertex-color distance oracles for general graphs of Theorem~\ref{thm:distance-oracle-general}.
We have that $\GDapx = 1+\eps = O(1)$.
The total number of portals is $\pornum = O(\eps^{-1} n / \rho \log D)$.
Every vertex belongs to $\piemaxvass = O(\log D)$ pieces and each piece is a planar graph containing $\piemax = O(\rho)$ vertices.
The total number of pieces is $\pietot = O(n / \rho \log D)$.
Each $\opermerge$ and $\opersplit$ operation requires $T_{\opermerge} = O(\eps^{-1} \frac{n}{\rho} \log n \log D)$ time, whereas a $\opernearest$ query can be handled in $T_{\operquery} = O((\eps^{-1} \log n + \rho) \log D)$ time.

It takes $O(\eps^{-1} n \log^2 n \log D)$ time to initialize one oracle, and, since we copy it $\Theta(\eps^{-1} \log D)$ and its space usage is $O(\eps^{-1} n \log^2 n \log D)$, the copying requires $O(\eps^{-2} n \log^2 n \log^2 D)$ time.
This also dominates the preprocessing time.
The space usage of the oracles and the algorithm is also $O(\eps^{-2} n \log^2 n \log^2 D)$.

Each terminal addition or deletion requires time
\begin{align*}
&O(\eps^{-2} \log^2 D T_{\opermerge} + \eps^{-1} \log D T_{\operquery} + \eps^{-1}\pornum + \eps^{-2} \log n \log D + (\piemaxvass \piemax + \eps^{-1})\log^4 n) \\
&\quad = O(\eps^{-3} \log^3 D (n / \rho) \log n + \eps^{-1} \log^2 D (\eps^{-1} \log n + \rho) + \eps^{-2} (n / \rho) \log D +  \eps^{-2} \log n \log D + \\
       & \quad \quad (\rho \log D + \eps^{-1}) \log^4 n) \\
&\quad= O(\eps^{-3} \log^3 D (n / \rho) \log n + \eps^{-1} \log^2 D (\eps^{-1} \log n + \rho) + \eps^{-2}(n / \rho) \log D  + \eps^{-2} \log n \log D + \\ 
        &\quad\quad (\rho \log D + \eps^{-1})\log^4 n).
\end{align*}
We set $\rho = \eps^{-1} \sqrt{n \log D}$ and obtain a running time of $O(\eps^{-2}\log^{2.5} D \sqrt{n} \log n + \sqrt{n} \log^{2.5} D \log^4 n) = \tilde{O}(\eps^{-2}\log^{2.5} D \sqrt{n})$.
We obtain approximation factor of $2(1+\eps)^3$, which can be reduced to $2+\eps$ by manipulating the value of $\eps$.
\end{proof}

We may again apply Theorem~\ref{thm:strecz-decrease_full} to decrease the dependency on the stretch of the metric induced by $G$.

\begin{theorem}
Let $G=(V, E, \dlug_G)$ be a planar graph, $n = |V|$ and $\eps > 0$.
Denote by $D$ the stretch of the metric induced by $G$.
Let $S \subseteq V$ be a dynamic set, subject to vertex additions and removals (initially $S = \emptyset$).
Then, after preprocessing in $\tilde{O}(\eps^{-5} n \log D)$ time, we may maintain a $(2+\eps)$-approximate Steiner tree that spans $S$, handling each change to $S$ in $\tilde{O}(\eps^{-5.5} \sqrt{n} \log D)$ amortized time.
\end{theorem}

\section{Steiner tree via bipartite emulators}\label{sec:enum-apply}
In this section we present an alternative solution for maintaining approximate Steiner tree.
This solution exposes a trade-off between the running time and the approximation ratio. The algorithms presented here are faster than the algorithms from previous sections, however their approximation ratio is much higher.
We introduce the notion of \emph{bipartite emulator}, which will be essential for the
algorithms proposed here. In essence, a bipartite emulator is a sparse bipartite graph that
stores the information about the approximate distances in the input graph. 
In this section we show how to use a bipartite emulator to maintain a
good approximation of the minimum Steiner tree.
In Section~\ref{sec:emulators} we show how to construct the emulators themselves.

\begin{definition}
Let $G$ be a graph and $\alpha \geq 1$.
A bipartite emulator of $G$ is a bipartite graph $B=(V(G) \uplus N, E_B, \dlug_B)$, where
the vertex set is partitioned into $V(G)$ and a set of auxiliary vertices $N$ representing the distances
in $G$. For every $u, v \in V(G)$, it holds that $\dist_{G}(u,v) \leq \dist_{B}(u,v) \leq \alpha\cdot\dist_{G}(u,v)$.
Moreover, for every $u,v \in V(G)$, there exists in $B$ a vertex $x \in N$,
such that $\dlug_B(ux)+\dlug_B(xv) \leq \alpha \cdot \dist_{G}(u,v)$.
We say that $\alpha$ is the \emph{stretch} of emulator $B$.
We denote $\faladist_B(u,v)= \min_{x \in N} \{ \dlug_B(ux)+\dlug_B(xv) \}$
and we let $\emuclo{G}=(V(G),\binom{V(G)}{2}, \faladist_B)$ be the approximation of the metric closure of $G$ given by the emulator.
\end{definition}
Note that a bipartite emulator is not an emulator according to the usual definition, as an emulator is required to have the same
vertex set as the graph it emulates.
Observe also that $\emuclo{G}$ may not be a metric space itself, as it may not satisfy the triangle inequality. See also Figure~\ref{fig:emul1} for an illustration.

\begin{figure}[bt]
\centering
\includegraphics[width=.6\linewidth]{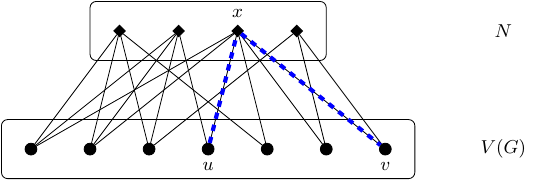}
\caption{An illustration of a bipartite emulator.
  The blue thick dashed path corresponds to the distance $\faladist_B(u,v) = \dlug_B(ux) + \dlug_B(xv)$.}
\label{fig:emul1}
\end{figure}

The definition shows that a bipartite emulator $B$ contains estimates of all distances in $G$ and that the distance estimate is given by a path of length two in $B$.
Moreover, it easily follows that, to approximate $\st(G,S)$, it suffices to focus
on $\mst(\indu{\emuclo{G}}{S})$.
\begin{lemma}\label{lem:emu-mst-apx}
Let $B$ be a bipartite emulator of $G$ with stretch $\alpha$. Then for any set of terminals
$S \subseteq V(G)$ it holds that
$$\faladist_B(\mst(\indu{\emuclo{G}}{S})) \leq 2\alpha \dlug_G(\st(G,S)).$$
\end{lemma}
\begin{proof}
Denote $T_{\mst} = \mst(\indu{\mclo{G}}{S})$ and observe that
\begin{align*}
\faladist_B(\mst(\indu{\emuclo{G}}{S})) &\leq \faladist_B(T_{\mst}) &\textrm{by the definition of }\mst\\
  &\leq \alpha\dlug_{\mclo{G}}(T_\mst) &\textrm{by the stretch of }B\\
  &\leq 2\alpha\dlug_G(\st(G,S)) &\textrm{by Lemma~\ref{lem:two-apx-st}.}
\end{align*}
\end{proof}

To obtain the necessary distances (distances in $\emuclo{G}$ between vertices of $S$) it is enough to
consider the subgraph $\indu{B}{S \cup \nei{S}}$ of $B$ induced by $S$ and its neighborhood in $B$.
We make sure that the degree of every vertex of $V(G)$ in $B$ is bounded.
This will allow us to maintain graph $\indu{B}{S \cup \nei{S}}$ efficiently as $S$ is undergoing changes.
If a vertex of $V(G)$ becomes a terminal, we add it and the adjacent edges to $\indu{B}{S \cup \nei{S}}$ in time proportional to its degree.
We apply to $\indu{B}{S \cup \nei{S}}$ the dynamic MSF algorithm (see Theorem~\ref{thm:Thorup_full}), which maintains the minimum spanning forest in a graph under edge additions or deletions.
Each edge insertion/removal is handled in polylogarithmic time, so the time we need to add a vertex to $B$ is bounded by its degree times logarithmic factors.
The only problem that remains is that the maintained spanning tree may be too heavy, as $\nei{S}$ may contain more vertices than
needed.
Note that every edge in $\indu{\emuclo{G}}{S}$ corresponds to a vertex in $\nei{S} \subseteq N$.
We need vertices corresponding to the edges of $\mst(\indu{\emuclo{G}}{S})$, and the rest of $\nei{S}$ unnecessarily increases the cost of the spanning tree.
In the next lemma we show, that if we remove leaves of $\mst(\indu{B}{S \cup \nei{S}})$ that belong to $N$, then the obtained tree costs at most twice the weight of $\mst(\indu{\emuclo{G}}{S})$.
The quite easy
proof is provided in Section~\ref{ss:swap:swaps}, where we develop a convenient language
for such claims (see also Figure~\ref{fig:emul2} for an illustration).

\begin{figure}[bt]
\centering
\includegraphics[width=.7\linewidth]{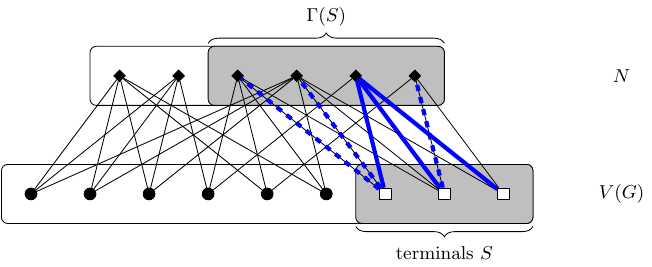}
\caption{A minimum spanning tree $T'$ of $\indu{B}{S \cup \nei{S}}$ (blue thick solid and dashed lines)
  and its subtree $T$ created by removing all leaves that lie
   in $N$ (blue thick solid lines).}
\label{fig:emul2}
\end{figure}

\begin{lemma}\label{lem:emu:cut-leaves}
Let $T' = \mst(\indu{B}{S \cup \nei{S}})$, and let $\mathcal{L}_N$ be the set of leaf vertices of $T'$ contained in $N$.
Let $T$ be a tree obtained from $T'$ by removing leaves in $N$, formally $T=\indu{T'}{V(T') \setminus \mathcal{L}_N}$.
Denote $T_{MST} = \mst(\indu{\emuclo{G}}{S})$.
Then $\dlug_B(T) \leq 2 \faladist_B(T_{MST})$.
\end{lemma}

Due to Lemma~\ref{lem:emu:cut-leaves},
$\mst(\indu{B}{S \cup \nei{S}})$ without the leaves in $N$ costs at most twice the weight of $\mst(\indu{\emuclo{G}}{S})$, which is, by Lemma~\ref{lem:emu-mst-apx}, at most $4\alpha\dlug_G(\st(G,S))$.
Clearly, the cost $\mst(\indu{B}{S \cup \nei{S}})$ without the leaves in $N$ cannot be smaller
then the cost of the minimum Steiner tree in $G$.
Thus, this approach gives us a $4\alpha$-approximation of the Steiner tree in the input graph $G$.

\begin{lemma}
\label{lem:basic_st_algorithm}
Let $G$ be a graph and let $B$ be its bipartite emulator of stretch $\alpha$
(possibly obtained by a randomized algorithm).
Moreover, assume that the (expected) degree in $B$ of any vertex
of $V(G)$ is bounded by $\Delta$.
Let $S$ be a set subject to insertions and deletions, such that at any time $S \subseteq V$.
Then, we can maintain a $4\alpha$-approximate Steiner tree of $G$ that spans $S$, handling each update to $S$ in (expected) amortized
$O(\Delta \log^4 n)$ time.
\end{lemma}

Let us comment on the possible randomization of the construction of a bipartite emulator.
In Section~\ref{ss:emulators:general} we construct a randomized emulator, with only a bound
on the expected degree of a vertex in $V(G)$. Then Lemma~\ref{lem:basic_st_algorithm}
gives us a bound on expected amortized time. That is, for any positive integer $r$,
the expected time spent on the first $r$ operations is bounded by $O(r\Delta \log^4n)$.

In Section~\ref{sec:emulators} we show two constructions of emulators.
The construction for general graphs gives an emulator with $\Delta = O(kn^{1/k})$ (in expectation) and stretch $2k-1$, which yields an algorithm for maintaining $(8k-4)$-approximate Steiner tree in $O(kn^{1/k}\log^4 n)$ expected time per terminal addition or deletion.
In case of planar graphs, we have an emulator with $\Delta = O(\eps^{-1} \log^2 n)$ and stretch $1+\eps$, so we obtain $(4+\eps)$-approximate algorithm, handling updates in $O(\eps^{-1} \log^6 n)$ time.

\subsection{Constructing bipartite emulators}\label{sec:emulators}
In this section we show how to construct emulators that can be plug into Lemma~\ref{lem:basic_st_algorithm} to obtain dynamic algorithms for maintaining the Steiner tree.

\subsubsection{General graphs}\label{ss:emulators:general}
Our emulator for general graphs is based on an approximate distance oracle by Thorup and Zwick~\cite{Thorup05}.

\begin{lemma}[\cite{Thorup05}]\label{lem:emulator_construction}
Let $G = (V, E, d)$ be a graph, $n = |V|$, $m = |E|$, and $k \geq 1$ be an integer.
Then, we can compute an emulator of $G$ of stretch $2k-1$ in $O(kmn^{1/k})$ expected time.
\end{lemma}

We now modify the construction slightly in order to obtain a \emph{bipartite} emulator of $G$.
The emulator constructed in~\cite{Thorup05} is obtained by computing for each $v \in V$ a set $B(v)$ called a \emph{bunch} and, for each $w \in B(v)$, adding an edge $vw$ of length $\dist(v,w)$.
The expected size of each bunch is $O(kn^{1/k})$.
Moreover, from the query algorithm it follows that for any two $u, v \in V$ there exists $w \in B(u) \cap B(v)$ such that $\dist(u, w) + \dist(w, v) \leq (2k-1)\dist(u, v)$.

Thus, the bipartite emulator $G' = (V \cup V', E', d')$ can be constructed as follows.
The vertex set of $G'$ consists of two copies of $V$, i.e., for each $v \in V$, there is its copy $v' \in V'$.
We connect each vertex $v \in V$ to its copy with an edge of length $0$.
Moreover, for each $v \in V$ and each $w \in B(v)$ we add to $G'$ an edge $vw'$ of length $\dist_G(v,w)$.

It follows that $G'$ is a bipartite graph, in which the degree of every $v \in V$ is $O(kn^{1/k})$ and for every two vertices $u, w$ there exists a two-edge path of length $\leq (2k-1)\dist(u, v)$.
Thus, using Lemma~\ref{lem:basic_st_algorithm} we obtain the following result.

\begin{theorem}
Let $G=(V, E, d)$ be a graph, $n = |V|$, $m = |E|$ and $k \geq 1$ be an integer.
Let $S \subseteq V$ be a dynamic set, subject to vertex insertions and removals (initially $S = \emptyset$).
Then, after preprocessing in $O(kmn^{1/k})$ expected time, we may maintain a $(8k-4)$-approximate Steiner tree that spans $S$, handling each update to $S$ in $O(kn^{1/k} \log^4 n)$ expected amortized time.
\end{theorem}

\subsubsection{Planar graphs}
In this section we show an emulator construction for planar graphs.
As a result we obtain an algorithm that maintains a $(4+\eps)$-approximate Steiner tree in polylogarithmic time per update.
In order to reach this goal, we use a construction by Thorup (section 3.8 in~\cite{Thorup04}) that we extend in order to construct an emulator (see Lemma~\ref{lem:set-of-connections}).

Let $G=(V,E)$ be a planar graph.
The overall idea uses recursive division of $G$ using balanced separators.
A \emph{balanced separator} of $G$ is a set $Q \subseteq V$ such that each connected component of $G \setminus Q$ is a constant fraction smaller than $G$.
We find a balanced separator of $G$ that consists of a constant number of shortest paths $P_1, \ldots, P_k$ (the separator consists of vertices contained in these paths).
For a shortest path $P_i$, we build an emulator that approximates all the shortest paths in $G$ that intersect $P_i$.
Then, we recurse on each of the connected components of $G \setminus (P_1 \cup \ldots \cup P_k)$.
Hence, we now focus on the following problem.
Given a planar graph $G$ and a shortest path $P$, build an emulator that approximates all shortest paths intersecting $P$.

We define a \emph{connection} to be an edge that connects a vertex $v \in V$ with a vertex $a \in P$ and has length $\dlug(va)$ at least $\dist_G(v, a)$ (it would be convenient to assume that the length is \emph{equal} to $\dist_G(v, a)$, but the algorithm we use may sometimes give longer connections).
A connection $vb$ \emph{$\eps$-covers} $x$ if $\dlug(vb) + \dist_G(b,x) \leq (1+\eps)\dist_G(v,x)$;
observe that the distance $\dist_G(b,x)$ can be equivalently measured along the path $P$.
A set of connections $C(v, P)$ is \emph{$\eps$-covering} if it $\eps$-covers every $x \in P$.

\begin{lemma}
\label{lem:epsilon-covering}
Let $G=(V, E, d)$ be a planar graph, $n = |V|$, and $0 < \eps \leq 1$.
Let $P$ be a shortest path in $G$.
For each $v \in V(G)$ we can construct an $\eps$-covering set $C(v, P)$ of size $O(\eps^{-1})$ in $O(\eps^{-1}n \log n)$ total time.
\end{lemma}

A very similar fact is shown in~\cite{Thorup04}, but the definition of $\eps$-covering used there is slightly different, so, for the sake of completeness, we rewrite the proof.

\begin{proof}
Let $\eps_0 = \eps/2$.
We say that a connection $vb$ \emph{strongly-$\eps$-covers} $a$ if $\dlug(vb) + (1+\eps)\dist_G(b,a) \leq (1+\eps)\dist_G(v,a)$.
By Lemma 3.18 in~\cite{Thorup04}, for each $v \in V(G)$ we can construct a strongly-$(\eps_0/2)$-covering set $D(v, P)$ of size $O(\eps_0^{-1}\log n)$ in $O(\eps_0^{-1}n \log n)$ time.\footnote{Strictly speaking, the strongly-$\eps$-covering set according to our definition is an $\eps/(\eps+1)$-covering set according to the definition used in the statement of Lemma 3.18 of~\cite{Thorup04}.}
We now show that we can use it to construct an $\eps$-covering set $C(v,P) \subseteq D(v, P)$ of size $O(\eps^{-1})$.

Let $vc$ be the shortest connection from $D(v,P)$ and $s$ be one of the two endpoints of $P$.
We add $vc$ to $C(v, P)$.
Now, iterate through connections in $D(v,P)$ starting from $vc$ and going towards $s$.
Let $vb$ be the connection that was most recently added to $C(v,P)$.
If for the current connection $va$ we have $\dlug(vb) + \dist_G(b, a) > (1+\eps_0)\dlug(va)$, we add $va$ to $C(v,P)$.
Then, we run a similar procedure using the other endpoint of $P$.

To prove that $C(v,P)$ covers every vertex between $c$ and $s$, consider some vertex $x \in P$.
There exists a connection $va \in D(v, P)$ that strongly-$(\eps_0/2)$-covers $vx$, so $\dlug(va) + (1+\eps_0/2)\dist_G(a,x) \leq (1+\eps_0/2)\dist_G(v,x)$.
If this connection is in $C(v,P)$ then $vx$ is strongly-$\eps_0/2$-covered and obviously also $\eps$-covered.
Otherwise, there exists a connection $vb$ such that $\dlug(vb) + \dist_G(b, a) \leq (1+\eps_0)\dlug(va)$.
We have
\begin{align*}
\dlug(vb) + \dist_G(b, x) & \leq \dlug(vb) + \dist_G(b,a) + \dist_G(a,x)\\
        & \leq (1+\eps_0)\dlug(va) + \dist_G(a,x)\\
        & \leq (1+\eps_0)\dlug(va) + (1+\eps_0)(1+\eps_0/2)\dist_G(a,x)\\
        & = (1+\eps_0)(\dlug(va) + (1+\eps_0/2)\dist_G(a,x))\\
        & \leq (1+\eps_0)(1+\eps_0/2) \dist_G(v,x) \\
        & \leq (1+\eps)\dist_G(v,x).
\end{align*}
The last inequality follows from $\eps_0 = \eps/2 \leq 1$.
It remains to bound the size of $C(v, P)$.
Let $f(vb) = \dlug(vb) + \dist_G(b,s)$.
As connections are added to $C(v, P)$ we trace the value of $f(vb)$, where $vb$ is the last connection that we have added.
Every time we add a connection $va$ we reduce the value of $f$ by
$$f(vb) - f(va) = \dlug(vb) + \dist_G(b,s) - \dlug(va) - \dist_G(a,s) = \dlug(vb) - \dlug(va) + \dist_G(b,a) > \eps_0\dlug(va) \geq \eps_0 \dist_G(v, c).$$
However, the total change equals
$$f(vc) - f(vs) = \dlug(vc) + \dist_G(c,s) - \dlug(vs) \leq (1+\eps_0)\dist_G(v,c) + \dist_G(c,s) - \dist_G(v,s) \leq (2+\eps_0) \dist_G(v,c).$$
Thus, at most $O(\eps_0^{-1}) = O(\eps^{-1})$ connections are added to $C(v, P)$.

The same procedure is then repeated for the other endpoint of $P$, so we get a total of $O(\eps^{-1})$ connections.
\end{proof}

\begin{lemma}
\label{lem:set-of-connections}
Let $G=(V, E, d)$ be a planar graph, $n = |V|$, $0 < \eps \leq 1$.
Let $P$ be a shortest path in $G$.
For each $v \in V(G)$ we can construct a set of connections $C'(v, P)$ of size $O(\eps^{-1}\log n)$, which satisfies the following property.
For any two vertices $u, w \in V$, if the shortest path between $u$ and $w$ intersects $P$, then for some $x \in P$ there exist connections $ux \in C'(u, P)$ and $wx \in C'(w,P)$, such that $\dist(u,w) \leq d(ux) + d(wx) \leq (1+\eps)\dist(u,w)$.
The sets $C'(v, P)$ can be constructed in $O(\eps^{-1}n \log n)$ time.
\end{lemma}

\begin{proof}
First, using Lemma~\ref{lem:epsilon-covering}, for every $v \in V$ we construct an $\eps$-covering sets of connections $C(v,P)$.
Assume that the shortest path $Q$ between $u$ and $w$ intersects $P$ in $x \in P$.
There exists a path $Q'$ between $u$ and $w$ which consists of a connection, subpath of $P$, denoted henceforth $Q'_P$, and another connection. Moreover, $d(Q') \leq (1+\eps)d(Q)$.
We call each path of this form an \emph{approximating path}.
Our goal is to substitute every approximating path with an approximating path that consists solely of two connections from $C'(v, P)$.

The construction is done recursively.
The parameter of the recursion is a subpath $P'$ of $P$.
Consider a single step, with a parameter $P' = p_1p_2\ldots p_k$.
Let $p_m = p_{\lfloor k/2 \rfloor}$ be the middle vertex of $P'$.
For any $v \in V$ and $p_i \in P'$, if there is an connection $vp_i \in C(v, P)$, we add a connection $vp_m$ of length $d(vp_i) + \dist(p_i, p_m)$ to $C'(v, P)$.
Then, we recurse on $P_1 = p_1p_2\ldots p_{m-1}$ and $P_2 = p_{m+1}\ldots p_k$.
Lastly, for each $p_i \in P$ we add a connection $p_ip_i$ of length $0$.

To prove the correctness of the procedure, consider now the aforementioned approximating path $Q'$,
and recall that $Q'_P = P \cap Q'$. Let $p$ be the vertex that is taken as $p_m$
in the closest to root node in the recursion tree of the algorithm,
among all vertices of $Q'_P$. 
   Observe that in the single recursive step when $p_m = p$, 
   we add to $C'(v, P)$ the connections $up$ and $wp$ of length exactly equal
 to the length of the 
  part of $Q'$ between $u$ and $p$, and the part of $q'$ between $p$ and $w$, respectively.
Also, the connections we add clearly do not cause any distances to be underestimated.

The running time of each step is proportional to the length of the subpath we consider and the number of connections incident to this subpath.
Moreover, every connection may be considered in at most $O(\log n)$ recursive calls, so we we add to $C'(v, P)$ at most $O(\eps^{-1} \log n)$ connections.
It follows that the total running time of the procedure is $O(\eps^{-1}n \log n)$.
\end{proof}

\begin{lemma}
\label{lem:planar-bipartite}
Let $G=(V, E, d)$ be a planar graph, $n = |V|$, $0 < \eps \leq 1$.
We can construct a bipartite emulator $B = (V \cup N, E_B, \dlug_B)$ of $G$ of stretch $1+\eps$.
The degree of every $v \in V$ in $B$ is $O(\eps^{-1} \log^2 n)$.
The graph $B$ can be constructed in $O(\eps^{-1} n \log^2 n)$.
\end{lemma}

\begin{proof}
We begin with $B$ being a graph with vertex set $V$ and no edges.
The construction is done recursively.
As it is shown, e.g., in~\cite{Thorup04}, 
   each planar graph admits a balanced separator
   that consists of a constant number of shortest paths $P_1, \ldots, P_k$,
   and, moreover, such a separator can be found in $O(n)$ time.
For each path $P_i$ we use Lemma~\ref{lem:set-of-connections} to construct a set of connections $C'(v, P_i)$ for every $v \in V$.
Next, we iterate through the vertices of the paths $P_i$.
For each vertex $w \in P_i$ we add a new auxiliary vertex $w'$ to $B$ and add an edge $uw'$ for each connection $uw$ from $G$ (the length of the edge is the length of the connection).
After that, we recurse on each connected component of $G \setminus (P_1 \cup \ldots \cup P_k)$.

Let us now prove the correctness of the construction.
Consider any two $v_1, v_2 \in V$ and the shortest path $Q$ between them.
At some step of the recursion, some vertex of $Q$ belongs to the separator that we have found.
From the construction, it follows that in this step we have added to $B$ a vertex $w'$ and edges $v_1w'$ and $v_2w'$ to $B$ of total length at most $(1+\eps)\dist_G(v_1, v_2)$.
Moreover, since the length of every connection is lower bounded by the length of the corresponding shortest path, $B$ may not underestimate the distances between vertices of $V$.

Since every vertex $v \in V$ takes part in $O(\log n)$ recursive steps and in every step we add $O(\eps^{-1} \log n)$ edges incident to any $v \in V$, we have that the degree of any vertex of $V$ in $B$ is $O(\eps^{-1}\log^2 n)$.
As shown in~\cite{Thorup04}, finding the separators requires $O(n \log n)$ total time.
The running time of every recursive step is dominated by the time from Lemma~\ref{lem:set-of-connections}.
Summing this over all recursive steps, we get that the construction can be done in $O(\eps^{-1} n \log^2 n)$ time.
\end{proof}

By constructing $B$ according to Lemma~\ref{lem:planar-bipartite} and applying Lemma~\ref{lem:basic_st_algorithm} we obtain the following.

\begin{theorem}
Let $G=(V, E, d)$ be a planar graph and $0 < \eps \leq 1$.
Let $S \subseteq V$ be a dynamic set, subject to vertex insertions and removals (initially $S = \emptyset$).
Then, after preprocessing in $O(\eps^{-1} n \log^2 n)$ time, we may maintain a $(4+\eps)$-approximate Steiner tree that spans $S$, handling each update to $S$ in $O(\eps^{-1} \log^6 n)$ amortized time.
\end{theorem}
%
%

\section{Conclusions}\label{sec:conclusions}

In our work we have given the first sublinear dynamic algorithms for the
Steiner tree problem and the subset TSP problem, where the set of terminals changes over time.
We have exhibited a tight connection to vertex-color distance oracles
and the replacement schemes previously developed for the online model.

We believe our result would inspire further work on the dynamic Steiner
tree problem, focusing not only on the number of replacements
required to maintain a good Steiner tree (as in the online model),
but also on the entire efficiency of updates.
In particular, we would like to emphasize here the following open questions.
\begin{enumerate}
\item The approximation guarantee of our dynamic algorithms, even for planar graphs,
  contains the term $2$, inherited from the fact that we approximate Steiner tree
  by a minimum spanning tree in the metric closure of the input graph.
  We have chosen this approach for its simplicity.
  Can any of the approximation algorithms that guarantees better approximation than 2
  be adapted to our dynamic setting?
  In particular, there exists a PTAS for the (static) Steiner tree problem in planar graphs~\cite{DBLP:journals/talg/BorradaileKM09}.
  Can we obtain a dynamic $(1+\eps)$-approximation algorithm for planar graphs? 
  The other direction to approach this question is to devise simple and almost linear time
  approximation algorithm for Steiner tree that would beat the factor of 2 in non-planar graphs. This could 
  be of some practical importance in big data applications.

  \item A second question concerning better approximation ratios
would be whether it is possible to obtain a $c$-approximate vertex-color distance
oracle in general graphs for some constant $c < 3$.
\item In the deletion step, we relied on a somehow weak variant of a \emph{color-to-color}
  distance query, where the returned distance takes into account only portal distances
  between the colors.
  Is it possible to enhance our oracles to perform a general color-to-color queries efficiently?
\item Our decremental and fully dynamic algorithms for general graphs
require preprocessing time of roughly order $O(\sqrt{n}m)$ and space of roughly order $O(n\sqrt{n})$.
Can it be improved to near-linear time and space?
\item The running time of our fully dynamic algorithms depends on the stretch of
the metric. Can this dependency be avoided, or at least further reduced?
\item Which other combinatorial problems, e.g., facility location, allow for dynamic sublinear time approximation algorithms? 
\end{enumerate}

\bibliographystyle{abbrv}
\bibliography{references}

\end{document}